%% file: lmcs.tex
\pdfoutput=1
\documentclass{lmcs} 
\pdfoutput=1

\usepackage{lastpage}
\lmcsdoi{20}{2}{7}
\lmcsheading{}{\pageref{LastPage}}{}{}%
{Jul.~20,~2023}{May~14,~2024}{}


\input{settings}
\input{macros}

\begin{document}

\title[Exploring Non-Regular Extensions of PDL with DL Features]{Exploring Non-Regular Extensions of Propositional Dynamic Logic with Description-Logics Features}

\author[B.~Bednarczyk]{Bartosz Bednarczyk\lmcsorcid{0000-0002-8267-7554}}  
\address{Computational Logic Group, TU Dresden \& Institute of Computer Science, University of Wroc\l{}aw} 
\email{bartosz.bednarczyk@cs.uni.wroc.pl}  

\input{sections/abstract}
\maketitle
\input{sections/introduction}
\input{sections/preliminaries}

\input{sections/self}

\input{sections/nominals}

\input{sections/querying}

\input{sections/conclusions}

\section*{Acknowledgements}
This work was supported by the ERC Consolidator Grant No.~771779 (\href{https://iccl.inf.tu-dresden.de/web/DeciGUT/en}{DeciGUT}).
Snake~and cobra icons were downloaded from \href{https://icons8.com/icon/9cOrIHyR3rRE/snake}{Icons8} and \href{https://www.flaticon.com/free-icons/cobra}{Flaticon}.

\begin{figure}[h]
    \centering
    \includegraphics[scale=0.05]{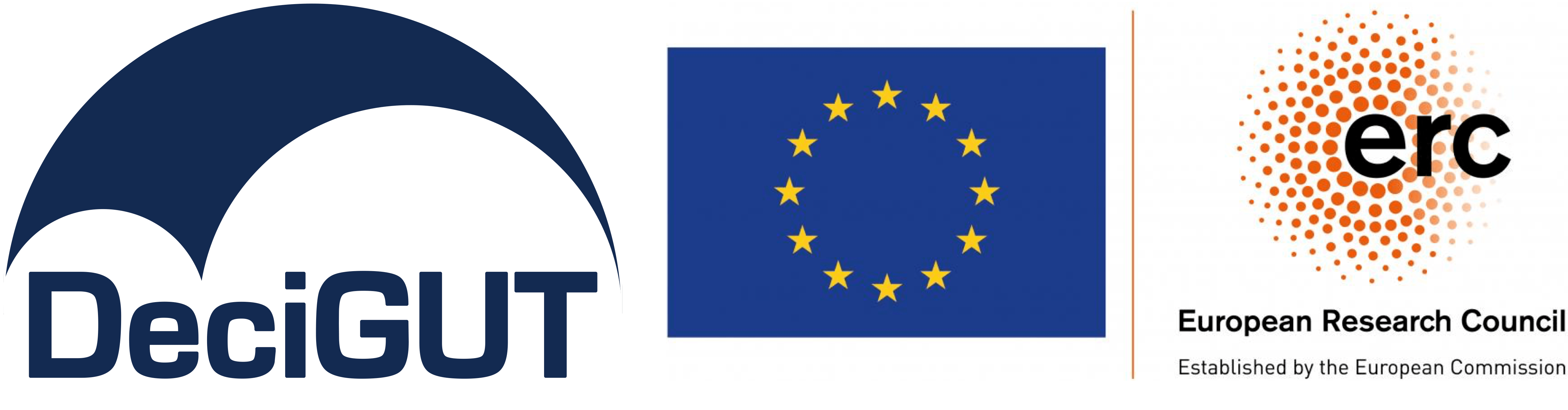}%
\end{figure}

\noindent This work would not be published without the support of Reijo Jaakkola, Sebastian Rudoph, and Witold Charatonik.
I thank Reijo, with whom I initially started this line of research, for all these hours spent in front of a whiteboard in Tampere.
I thank Witek for very careful proofreading of a conference version of this paper, constant support, and for his pedantic approach to writing.
Finally, I thank Sebastian who did truly amazing job as my PhD supervisor. Among all the other things, I would like to thank him for proofreading my drafts, pointing out the work of Calvanese, De Giacomo and Rosati~\cite{CalvaneseGR98}, and for suggesting simplifications of undecidability proofs from the conference version of Section~\ref{sec:querying-negative}.

Last but certainly not least, I extend my gratitude to the anonymous LMCS reviewers for their invaluable feedback on the submitted paper. Frankly speaking, I cannot recall ever encountering reviews as detailed as these before.

\bibliographystyle{alphaurl}
\bibliography{references}

\end{document}

%% file: settings.tex
\usepackage[english]{babel}
\usepackage[latin1]{inputenc}
\usepackage{hyperref}
\usepackage{subfigure}
\usepackage{fontawesome5}
\usepackage{etex}

\usepackage{epsfig}
\usepackage{latexsym}
\usepackage{graphicx}
\usepackage{amssymb}
\usepackage{amsfonts}
\usepackage{amsthm}
\usepackage{amsmath}
\usepackage{verbatim}
\usepackage{url}
\usepackage{colortbl}
\usepackage{stmaryrd}
\usepackage{epic,eepic}
\usepackage{xspace}
\usepackage{tikz}
\usepackage[noend]{algpseudocode}
\usepackage{tcolorbox}
\usepackage{multirow}

\usetikzlibrary{automata}
\usetikzlibrary{arrows}
\usetikzlibrary{backgrounds}
\usetikzlibrary{positioning}
\usetikzlibrary{fit}
\usetikzlibrary{calc}
\usetikzlibrary{matrix}
\usetikzlibrary{patterns}

\usepackage{stackengine}
\usepackage{xspace}

\usepackage{xcolor}
\usepackage[linesnumbered,ruled,vlined]{algorithm2e}

\SetCommentSty{mycommfont}
\SetKwInput{KwData}{Input}

\usepackage{microtype} 
\usepackage{ellipsis}         
\usepackage[abbreviations,british]{foreign} 

\usetikzlibrary{shadows}
\tikzset{
diagonal fill/.style 2 args={fill=#2, path picture={
\fill[#1, sharp corners] (path picture bounding box.south west) -|
                         (path picture bounding box.north east) -- cycle;}},
reversed diagonal fill/.style 2 args={fill=#2, path picture={
\fill[#1, sharp corners] (path picture bounding box.north west) |- 
                         (path picture bounding box.south east) -- cycle;}}
}

 \makeatletter
 \def\desclabel#1#2{\begingroup
    \def\@currentlabel{#1}%
    #1\label{#2}\endgroup
 }
 \makeatother

\input{graphics/nicolas-markey-expose.sty}

\DeclareMathAlphabet{\mathdutchcal}{U}{dutchcal}{m}{n}
\SetMathAlphabet{\mathdutchcal}{bold}{U}{dutchcal}{b}{n}
\DeclareMathAlphabet{\mathdutchbcal}{U}{dutchcal}{b}{n}

\usepackage{todonotes}
\definecolor{ao(english)}{rgb}{0.0, 0.5, 0.0}

\let\captioncaption\caption
\usepackage{hyperref}
\let\caption\captioncaption

\usetikzlibrary{shadows}
\tikzset{
diagonal fill/.style 2 args={fill=#2, path picture={
\fill[#1, sharp corners] (path picture bounding box.south west) -|
                         (path picture bounding box.north east) -- cycle;}},
reversed diagonal fill/.style 2 args={fill=#2, path picture={
\fill[#1, sharp corners] (path picture bounding box.north west) |- 
                         (path picture bounding box.south east) -- cycle;}}
}

\usetikzlibrary{hobby,backgrounds,calc,trees}
\usetikzlibrary{decorations.pathreplacing}

%% file: macros.tex
\newcommand{\len}[1]{\left\lVert #1 \right\rVert}

\newrobustcmd{\newnotion}[1]{\emph{#1}}
\newcommand{\deff}{\coloneqq}

\newcommand{\dland}{\sqcap}
\newcommand{\dlor}{\sqcup}
\newcommand{\topconcept}{\top}
\newcommand{\botconcept}{\bot}
\newcommand{\dlsubseteq}{\sqsubseteq}
\newcommand{\Self}{\mathsf{Self}}
\newcommand{\bigdlor}{\bigsqcup}

\makeatletter
\providecommand{\bigdland}{%
  \mathop{%
    \mathpalette\@updown\bigsqcup
  }%
}
\newcommand*{\@updown}[2]{%
  \rotatebox[origin=c]{180}{$\m@th#1#2$}%
}
\makeatother

\makeatletter
\let\amsmath@bigm\bigm

\renewcommand{\bigm}[1]{%
  \ifcsname fenced@\string#1\endcsname
    \expandafter\@firstoftwo
  \else
    \expandafter\@secondoftwo
  \fi
  {\expandafter\amsmath@bigm\csname fenced@\string#1\endcsname}%
  {\amsmath@bigm#1}%
}

\newcommand{\DeclareFence}[2]{\@namedef{fenced@\string#1}{#2}}
\makeatother

\DeclareFence{\mid}{|}

\newcommand{\DL}[1]{\ensuremath{\mathcal{#1}}}  

\newcommand{\CQ}{\textrm{CQ}}   
   
\newcommand{\PEQ}{\textrm{PEQ}}

\newcommand{\VPQ}{\textrm{VPQ}}   
\newcommand{\TwoVPQ}{\textrm{2VPQ}}

\newcommand{\ALC}{\DL{ALC}}                     
\newcommand{\ALCO}{\DL{ALCO}}                   
\newcommand{\ALCSelf}{\DL{ALC}^{\textsf{Self}}} 
\newcommand{\ALCreg}{\DL{ALC}_{\textsf{reg}}} 
 
\newcommand{\ALCOreg}{\DL{ALCO}_{\textsf{reg}}}

\newcommand{\Z}{\DL{Z}}                         

\newcommand{\abstrDL}{\DL{DL}}     

\newcommand{\complexityclass}[1]{\textsc{#1}} 
\newcommand{\TwoExpTime}{\complexityclass{2ExpTime}} 
\newcommand{\lang}[1]{\mathbf{#1}}  
\newcommand{\Ilang}{\lang{N_I}}     
\newcommand{\Rlang}{\lang{N_R}}     
\newcommand{\Clang}{\lang{N_C}}     
\newcommand{\Vlang}{\lang{N_V}}     

\newcommand{\query}[1]{\mathit{#1}}  
\newcommand{\queryq}{\query{q}}      
\newcommand{\match}[1]{#1}          
\newcommand{\matcheta}{\match{\eta}}  

\newcommand{\var}[1]{\mathit{#1}}   
\newcommand{\varx}{\var{x}}         
\newcommand{\vary}{\var{y}}         
\newcommand{\varz}{\var{z}}         
\newcommand{\varv}{\var{v}}         
\newcommand{\varu}{\var{u}}         
\newcommand{\varw}{\var{w}}         

\newcommand{\role}[1]{\mathit{#1}}      
\newcommand{\roler}{\role{r}}           
\newcommand{\roles}{\role{s}}           
\newcommand{\rolet}{\role{t}}           

\newcommand{\concepts}{\lang{C}}            
\newcommand{\alcconcepts}{\concepts} 
\newcommand{\alccapconcepts}{\concepts} 

\newcommand{\concept}[1]{\mathrm{#1}}       
\newcommand{\conceptA}{\concept{A}}         
\newcommand{\conceptC}{\concept{C}}         
\newcommand{\conceptD}{\concept{D}}         

\newcommand{\indv}[1]{\textls[-85]{\texttt{#1}}}  
\newcommand{\indva}{\indv{a}}       
\newcommand{\indvb}{\indv{b}}       

\newcommand{\kb}[1]{\mathcal{#1}}   
\newcommand{\kbK}{\kb{K}}           
\newcommand{\abox}[1]{\mathcal{#1}} 
\newcommand{\aboxA}{\abox{A}}       
\newcommand{\tbox}[1]{\mathcal{#1}} 
\newcommand{\tboxT}{\tbox{T}}       
\newcommand{\ind}[1]{\mathsf{ind}{(#1)}} 
\newcommand{\indA}{\ind{\aboxA}}    
\newcommand{\indK}{\ind{\kbK}}    
\newcommand{\names}{\mathsf{N}}     

\newcommand{\inter}[1]{\mathcal{#1}}    
\newcommand{\interI}{\inter{I}}         
\newcommand{\interJ}{\inter{J}}         

\newcommand{\DeltaInter}[1]{\Delta^{#1}}                
\newcommand{\DeltaI}{\DeltaInter{\interI}}              
\newcommand{\DeltaJ}{\DeltaInter{\interJ}}              

\newcommand{\cdotInter}[1]{\cdot^{#1}}          
\newcommand{\cdotI}{\cdotInter{\interI}}        

\newcommand{\domelem}[1]{\mathrm{#1}}                           
\newcommand{\domelemd}{\domelem{d}}                             
\newcommand{\domeleme}{\domelem{e}}                             

\newcommand{\N}{\mathbb{N}}
\newcommand{\set}[1]{\{#1\}}

\newcommand{\ZZ}{\mathbb{Z}}

\newcommand{\pathrho}{\rho}

\newcommand{\interfwdunrav}[3]{{#1}^{#2}_{#3}}  
\newcommand{\Deltafwdunrav}[3]{\DeltaInter{\interfwdunrav{#1}{#2}{#3}}}  
\newcommand{\cdotfwdunrav}[3]{\cdotInter{\interfwdunrav{#1}{#2}{#3}}}    

\newrobustcmd{\interomegafwdunravI}[1]{\interfwdunrav{\interI}{\protect\vv{\omega}}{#1}}   
\newrobustcmd{\interomegafwdunravInames}{\interomegafwdunravI{\names}}  
\newrobustcmd{\interomegafwdunravIindA}{\interomegafwdunravI{\indA}}
\newrobustcmd{\interomegafwdunravIindK}{\interomegafwdunravI{\indK}}

\newrobustcmd{\DeltaomegafwdunravI}[1]{\Deltafwdunrav{\interI}{\protect\vv{\omega}}{#1}}   
\newrobustcmd{\DeltaomegafwdunravInames}{\DeltaomegafwdunravI{\names}}  
\newrobustcmd{\DeltaomegafwdunravIindA}{\DeltaomegafwdunravI{\indA}}

\newrobustcmd{\cdotomegafwdunravI}[1]{\cdotfwdunrav{\interI}{\protect\vv{\omega}}{#1}}     
\newrobustcmd{\cdotomegafwdunravInames}{\cdotomegafwdunravI{\names}}  
\newrobustcmd{\cdotomegafwdunravIindA}{\cdotomegafwdunravI{\indA}}

\newrobustcmd{\internfwdunravI}[1]{\interfwdunrav{\interI}{\protect\vv{n}}{#1}}   
\newrobustcmd{\internfwdunravInames}{\internfwdunravI{\names}}
\newrobustcmd{\internfwdunravIindK}{\internfwdunravI{\indK}}

\newrobustcmd{\internfwdunravJ}[1]{\interfwdunrav{\interJ}{\protect\vv{n}}{#1}}   
\newrobustcmd{\internfwdunravJnames}{\internfwdunravJ{\names}}
\newrobustcmd{\internfwdunravJindK}{\internfwdunravJ{\indK}}

\newrobustcmd{\DeltanfwdunravI}[1]{\Deltafwdunrav{\interI}{\protect\vv{n}}{#1}} 
\newrobustcmd{\DeltanfwdunravInames}{\DeltanfwdunravI{\names}}
\newrobustcmd{\DeltanfwdunravIindK}{\DeltanfwdunravI{\indK}}

\newrobustcmd{\cdotnfwdunravI}[1]{\cdotfwdunrav{\interI}{\protect\vv{n}}{#1}} 
\newrobustcmd{\cdotnfwdunravInames}{\cdotnfwdunravI{\names}}
\newrobustcmd{\cdotnfwdunravIindK}{\cdotnfwdunravI{\indK}}

\newcommand{\statesQ}{\mathtt{Q}}

\newcommand{\stateq}{\mathtt{q}}

\newcommand{\transreldelta}{\mathtt{T}} 

\newcommand{\lettera}{\mathtt{a}}
\newcommand{\letterb}{\mathtt{b}}
\newcommand{\letterc}{\mathtt{c}}
\newcommand{\letterx}{\mathtt{x}}
\newcommand{\tapeword}[1]{\mathtt{#1}}
\newcommand{\tapewordw}{\tapeword{w}}
\newcommand{\tapewordu}{\tapeword{u}}

\newcommand{\automatonA}{\mathdutchbcal{A}}

\newcommand{\automatonC}{\mathdutchbcal{C}}
\newcommand{\alphabet}{\Sigma}
\newcommand{\alphabetall}{\Sigma_{\mathsf{all}}}
\newcommand{\alphabetvpl}{\Sigma_{\mathsf{vpl}}}
\newcommand{\stackalphabet}{\Gamma}
\newcommand{\stacksigma}{\sigma}
\newcommand{\finalstates}{\mathtt{F}}
\newcommand{\initialstates}{\mathtt{I}}
\newcommand{\botofstack}{\lhd}
\newcommand{\languageL}{\mathdutchbcal{L}}
\newcommand{\langvpaeq}[2]{{#1}^{\#}\hspace{-0.2em}{#2}^{\#}}

\newcommand{\langvpamore}[2]{#1^{\#} #2^{>\#}}
\newcommand{\ALCvpl}{\ALC_{\textsf{vpl}}}

\newcommand{\ALCOvpl}{\ALCO_{\textsf{vpl}}}

\newcommand{\ALCall}{\ALC_{\textsf{all}}}

\newcommand{\proj}{\pi}

\newcommand{\REG}{\mathbb{REG}}
\newcommand{\ALL}{\mathbb{ALL}}
\newcommand{\VPL}{\mathbb{VPL}}

\newcommand{\langclassC}{\mathbb{C}} 

\newcommand{\ALCregrhashshash}{\DL{ALC}_{\textsf{reg}}^{\roler\#\roles\#}}
\newcommand{\ALCOregrhashshash}{\DL{ALCO}_{\textsf{reg}}^{\roler\#\roles\#}}

\newcommand{\ALCSelfvpl}{\DL{ALC}^{\textsf{Self}}_{\mathsf{vpl}}} 

\newcommand{\wang}[6]{
\draw [black,fill=#3] (#1,#2)--(#1+0.5,#2+0.5)--(#1,#2+1)--cycle;
\draw [black,fill=#4] (#1,#2)--(#1+0.5,#2+0.5)--(#1+1,#2)--cycle;
\draw [black,fill=#5] (#1+1,#2+1)--(#1+0.5,#2+0.5)--(#1+1,#2)--cycle;
\draw [black,fill=#6] (#1+1,#2+1)--(#1+0.5,#2+0.5)--(#1,#2+1)--cycle;
}

\newcommand{\smallwang}[6]{
\draw [black,fill=#3] (#1,#2)--(#1+0.125,#2+0.125)--(#1,#2+0.25)--cycle;
\draw [black,fill=#4] (#1,#2)--(#1+0.125,#2+0.125)--(#1+0.25,#2)--cycle;
\draw [black,fill=#5] (#1+0.25,#2+0.25)--(#1+0.125,#2+0.125)--(#1+0.25,#2)--cycle;
\draw [black,fill=#6] (#1+0.25,#2+0.25)--(#1+0.125,#2+0.125)--(#1,#2+0.25)--cycle;
}

\newcommand{\cyanBox}{\protect\tikz \protect\node [rectangle,draw, fill = cyan] at (2.5,-2) {};}
\newcommand{\rougeBox}{\protect\tikz \protect\node [rectangle,draw, fill = rouge] at (2.5,-2) {};}
\newcommand{\whiteBox}{\protect\tikz \protect\node [rectangle,draw, fill = white] at (2.5,-2) {};}
\newcommand{\limeBox}{\protect\tikz \protect\node [rectangle,draw, fill = lime] at (2.5,-2) {};}
\newcommand{\intextwang}[4]{\protect\tikz \protect\smallwang{0}{0}{#1}{#2}{#3}{#4};}

\newcommand{\rmT}{\mathrm{T}}
\newcommand{\namesrulerT}{\names_{\text{\faRuler}}^{\rmT}}
\newcommand{\rmM}{\mathrm{M}}
\newcommand{\rmN}{\mathrm{N}}
\newcommand{\rmH}{\mathrm{H}}
\newcommand{\rmV}{\mathrm{V}}
\newcommand{\indvyst}{\indv{st}}
\newcommand{\indvmd}{\indv{md}}
\newcommand{\indvyend}{\indv{end}}
\newcommand{\conceptyardstick}{\conceptC_{\text{\faRuler}}}
\newcommand{\tilingsys}{\mathdutchcal{D}}
\newcommand{\tiles}{\mathrm{T}}
\newcommand{\tilesCol}{\mathrm{Col}}
\newcommand{\tileswhite}{\whiteBox}%
\newcommand{\tile}{\mathrm{t}}
\newcommand{\colour}{\mathrm{c}}
\newcommand{\tilesmap}{\xi}
\newcommand{\indvsld}{\indv{ld}}
\newcommand{\indvsrd}{\indv{rd}}
\newcommand{\indvsru}{\indv{ru}}
\newcommand{\indvslu}{\indv{lu}}
\newcommand{\namessnakeT}{\names_{\text{\faPuzzlePiece}}^{\rmT}}
\newcommand{\conceptsnake}{\conceptC_{\includegraphics[scale=0.25]{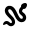}}}%
\newcommand{\concepttilepath}{\conceptC_{\text{\faPuzzlePiece}}}
\newcommand{\indvcbra}{\indv{cbra}}
\newcommand{\conceptcobra}{\conceptC_{\includegraphics[scale=0.15]{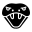}}}%
\newcommand{\octant}{\mathbb{O}}

%% file: sections/abstract.tex

\begin{abstract}\label{sec:abstract}
We investigate the impact of non-regular path expressions on the decidability of satisfiability checking and querying in description logics extending $\ALC$.
Our primary objects of interest are $\ALCreg$ and $\ALCvpl$, the extensions of $\ALC$ with path expressions employing, respectively, regular and visibly-pushdown languages.
The first one, $\ALCreg$, is a notational variant of the well-known Propositional Dynamic Logic of Fischer and Ladner.
The second one, $\ALCvpl$, was introduced and investigated by L\"oding and Serre in 2007. 
The~logic $\ALCvpl$~generalises many known decidable non-regular extensions of $\ALCreg$.

We provide a series of undecidability results. 
First, we show that decidability of the concept satisfiability problem for $\ALCvpl$ is lost upon adding the seemingly innocent $\Self$ operator.
Second, we establish undecidability for the concept satisfiability problem for $\ALCvpl$ extended with nominals.
Interestingly, our undecidability proof relies only on one single non-regular (visibly-pushdown) language, namely on $\roler^\#\roles^\# \deff \{ \roler^n \roles^n \mid n \in \N \}$ for fixed role names $\roler$ and $\roles$.
Finally, in contrast to the classical database setting, we establish undecidability of query entailment for queries involving non-regular atoms from $\roler^\#\roles^\#$, already in the case of  $\ALC$-TBoxes.
\end{abstract}

%% file: sections/introduction.tex

\section{Introduction}\label{sec:intro}

Formal ontologies play a crucial role in artificial intelligence, serving as the backbone of various applications such as the Semantic Web, ontology-based information integration, and peer-to-peer data management.
In reasoning about graph-structured data, a significant role is played by \emph{description logics} (DLs)~\cite{dlbook}, a robust family of logical formalisms serving as the logical foundation of contemporary standardised ontology languages, including OWL~2 by the W3C~\cite{grau2008owl,OWL2Primer}.
Among many features present in extensions of the basic description logic $\ALC$, an especially useful one is $\cdot_{\mathsf{reg}}$, supported by the popular $\Z$-family of description logics~\cite{CalvaneseEO09}.
With $\cdot_{\mathsf{reg}}$ one can specify regular path constraints, allowing the user to navigate graph-structured data.
In recent years, many extensions of $\ALCreg$ for ontology-engineering were proposed~\cite{BienvenuCOS14,CalvaneseOS16,ortiz2010query}, and the complexity landscape of their reasoning problems is now mostly well-understood~\cite{CalvaneseEO09,BednarczykR19,BednarczykK22}.
In fact, the logic $\ALCreg$ was already studied in 1979 by the formal-verification community~\cite{FischerLadner1979}, under the name of Propositional Dynamic~Logic~(PDL). 
Relationship between (extensions of) PDL and $\ALCreg$ were investigated by De Giacomo and~Lenzerini~\cite{GiacomoL94}. 

The spectrum of recognizable word languages is relatively wide. Hence the question of whether regular constraints in path expressions of $\ALCreg$ can be lifted to more expressive classes of languages received a lot of attention from researchers.
We call such extensions of PDL \emph{non-regular}.
After the first undecidability proof of satisfiability of $\ALCreg$ by context-free languages~\cite[Cor.~2.2]{harel1981propositional}, several decidable cases were identified. %
For instance, Koren and Pnueli~\cite[Sec. Decidability]{TmimaP1983} proved that $\ALCreg$ extended with the simplest non-regular language $\langvpaeq{\roler}{\roles} \deff \{ \roler^n \roles^n \mid n \in \N \}$ for \emph{fixed} roles $\roler, \roles$ is decidable; while combining it with~$\langvpaeq{\roles}{\roler}$ leads to undecidability~\cite[Thm.~3.2]{HarelPS81}.
This surprises at first glance, but as it was shown later~\cite[Thm.~18]{LodingLS07}, PDL extended with a broad class of input-driven context-free languages, called \emph{visibly pushdown languages}~\cite[Sec.~5]{alur2009}, remains decidable.
This generalises all previously known decidability results, and partially explains the reason behind known failures (\eg the languages $\langvpaeq{\roler}{\roles}$ and $\langvpaeq{\roles}{\roler}$ cannot be both visibly-pushdown under the same partition of the~alphabet).
Three years ago, the decidability boundary was pushed even further~\cite[Ex.~1]{BruseL21}, by allowing for mixing modalities in visibly-pushdown expressions (for instance, allowing the user to specify that ``for all positive integers $n \in \N$, all $\rolet$-successors of $\roler^n$-reachable elements can $\roles^n$-reach an element fulfilling $\varphi$'').

\paragraph*{Our motivation and our contribution.}
Despite the presence of a plethora of various results concerning non-regular extensions of PDL~\cite{TmimaP1983,HarelPaterson1984,HarelS96,HarelRaz93,BruseL21}, to the best of our knowledge the extensions of non-regular PDL with popular features supported by W3C ontology languages are yet to be investigated.
Such extensions include, among others, \emph{nominals} (constants), \emph{inverse roles} (inverse programs), \emph{functionality} or \emph{counting} (deterministic programs or graded modalities), and the $\Self$ operator~(self-loops).
The honourable exception is the unpublished undecidability result for $\ALCreg$ extended with the language $\{ \roler^{n}\roles(\roler^-)^n \mid n \in \N \}$, where~$\roler^-$ denotes the converse of $\roler$, from G\"oller's thesis~\cite{Goller2008} (answering an open problem of Demri~\cite[Probl\'eme ouvert 29]{Demri07}). 
The~lack of results on entailment of non-regular queries over ontologies is also intriguing, taking into account positive results for conjunctive visibly-pushdown queries in the setting of relational-databases~\cite[Thm.~2]{LangeL15}.

In this paper we contribute to the further understanding of the aforementioned~questions. 
Our results are negative.   
Section~\ref{sec:self} establishes undecidability of the concept satisfiability of~$\ALCvpl$ extended with the seemingly innocent $\Self$ operator.
Section~\ref{sec:nominals} establishes undecidability of the concept satisfiability of $\ALCvpl$ extended with nominals. 
More specifically, the undecidability arises already if the only non-regular language present in concepts is~$\roler^{\#}\roles^{\#}$.
Finally, Section~\ref{sec:querying-negative} establishes undecidability of the query entailment problem over $\ALC$-TBoxes, in which our queries can employ atoms involving the language $\langvpaeq{\roler}{\roles}$. 
On the positive side, a combination of existing works shows that the entailment problem of $\REG$-{\PEQ}s over $\ALCvpl$-TBoxes is $\TwoExpTime$-complete.
We conclude with an open problem list.\\

This work is a revised and extended version of our JELIA 2023 paper~\cite{Bednarczyk23}. In~addition to the inclusion of full proofs and minor corrections, the paper contains improved and new undecidability results concerning entailment of non-regular queries (Section~\ref{sec:querying-negative}).

%% file: sections/preliminaries.tex

\section{Preliminaries}\label{sec:preliminaries}

We assume familiarity with basics on the description logic $\ALC$~\cite[Sec. 2.1--2.3]{dlbook}, regular and context-free languages, Turing machines and computability~\cite[Sec. 1--5]{sipser13}.
We employ $\N$ to denote the set of \emph{non-negative} integers and $\ZZ_n$ to denote the set $\{ 0, 1, \ldots, n{-}1\}$.

\paragraph{Basics on $\ALC$.} 
We recall ingredients of the well-known description logic $\ALC$~\cite[Sec. 2.1--2.3]{dlbook}.\footnote{To guide people with modal logic background: concept names are propositional letters, interpretations are Kripke structures, ABoxes (a.k.a. Assertion Boxes) are partial descriptions of Kripke structures, and TBoxes (a.k.a. Terminological  Boxes) are limited forms of the universal modality. Consult Section~2.6.2 of the DL Book~\cite{dlbook} for a more detailed comparison.}
We fix countably infinite pairwise disjoint sets of \emph{individual names} \(\Ilang \), \emph{concept names} \(\Clang \), and \emph{role names} \( \Rlang \) and introduce the description logic \( \ALC \).
Starting from \( \Clang \) and~\( \Rlang \), the set \( \alcconcepts \) of \( \ALC \)-\emph{concepts} is built using the following concept constructors: \emph{negation} \((\neg \conceptC) \), \emph{conjunction} \((\conceptC \dland \conceptD) \), \emph{existential restriction} (\(\exists{\roler}.\conceptC \)), and the \emph{top concept} $\topconcept$ with the grammar:
\begin{equation*} \label{eq:alc-grammar}
\conceptC, \conceptD \; \Coloneqq \topconcept \; \mid \; \conceptA \; \mid \; \neg \conceptC \; \mid \; \conceptC \dland \conceptD \; \mid \; \exists{\roler}.\conceptC,
\end{equation*}
where \(\conceptC,\conceptD \in \alcconcepts \), \(\conceptA \in \Clang \) and \(\roler \in \Rlang \). 
We extend the syntax by notational shortcuts for disjunction \(\conceptC \dlor \conceptD \deff \neg (\neg \conceptC  \dland \neg \conceptD) \), universal restrictions \(\forall{\roler}.\conceptC \deff \neg \exists{\roler}.\neg\conceptC \), the bottom concept \(\botconcept \deff \neg \topconcept \), implication~$\conceptC \to \conceptD \deff \neg \conceptC \dlor \conceptD$, and equivalence~$\conceptC \leftrightarrow \conceptD \deff (\conceptC \to \conceptD) \dland (\conceptD \to \conceptC)$.
  \emph{Assertions} are (possibly negated) expressions of the form \(\conceptC(\indva) \) or \(\roler(\indva,\indvb) \) for \(\indva,\indvb \in \Ilang \), \(\conceptC \in \alccapconcepts \) and \(\roler \in \Rlang \).
  A \emph{general concept inclusion} (GCI) has the form \(\conceptC \dlsubseteq \conceptD \) for concepts \(\conceptC, \conceptD \in \alccapconcepts \). 
  A~\emph{knowledge base} (KB) \(\kbK:=(\aboxA, \tboxT) \) is composed of a finite set \(\aboxA \) (\emph{ABox}) of assertions and a finite set \(\tboxT \) (\emph{TBox}) of GCIs. 
  We call the elements of \(\aboxA \cup \tboxT \)~\emph{axioms}.

The semantics of \(\ALC \) is defined via \emph{interpretations}~\(\interI \deff (\DeltaI, \cdotI) \) composed of a non-empty set \(\DeltaI \), called the \emph{domain of \(\interI \)} and an \emph{interpretation function} \(\cdotI \), mapping individual names to elements of \(\DeltaI \), concept names to subsets of \(\DeltaI \), and role names to subsets of \(\DeltaI \times \DeltaI \). 
This mapping is then extended to concepts (see Table~\ref{tab:ALC}), and finally used to define \emph{satisfaction} of assertions and GCIs. 
More precisely $\interI \models \conceptC \dlsubseteq \conceptD$ if $\conceptC^{\interI} \subseteq \conceptD^{\interI}$, $\interI \models \conceptC(\indva)$ if $\indva^{\interI} \in \conceptC^{\interI}$, and $\interI \models \roler(\indva, \indvb)$ if $(\indva^{\interI},\indvb^{\interI}) \in \roler^{\interI}$.
We say that an interpretation \(\interI \) \emph{satisfies} a concept~\( \conceptC \) (or \(\interI \) is a \emph{model} of \(\conceptC \), written:~\(\interI \models \conceptC \)) if $\conceptC^{\interI} \neq \emptyset$.
A concept is \emph{consistent} (or \emph{satisfiable}) if it has a model, and \emph{inconsistent} (or \emph{unsatisfiable}) otherwise. 
In the \emph{consistency} (or \emph{satisfiability}) \emph{problem} we ask, whether an input concept has a model.
Similarly, we say that $\interI$ satisfies a KB $\kbK \deff (\aboxA, \tboxT)$ (written: $\interI \models \kbK$) if $\interI$ satisfies all axioms of~$\aboxA \cup \tboxT$. 
The~knowledge base satisfiability problem is defined analogously. 

\begin{table}[h]
    \caption{Concepts and roles in \(\ALC \).\label{tab:ALC}}
    \centering
        \begin{tabular}{l@{\ \ \ }c@{\ \ \ }l@{}}
            \hline\\[-2ex]
            Name & Syntax & Semantics \\ \hline \\[-2ex]
            top concept & \( \topconcept \) & \( \DeltaI  \) \\
            concept name & \(\conceptA \) & \(\conceptA^\interI \subseteq \DeltaI  \) \\ 
            role name & \(\roler \) & \(\roler^\interI \subseteq \DeltaI {\times} \DeltaI \) \\ 
            concept negation & \(\neg\conceptC \)& \(\DeltaI \setminus \conceptC^{\interI} \) \\  
            concept intersection & \(\conceptC \dland \conceptD \)& \(\conceptC^{\interI}\cap \conceptD^{\interI} \) \\  
            existential restriction & \(\exists{\roler}.\conceptC \) & 
            \(\set{ \domelemd \mid \exists{\domeleme \in \conceptC^{\interI}} \ (\domelemd,\domeleme)\in \roler^{\interI}} \)
            \\\hline
        \end{tabular}
\end{table}

A path in an interpretation $\interI$ is a finite word in $(\DeltaI)^+$. Given a path $\pathrho$ we enumerate its components with $\pathrho_1, \ldots, \pathrho_{n}$. 
We use $\len{\pathrho}$ to denote the length of $\pathrho$ (note that $\len{\pathrho} = |\pathrho|{-}1$).
We say that $\pathrho$ \emph{starts from} (resp. \emph{ends in}) $\domelemd$ if $\pathrho_1 = \domelemd$ holds (resp. $\pathrho_{|\pathrho|} = \domelemd$).
If $\names \subseteq \Ilang$ is a set of individual names, we call an element $\domelemd \in \DeltaI$ \emph{$\names$-named} if $\domelemd = \indva^{\interI}$ holds for some~$\indva \in \names$.

\paragraph{$\ALC$ with extras.}
In the next sections we consider several popular description-logics features: \emph{nominals}~$(\mathcal{O})$, and the \emph{$\Self$ operator}  $(\cdot^{\Self})$. Their semantics is recalled in Table~\ref{tab:ALC-extra-features}, in which $\roler \in \Rlang$, $\ell \in \N$, $\indva \in \Ilang$, and $\conceptC$ is a concept.
\begin{table}[!htb]
    \caption{Popular description-logics features.\label{tab:ALC-extra-features}}
    \centering
        \begin{tabular}{l@{\ \ \ }c@{\ \ \ }l@{}}
            \hline\\[-2ex]
            Name & Syntax & Semantics \\ \hline \\[-2ex]
            nominal & $\{ \indva \}$  & $\{ \indva^{\interI} \}$\\
            self-operator & $\exists{\roler}.\Self$ & $\{ \domelemd \mid (\domelemd, \domelemd) \in \roler^{\interI} \}$\\\hline
        \end{tabular}
\end{table}

\paragraph{$\ALC$ with path expressions.}
We treat the set $\alphabetall \deff \Rlang \cup \{ \conceptC? \mid \conceptC \in \Clang\}$ as an infinite alphabet.
Let $\ALL$ and $\REG$ be classes of all Turing-recognizable (resp. all regular) languages of finite words over any finite subset of $\alphabetall$. 
Given a language $\languageL \in \ALL$ and a path $\pathrho \deff \pathrho_1 \pathrho_2 \ldots \pathrho_n \pathrho_{n{+}1}$ in an interpretation $\interI$, we say that $\pathrho$ is an $\languageL$-\emph{path}, if there exists a word $\tapewordw \deff \tapewordw_1\tapewordw_2\ldots\tapewordw_n  \in \languageL$ such that for all indices $i \leq n$ we have either (i) $\tapewordw_i \in \Rlang$ and $(\pathrho_i, \pathrho_{i{+}1}) \in (\tapewordw_i)^{\interI} $, or (ii) $\tapewordw_i$ is of the form $\conceptC?$ for some $\conceptC \in \Clang$, $\pathrho_i = \pathrho_{i{+}1}$ and $\pathrho_i \in \conceptC^{\interI}$.
Intuitively the word~$\tapewordw$ either traverses roles between the elements of $\pathrho$ or loops at an element to check the satisfaction of concepts.
For convenience, we will also say that $\domeleme \in \DeltaI$ is $\languageL$-\emph{reachable from} $\domelemd \in \DeltaI$ (or alternatively that $\domelemd$ $\languageL$-\emph{reaches} $\domeleme$) if there is an $\languageL$-path $\pathrho$ that starts from $\domelemd$ and ends in $\domeleme$.
The logic $\ALCall$ extends $\ALC$ with concept constructors of the form $\exists{\languageL}.\conceptC$, where $\languageL \in \ALL$ and $\conceptC$ is an $\ALCall$-concept.  
Their semantics is as follows: $(\exists{\languageL}.\conceptC)^{\interI}$ is the set of all $\domelemd \in \DeltaI$ that can $\languageL$-reach some $\domeleme \in \conceptC^{\interI}$, and $\forall{\languageL}.\conceptC$ stands for $\neg\exists{\languageL}.\neg\conceptC$.
The logic $\ALCreg$ is a restriction of $\ALCall$ to regular languages. It is a  notational variant of Propositional Dynamic Logic~\cite{FischerLadner1979}, popular in the community of formal verification.

\paragraph{Visibly-pushdown languages.}

The class $\VPL$ of \emph{Visibly-pushdown languages}~\cite[Sec.~5]{alur2009} (VPLs)
is a well-behaved family of context-free languages, 
in which the usage of the stack in the underlying automaton model 
is input-driven. 
For the exposition of~VPLs we follow the work of L\"oding et al.~\cite[Sec. 2.2]{LodingLS07}.
A \emph{pushdown alphabet} $\alphabet$ is an alphabet equipped with 
a partition $(\alphabet_c, \alphabet_i, \alphabet_r)$.
The elements of $\alphabet_c, \alphabet_i,$ and $\alphabet_r$ are called, respectively,
\emph{call} letters, \emph{internal} letters, and \emph{return} letters.
A \emph{visibly-pushdown automaton}~(VPA) $\automatonA$ over a pushdown alphabet $\alphabet$
is a tuple 
$(\statesQ, \initialstates, \finalstates, \stackalphabet, \transreldelta)$,
where $\statesQ$ is a finite set of states,
$\initialstates$ is a finite subset of \emph{initial states},
$\finalstates$ is a finite subset of \emph{final states},
$\stackalphabet$ is a finite stack alphabet that contains a 
\emph{bottom-of-stack} symbol $\botofstack$, and
$\transreldelta$ is a transition relation~of~type
\[
  \transreldelta \subseteq
  \left( \statesQ \times \alphabet_c \times \statesQ \times (\stackalphabet \setminus \botofstack) \right) \; \; \; \cup \; \; \; 
  \left( \statesQ \times \alphabet_r \times \stackalphabet \times \statesQ) \right) 
  \; \; \; \cup \; \; \; 
  \left( \statesQ \times \alphabet_i \times \statesQ \right).
\]
A \emph{configuration} of a VPA $\automatonA$ is
a pair $(\stateq, \stacksigma) \in \statesQ \times (\stackalphabet \setminus \botofstack)^*\botofstack$
of a state $\stateq$ and a stack content~$\stacksigma$.
Given a letter $\lettera$ and a configuration $(\stateq, \stacksigma)$ we say that $(\stateq', \stacksigma')$ is an $\lettera$-successor of $(\stateq, \stacksigma)$ if one of the following cases hold: 
\begin{itemize}\itemsep0em
\item $\lettera \in \alphabet_c$, $\stacksigma' = \gamma \stacksigma$ and there is a transition 
$(\stateq, \lettera, \stateq', \gamma) \in \transreldelta$.
\item $\lettera \in \alphabet_i$, $\stacksigma' = \stacksigma$ and there is a transition 
$(\stateq, \lettera, \stateq') \in \transreldelta$. 
\item $\lettera \in \alphabet_r$, either
(i) $\stacksigma = \gamma \stacksigma'$ and there is a transition $(\stateq, \lettera, \gamma, \stateq') \in \transreldelta$, or (ii) $\stacksigma = \stacksigma' = \botofstack$ and~$(\stateq, \lettera, \botofstack, \stateq') \in \transreldelta$.
\end{itemize}
We denote this fact by $(\stateq, \stacksigma) \to_\lettera (\stateq', \stacksigma')$.
Given a finite word $\tapewordw \deff \lettera_1 \ldots \lettera_{n}$, a \emph{run} of $\automatonA$ on~$\tapewordw$ is a sequence 
$(\stateq_0, \botofstack) \to_{\lettera_1} (\stateq_1, \sigma_1) \to_{\lettera_2} \ldots \to_{\lettera_n} (\stateq_n, \sigma_n)$ where $\stateq_0 \in \initialstates$.
We call $\tapewordw$ \emph{accepted} by~$\automatonA$ if there is a run in which the last configuration contains a final state.
The language $\languageL(\automatonA)$ of $\automatonA$ is composed of all words accepted by $\automatonA$.
A~language $\languageL$ (\ie a set of words) over $\alphabet$ is \emph{visibly-pushdown} if there is a VPA $\automatonA$ over $\alphabet$ for which $\languageL(\automatonA) = \languageL$.

\begin{exa}\label{example:languages-accepted-by-VPL}
Suppose that $\roler$ is a call letter and $\roles$ is a return letter. 
Then the languages $\langvpaeq{\roler}{\roles} \deff \{ \roler^n \roles^n \mid n \in \N \}$ 
and $\langvpamore{\roler}{\roles} \deff \{ \roler^n \roles^{n{+}m{+}1} \mid n, m \in \N \}$
are visibly-pushdown.
Under such a choice of $\roler$ and $\roles$, the language $\langvpaeq{\roles}{\roler}$ \emph{is not} visibly-pushdown.
Moreover, every regular language is visibly-pushdown.
\end{exa}

%
\paragraph{The logic $\ALCvpl$.}

Throughout the paper, $\alphabetall$ is presented as a pushdown~alphabet
\[
\alphabetvpl \deff \left((\Rlang)_c, (\Rlang)_i \cup \{ \conceptC? \mid \conceptC \in \Clang\}, (\Rlang)_r\right),
\]
where the sets $(\Rlang)_c, (\Rlang)_i, (\Rlang)_r$ form a partition of $\Rlang$.
Hence, we define~$\ALCvpl$ as the restriction of $\ALCall$ to visibly-pushdown languages over finite subsets of $\alphabetvpl$, in which languages in existential and universal restrictions are represented by means of nondeterministic VPA.
The logic $\ALCvpl$ generalises many other logics with non-regular path expressions, and has a $\TwoExpTime$-complete~\cite[Thm.~18--19]{LodingLS07} concept satisfiability problem.
As a special case of $\ALCvpl$, we also consider a, rather minimalistic, extension of~$\ALCreg$ with the language $\langvpaeq{\roler}{\roles}$ for one \emph{fixed} call letter $\roler \in (\Rlang)_c$ and one \emph{fixed} return letter $\roles \in (\Rlang)_r$. 
We~denote it here by $\ALCregrhashshash$.
The concept satisfiability problem for $\ALCregrhashshash$ was shown to be decidable~\cite[Sec. Decidability]{TmimaP1983} already 40 years ago by Koren and Pnueli, but its extensions with popular features like nominals or functionality are still unexplored.

\paragraph{The query entailment problem.}

Given a class $\langclassC$ of languages, the class of \emph{$\langclassC$-enriched Positive Existential Queries} (abbreviated as $\langclassC$-{\PEQ}s) is defined with the following syntax:
\[
 \queryq, \queryq'  \Coloneqq \bot \; \mid \; \conceptA(\varx) \; \mid \;  \roler(\varx, \vary) \; \mid \; \languageL(\varx, \vary) \; \mid \; \queryq \vee \queryq' \; \mid \; \queryq \land \queryq',
\]
where $\conceptA \in \Clang$, $\roler \in \Rlang$, $\languageL \in \langclassC$, and $\varx, \vary$ are variables from a countably-infinite set \( \Vlang \). 
Their semantics is defined as expected, \eg $\languageL(\varx, \vary)$ evaluates to true under a variable assignment $\matcheta\colon \Vlang \to \DeltaI$ if and only if~$\matcheta(\varx)$ can $\languageL$-reach $\matcheta(\vary)$ in $\interI$.
$\emptyset$-{\PEQ}s (or Positive Existential Queries) are a well-known generalization of unions of \emph{conjunctive queries} (in which disjunction is allowed only at the outermost level).
$\REG$-{\PEQ}s (or Positive Conjunctive Regular Path Queries) are among the most popular query languages nowadays~\cite{FlorescuLS98,OrtizS12}.
Finally, $\VPL$-{\PEQ}s recently received some attention~in~\cite{LangeL15}.
A $\langclassC$-{\CQ} is a disjunction-free $\langclassC$-\PEQ.
We~say that an interpretation~$\interI$ \emph{satisfies} a query $\queryq$ (written $\interI \models \queryq$), if there exists an assignment $\matcheta$ of variables from~$\queryq$ to $\DeltaI$ under which $\queryq$ evaluates to true (a \emph{match}).
A~concept~$\conceptC$ entails a query $\queryq$ (written $\conceptC \models \queryq$) if every model of $\conceptC$ satisfies $\queryq$. 
In total analogy we define the notion of entailment and overload the ``$\models$'' symbol for the case of TBoxes and knowledge-bases in place of concepts.
For a given description logic $\abstrDL$, the $\langclassC$-{\PEQ} \emph{entailment problem over $\abstrDL$-knowledge-bases} asks to decide if an input $\abstrDL$-knowledge-base entails an input $\langclassC$-\PEQ.\@
The~problem's definition for concepts and TBoxes is analogous. 

\paragraph{Tree models.}

L\"oding et al. established that $\ALCvpl$ possesses a tree-model property.
An~interpretation $\interI$ is \emph{tree-like} if its domain is a prefix-closed subset of $\N^*$, and for all $\roler \in \Rlang$ and $\domelemd, \domeleme \in \DeltaI$ the condition ''if $(\domelemd, \domeleme) \in \roler^{\interI}$, then $\domeleme = \domelemd n$ for some~$n \in \N$'' holds. 
An interpretation is \emph{single-role} if any two domain elements are connected by at most one role.
The following lemma follows immediately from the proof of Proposition~8 by L\"oding et al.~\cite{LodingLS07}.

\begin{cor}[Consequence of the proof of Prop.~8 of~\cite{LodingLS07}]\label{cor:tree-model-property-for-ALCVPL}
Every satisfiable $\ALCvpl$-TBox~$\tboxT$ has a single-role tree-like model.\footnote{The original work considers concepts only. However, all their results transfer immediately to the case of TBoxes, as TBoxes can be internalised in concepts in the presence of regular expressions~\cite[p.~186]{2003handbook}. The queries are not mentioned either: the so-called tree model property is established with a suitable notion of unravelling, which produces interpretations that can be then homomorphically mapped to the original interpretations (entailing the preservation of the non-satisfaction of query).}
Moreover, for any query $\queryq$ preserved under homomorphisms we have that $\tboxT \not\models \queryq$ if and only if there exists a single-role tree-like model $\interI \models \tboxT$ such that $\interI \not\models \queryq$.
\end{cor}

\paragraph{Undecidability results for extensions of $\ALCvpl$ that follow from the literature.}
Other popular (not yet mentioned) features supported by W3C ontology languages are \emph{inverse roles} and \emph{role hierarchies}. 
For bibliographical purposes, we would like to use the extra space given here to briefly discuss how these features result in the undecidability of the respective extensions of $\ALCvpl$.
To define the first feature, we associate with each role name $\roler \in \Rlang$ a role name $\roler^-$ that is enforced to be interpreted as the inverse of the interpretation of $\roler$. 
It was shown by Stefan G\"oller in his PhD thesis~\cite[Prop.~2.32]{Goller2008} that $\ALCreg$ extended with the single visibly-pushdown language $\{ \roler^{n}\roles(\roler^-)^n \mid n \in \N \}$ is~undecidable. 

\begin{cor}\label{cor:Goller}
The concept satisfiability problem for $\ALCvpl$ with inverses is undecidable, even if the only allowed non-regular language is $\{ \roler^{n}\roles(\roler^-)^n \mid n \in \N \}$ for fixed $\roler, \roles \in \Rlang$.
\end{cor}

The second feature allows for specifying containment of roles by means of expressions of the form for $\roler \subseteq \roles$.
The semantics is that $\interI \models \roler \subseteq \roles$ if and only if $\roler^{\interI} \subseteq \roles^{\interI}$.
Fix $\roler$ and $\roler'$ to be call letters, and $\roles$ and $\roles'$ to be return letters.
Suppose that an interpretation $\interI$ satisfies all of the statements $\roles \subseteq \roler'$, $\roler' \subseteq \roles$, $\roles' \subseteq \roler$, and $\roler \subseteq \roles'$.
Clearly, for all elements $\domelemd \in \DeltaI$ and concepts $\conceptC$ we have that $\domelemd \in (\exists{\langvpaeq{\roles}{\roler}}.\conceptC)^{\interI}$ if and only if $\domelemd \in (\exists{\langvpaeq{\roler'}{\roles'}}.\conceptC)^{\interI}$.
Thus $\ALCvpl$ with role-hierarchies can express concepts of $\ALCreg$ extended with both $\langvpaeq{\roler}{\roles}$ and $\langvpaeq{\roles}{\roler}$.
By undecidability of the concept satisfiability~\cite[Sec. Decidability]{TmimaP1983} of the~latter~we~conclude:

\begin{cor}
The concept satisfiability problem for $\ALCvpl$ with role-hierarchies is undecidable, even if the only allowed non-regular languages are $\langvpaeq{\roler}{\roles}$ and $\langvpaeq{\roler'}{\roles'}$ for fixed call letters~$\roler, \roler'$ and return letters~$\roles, \roles'$.
\end{cor}

%% file: sections/self.tex

\section{Negative results I:\@ The Seemingly innocent \texorpdfstring{$\Self$}{Self} operator}\label{sec:self}

We start our series of negative results, by showing (in our opinion) a rather surprising undecidability result. 
Henceforth we employ the $\Self$ operator, a modelling feature supported by two profiles of the OWL 2 Web Ontology Language~\cite{HorrocksKS06,KRH2008}.
Recall that the $\Self$ operator allows us to specify the situation when an element is related to it\emph{self} by a binary relationship, \ie we interpret the concept $\exists{\roler}.\Self$ in an interpretation $\interI$ as the set of all those elements $\domelemd$ for which $(\domelemd,\domelemd)$ belongs to $\roler^{\interI}$. 
We provide a reduction from the undecidable problem of non-emptiness of the intersection of deterministic one-counter automata (DOCA)~\cite[p. 75]{valiant1973decision}. 
Such an automaton model is similar to pushdown automata, but its stack alphabet is single-letter only.
The $\Self$ operator will be especially useful to introduce ``disjunction''~to~paths.

Let $\alphabet$ be an alphabet and $\tapewordw \deff (\lettera_1, \letterb_1) \ldots (\lettera_{n}, \letterb_{n})$ be a word over $\alphabet \times \{ c, r, i\}$.
We call the word $\proj_1(\tapewordw) \deff \lettera_1 \ldots \lettera_{n}$ from $\alphabet^*$ the \emph{projection} of $\tapewordw$.
An important property of DOCAs is that they can be made visibly pushdown in the following sense.

\begin{lem}\label{lemma:from-cfl-to-vpl}
For any deterministic one-counter automaton $\automatonA$ over the alphabet $\alphabet$, 
we can construct a visibly-pushdown automaton $\tilde{\automatonA}$ over a pushdown alphabet $\tilde{\alphabet} \deff (\alphabet \times \{ c \}, (\alphabet \times \{ i \}) \cup \{ \letterx \}, \alphabet \times \{ r \})$ with a fresh internal letter $\letterx$ such that all words in $\languageL(\tilde{\automatonA})$ are of the form $\tilde{\lettera_1} \letterx \tilde{\lettera_2} \letterx \ldots \letterx \tilde{\lettera_n}$ for $\tilde{\lettera_1}, \ldots, \tilde{\lettera_n} \in \alphabet \times \{ c, i, r\}$, and
\[
  \languageL(\automatonA) = \lbrace \proj_1(\tilde{\tapewordw}) \mid \tilde{\tapewordw} \deff \tilde{\lettera_1}\ldots\tilde{\lettera_n}, \; \; \tilde{\lettera_1}\letterx\ldots\letterx\tilde{\lettera_n} \in \languageL(\tilde{\automatonA}) \rbrace. 
\]
\end{lem}
\begin{proof}[Proof sketch.]
  Alur and Madhusudan proved~\cite[Thm. 5.2]{alur2009} that for any context-free language $\languageL$ over~$\alphabet$ there exists a VPL $\hat{\languageL}$, over the pushdown alphabet $(\alphabet \times \{ c \}, (\alphabet \times \{ i \}), \alphabet \times \{ r \})$, for which $\languageL = \{ \proj_1(\tapewordw) \mid \tapewordw \in \hat{\languageL} \}$ holds.\footnote{The main proof idea here is to take an input DOCA, and decorate the letters on its transitions with $c$, $i$, and $r$, depending on the counter action of the transition.} 
  Suppose now that a one-counter automaton $\automatonA$ is given. 
  By means of the previous construction, we obtain a visibly-pushdown automaton~$\hat{\automatonA} \deff (\statesQ, \initialstates, \finalstates, \stackalphabet, \transreldelta)$ for which $\languageL(\automatonA) = \{ \proj_1(\tapewordw) \mid \tapewordw \in \languageL(\hat{\automatonA})\}$ holds. 
  What remains to be done is to ``insert'' the \emph{internal} letter $\letterx$ after every position of a word accepted by $\hat{\automatonA}$. 
  As reading internal letters by visibly-pushdown automata do not affect the content of their stacks, we may proceed as in standard constructions from the theory of regular languages~\cite[Ex. 1.31]{sipser13}.
  As the first step, we expand the set of states $\statesQ$ with fresh states of the form~$\stateq_\delta$ for all $\delta \in \transreldelta$.
  As the second step, we ``split'' every transition $\delta$ in $\transreldelta$ into two ``parts''.
  Suppose that $\delta$ leads from $\stateq$ to $\stateq'$ after reading the letter $\lettera$.
  We thus (i) replace $\delta$ in $\transreldelta$ with the transition that transforms $\stateq$ into~$\stateq_\delta$ after reading $\lettera$ (and has the same effect on the stack as $\delta$ has), and (ii) append the transition $(\stateq_\delta, \letterx, \stateq')$ to $\transreldelta$.
  Call the resulting automaton~$\tilde{\automatonA}$.
  It can now be readily verified that $\languageL(\tilde{\automatonA}) = \{ \hat{\lettera_1}\letterx\hat{\lettera_2}\letterx\ldots\letterx\hat{\lettera_n} \mid \hat{\tapewordw} \deff \hat{\lettera_1}\ldots\hat{\lettera_n}, \ \hat{\tapewordw} \in \languageL(\hat{\automatonA}) \}$ holds,
  and thus, by the relationship between $\automatonA$ and~$\hat{\automatonA}$, the automaton $\tilde{\automatonA}$ is as desired.
\end{proof}

Let us fix a finite alphabet $\alphabet \subseteq \Rlang$.
We also fix two deterministic one-counter automata~$\automatonA_1$ and~$\automatonA_2$ over $\alphabet$, and let~$\automatonC_1$ and~$\automatonC_2$ be  deterministic one-counter automata recognizing the complement of the languages of $\automatonA_1$ and~$\automatonA_2$ (they exist as DOCA are closed under complement~\cite[p.~76]{valiant1973decision}). 
Finally, we apply Lemma~\ref{lemma:from-cfl-to-vpl} to construct their visibly-pushdown counterparts $\tilde{\automatonA}_1$, $\tilde{\automatonA}_2$, $\tilde{\automatonC}_1$, $\tilde{\automatonC}_2$ over \emph{the same} pushdown alphabet $\tilde{\alphabet}$. 
We stress that the letter $\letterx$, playing the role of a ``separator'', is identical for all of the aforementioned visibly-pushdown automata.
Moreover, note that the non-emptiness of $\languageL(\tilde{\automatonA}_1) \cap \languageL(\tilde{\automatonA}_2)$ is not equivalent to the non-emptiness of $\languageL(\automatonA_1) \cap \languageL(\automatonA_1)$, as the projection of a letter $\lettera \in \tilde{\alphabet}$ may be used by $\automatonA_1$ and $\automatonA_2$ in different contexts (\eg both as a call or as a return).

We are going to encode words accepted by one-counter automata by means of word-like interpretations.
A pointed interpretation $(\interI, \domelemd)$ is \emph{$\alphabet$-friendly} if 
for every element $\domeleme \in \DeltaI$ that is $\letterx^*$-reachable from $\domelemd$ in $\interI$ there exists a unique letter $\lettera \in \alphabet$ so that $\domeleme$ carries all $\tilde{\lettera}$-self-loops for all $\tilde{\lettera} \in \tilde{\alphabet}$ with~$\proj_1(\tilde{\lettera}) = \lettera$, and no self-loops for all other letters in $\tilde{\alphabet}$ (also including $\letterx$).

Observe that $\alphabet$-friendly interpretations can be axiomatised with an $\ALCSelf$-concept~$\conceptC_{\textrm{fr}}^{\alphabet}$:
\[
  \conceptC_{\textrm{fr}}^{\alphabet} \deff \forall{\letterx^*}.\bigdlor_{\lettera \in \alphabet}\ \bigdland_{\letterb \neq \lettera, \letterb \in \alphabet, \proj_1(\tilde{\lettera}) = \lettera, \proj_1(\tilde{\letterb}) = \letterb} \ \Big( [\exists{\tilde{\lettera}}.\Self] \dland \neg [\exists{\tilde{\letterb}}.\Self] \dland \neg [\exists{\letterx}.\Self] \Big).
\]
Moreover, every $\letterx^*$-path $\pathrho$ in a $\alphabet$-friendly $(\interI, \domelemd)$ \emph{represents} a word in $\alphabet^*$ in the following sense: the $i$-th letter of such a word is $\lettera$ if and only if the $i$-th element of the path carries an $(\lettera,c)$-self-loop.
This is well-defined, by the fact that every $\letterx^*$-reachable element in $\alphabet$-friendly~$(\interI, \domelemd)$ carries a $(\lettera,c)$-self-loop for a unique letter $\lettera \in \alphabet$. Consult Figure~\ref{fig:sigma-friendly-interpretation}~for~a~visualization.

\begin{figure}[h]
  \centering 
 \begin{tikzpicture}[transform shape]
       \draw (-0.5, 0) node[] (A0) {$\domelemd^{\interI}$};
      \draw (0, 0) node[medrond] (A1) {};
      \draw (2.5, 0) node[medrond] (A2) {};
      \draw (5, 0) node[medrond] (A3) {};
      \draw (7.5, 0) node[medrond] (A4) {};
      \draw (10, 0) node[medrond] (A5) {};

  \path[->, >=stealth] (A1) edge [loop above] node {$(\lettera, c), (\lettera, r)$} ();
  \path[->, >=stealth] (A1) edge [loop below] node {$(\lettera, i)$} ();

  \path[->, >=stealth] (A2) edge [loop above] node {$(\letterb, c), (\letterb, r)$} ();
  \path[->, >=stealth] (A2) edge [loop below] node {$(\letterb, i)$} ();

  \path[->, >=stealth] (A3) edge [loop above] node {$(\letterb, c), (\letterb, r)$} ();
  \path[->, >=stealth] (A3) edge [loop below] node {$(\letterb, i)$} ();

    \path[->, >=stealth] (A4) edge [loop above] node {$(\lettera, c), (\lettera, r)$} ();
  \path[->, >=stealth] (A4) edge [loop below] node {$(\lettera, i)$} ();

    \path[->, >=stealth] (A5) edge [loop above] node {$(\letterc, c), (\letterc, r)$} ();
  \path[->, >=stealth] (A5) edge [loop below] node {$(\letterc, i)$} ();

      \path[->] (A1) edge [] node[yshift=3] {$\letterx$} (A2);
      \path[->] (A2) edge [] node[yshift=3] {$\letterx$} (A3);
      \path[->] (A3) edge [] node[yshift=3] {$\letterx$} (A4);
      \path[->] (A4) edge [] node[yshift=3] {$\letterx$} (A5);

  \end{tikzpicture}    
    \caption{An example $\alphabet$-friendly $(\interI, \domelemd)$ encoding the word $\tapeword{abbac}$.}
      \label{fig:sigma-friendly-interpretation}
\end{figure}

As a special class of $\alphabet$-friendly interpretations we consider $\alphabet$-metawords.
We say that $(\interI, \domelemd)$ is a \emph{$\alphabet$-metaword} if it is a $\alphabet$-friendly interpretation of the domain $\ZZ_n$ for some positive~$n \in \N$, the  role name~$\letterx$ is interpreted as the set $\{ (i,i{+}1) \mid 0 \leq i \leq n{-}2 \}$, and all other role names are either interpreted as~$\emptyset$ or are subsets of the diagonal $\{ (i,i) \mid i \in \ZZ_n\}$ (or, put differently, they appear only as self-loops).
The example $\alphabet$-friendly $(\interI, \domelemd)$ from Figure~\ref{fig:sigma-friendly-interpretation} is actually a $\alphabet$-metaword.
Note that for every word $\tapewordw \in \alphabet^+$ there is a $\alphabet$-metaword representing~$\tapewordw$.
A~crucial observation regarding $\alphabet$-metawords is as follows.

\begin{obs}\label{obs:meta-word-and-paths}
Let $\ell \in \{ 1, 2\}$, $\alphabet$-metaword $(\interI, \domelemd)$, and $\tilde{\tapewordw}$ be in the language of $ \tilde{\automatonA_\ell}$. 
Then the element $\domelemd$ can $\{ \tilde{\tapewordw} \}$-reach an element $\domeleme$ via a path $\pathrho$ if and only if for all odd indices $i$ we have $\pathrho_i = \pathrho_{i{+}1}$ and for all even indices $i$ we have $\pathrho_i + 1 = \pathrho_{i{+}1}$.
\end{obs}

As the next step of the construction, we are going to decorate $\alphabet$-friendly interpretations with extra information on whether or not words represented by paths are accepted by $\automatonA_1$. 
This is achieved by means of the following~concept 
\[
\conceptC_{\automatonA_1} \deff \conceptC_{\textrm{fr}}^{\alphabet} \dland \forall{\languageL(\tilde{\automatonA}_1)}.\concept{Acc}_{\automatonA_1} \dland \forall{\languageL(\tilde{\automatonC}_1)}.\neg\concept{Acc}_{\automatonA_1}, 
\]
for a fresh concept name $\concept{Acc}_{\automatonA_1}$. 
We define the concept $\conceptC_{\automatonA_2}$ analogously.
We have that:
\begin{lem}\label{lemma:sigma-friendly}
Fix $\ell \in \{ 1, 2\}$.
If $\conceptC_{\automatonA_\ell}$ is satisfied by a pointed interpretation~$(\interI, \domelemd)$,
then $(\interI, \domelemd)$ is $\alphabet$-friendly and for every element $\domeleme \in \DeltaI$ that is $\letterx^*$-reachable from~$\domelemd$ via a path $\pathrho$ we have~$\domeleme \in (\concept{Acc}_{\automatonA_\ell})^{\interI}$ if and only if the $\alphabet$-word represented by $\pathrho$ belongs to $\languageL(\automatonA_\ell)$.
Moreover, after reinterpreting the concept name $\concept{Acc}_{\automatonA_\ell}$, every $\alphabet$-metaword becomes a model of~$\conceptC_{\automatonA_\ell}$.
\end{lem}
\begin{proof}
  The proof relies on Observation~\ref{obs:meta-word-and-paths}.
  Suppose that $(\interI, \domelemd)$ is a pointed interpretation and $\domelemd \in (\conceptC_{\automatonA_\ell})^{\interI}$.
  We know that $(\interI, \domelemd)$ is $\alphabet$-friendly by the satisfaction of $\conceptC_{\textrm{fr}}^{\alphabet} $. 
  For the remainder of the proof, consider any element $\domeleme \in \DeltaI$ that is $\letterx^*$-reachable from~$\domelemd$, say, via a path $\pathrho \deff \pathrho_1 \ldots \pathrho_n$. 
  Let $\tapewordw$ be the word represented by $\pathrho$.
  This implies, that for every index $i$, the element $\pathrho_i$ in $\interI$ is equipped with a family of self-loops involving (a decorated) letter $\tapewordw_i$.
  We consider two cases:
  \begin{itemize}\itemsep0em
    \item Assume that $\domeleme \in (\concept{Acc}_{\automatonA_\ell})^{\interI}$. 
    We will show that $\tapewordw \in \languageL(\automatonA_\ell)$. 
    Ad absurdum, suppose that $\tapewordw \not\in \languageL(\automatonA_\ell)$.
    Then, by definition of $\automatonC_\ell$, we have that $\tapewordw$ belongs to $\languageL(\automatonC_\ell)$.
    By the construction of $\tilde{\automatonC_\ell}$ there exists a sequence~$\star_1, \ldots, \star_{n} \in \{ c, i, r \}$, for which the word $\tapewordu \deff (\tapewordw_1, \star_1) \letterx \ldots \letterx (\tapewordw_{n}, \star_{n})$ is accepted by $\tilde{\automatonC_\ell}$.
    But then the path $\pathrho' \deff \pathrho_1 \pathrho_1 \pathrho_2 \pathrho_2 \ldots \pathrho_n \pathrho_n$ witnesses $\{ \tapewordu \}$-reachability (and thus $\languageL(\tilde{\automatonC_\ell})$-reachability) of $\domeleme$ from $\domelemd$.
    Due to the satisfaction of $\forall{\languageL(\tilde{\automatonC}_\ell)}.\neg\concept{Acc}_{\automatonA_\ell}$ by $(\interI, \domelemd)$ we infer $\domeleme \in (\neg\concept{Acc}_{\automatonA_\ell})^{\interI} = \DeltaI \setminus (\concept{Acc}_{\automatonA_\ell})^{\interI}$. 
    A contradiction.
    \item Assume that $\tapewordw \in \languageL(\automatonA_\ell)$. 
    We proceed analogously to the previous case.
    By the construction of $\tilde{\automatonA_\ell}$ there exists a sequence~$\star_1, \ldots, \star_{n} \in \{ c, i, r \}$, for which the word $\tapewordu \deff (\tapewordw_1, \star_1) \letterx \ldots \letterx (\tapewordw_{n}, \star_{n})$ is accepted by $\tilde{\automatonA_\ell}$.
    Once again, the path $\pathrho' \deff \pathrho_1 \pathrho_1 \pathrho_2 \pathrho_2 \ldots \pathrho_n \pathrho_n$ witnesses $\{ \tapewordu \}$-reachability (and thus $\languageL(\tilde{\automatonA_\ell})$-reachability) of $\domeleme$ from $\domelemd$.
    Due to the satisfaction of $\forall{\languageL(\tilde{\automatonA}_\ell)}.\concept{Acc}_{\automatonA_\ell}$ by $(\interI, \domelemd)$ we infer $\domeleme \in (\concept{Acc}_{\automatonA_\ell})^{\interI}$, as desired. 
  \end{itemize}
  For the last statement of the proof, take a $\alphabet$-metaword $\interI$ that represents a word $\tapewordw \in \alphabet^*$.
  We alter the interpretation of the concept name $\concept{Acc}_{\automatonA_\ell}$ in $\interI$ so that $(\concept{Acc}_{\automatonA_\ell}^{\interI}) = \{ i{-}1 \mid \tapewordw_1\ldots\tapewordw_i \in \languageL(\automatonA_\ell) \}$.
  It follows that  $\forall{\languageL(\tilde{\automatonA}_\ell)}.\concept{Acc}_{\automatonA_\ell} \dland \forall{\languageL(\tilde{\automatonC}_\ell)}.\neg\concept{Acc}_{\automatonA_\ell}$ is indeed satisfied by~$(\interI, 0)$.
\end{proof}

Equipped with Lemma~\ref{lemma:sigma-friendly}, we are ready to prove correctness of our reduction.
\begin{lem}\label{lemma:sat-iff-intersection-non-empty}
$\conceptC_{\automatonA_1} \dland \conceptC_{\automatonA_2} \dland \exists{\letterx^*}.\left( \concept{Acc}_{\automatonA_1} {\dland} \concept{Acc}_{\automatonA_2} \right)$ is satisfiable iff $\languageL(\automatonA_1) \cap \languageL(\automatonA_2) \neq \emptyset$.
\end{lem}
\begin{proof}
For one direction, take $\tapewordw \in \languageL(\automatonA_1) \cap \languageL(\automatonA_2)$.
Let $(\interI, \domelemd)$ be a $\alphabet$-metaword representing~$\tapewordw$.
By Lemma~\ref{lemma:sigma-friendly} we decorate $\interI$ with concepts $\concept{Acc}_{\automatonA_1}$ and $\concept{Acc}_{\automatonA_2}$ so that $(\interI, \domelemd) \models \conceptC_{\automatonA_1} \dland \conceptC_{\automatonA_2}$, and the interpretations of concepts $\concept{Acc}_{\automatonA_\ell}$ contain precisely the elements $k$ for which the $k$-letter prefix of $\tapewordw$ belongs to $\languageL(\automatonA_\ell)$.
In particular, this means that $(|\tapewordw|{-}1) \in (\concept{Acc}_{\automatonA_1} \dland \concept{Acc}_{\automatonA_2})^{\interI}$.
As $(|\tapewordw|{-}1)$ is $\letterx^*$-reachable from $0$, we conclude the satisfaction of $\exists{\letterx^*}.\left( \concept{Acc}_{\automatonA_1} \dland \concept{Acc}_{\automatonA_2} \right)$ by~$(\interI, 0)$. 
Hence, the concept from the statement of Lemma~\ref{lemma:sat-iff-intersection-non-empty} is indeed~satisfiable.

For the other direction, assume that $(\interI, \domelemd)$ is a model of $\conceptC_{\automatonA_1} \dland \conceptC_{\automatonA_2} \dland \exists{\letterx^*}.\left( \concept{Acc}_{\automatonA_1} \dland \concept{Acc}_{\automatonA_2} \right)$. 
Then there exists an $\letterx^*$-path $\pathrho$ from $\domelemd$ to some $\domeleme \in (\concept{Acc}_{\automatonA_1} \dland \concept{Acc}_{\automatonA_2})^{\interI}$.
Hence, by Lemma~\ref{lemma:sigma-friendly}, for all~$\ell \in \{1,2\}$ the word represented by $\pathrho$ belongs to $\languageL(\automatonA_\ell)$, and thus also $\languageL(\automatonA_1) \cap \languageL(\automatonA_2)$.
\end{proof}

By the undecidability of the non-emptiness problem for intersection of one-counter languages~\cite[p. 75]{valiant1973decision}, we conclude Theorem~\ref{thm:ALCSelfreg-undec-with-VPL}. 
\begin{thm}\label{thm:ALCSelfreg-undec-with-VPL}
The concept satisfiability problem for $\ALCSelfvpl$ is undecidable, even if only visibly-pushdown languages that are encodings of DOCA languages are allowed in concepts.
\end{thm}%

There is nothing special about deterministic one-counter automata used in the proof.
In fact, any automaton model would satisfy our needs as long as it would (i) have an undecidable non-emptiness problem for the intersection of languages, (ii) enjoy the analogue of Lemma~\ref{lemma:from-cfl-to-vpl}, and (iii) be closed under complement.
We leave it is an open problem to see if there exists a \emph{single} visibly-pushdown language $\languageL$ that makes the concept satisfiability of $\ALCSelf_{\mathsf{reg}}$ extended with $\languageL$ undecidable.
For instance, the decidability status of $\ALCregrhashshash$ with $\Self$~is~open.

%% file: sections/nominals.tex

\section{Negative results II:\@ Nominals meet \texorpdfstring{$\langvpaeq{\roler}{\roles}$}{rnsn}}\label{sec:nominals}

We next provide an undecidability proof for the concept satisfiability problem for~$\ALCOregrhashshash$.
To~achieve this, we employ a slight variant of the classical domino tiling problem~\cite{Wang1961ProvingTB}.

A \emph{domino tiling system} is a triple $\tilingsys \deff (\tilesCol, \tiles, \tileswhite)$,
where $\tilesCol$ is a finite set of \emph{colours}, 
$\tiles \subseteq \tilesCol^4$ is a set of $4$-sided \emph{tiles},
and $\tileswhite \in \tilesCol$ is a distinguished colour called \emph{white}.
For brevity, we call a tile $(\colour_l, \colour_d, \colour_r, \colour_u) \in \tiles$ (i) \emph{left-border} if $\colour_l = \whiteBox$, (ii) \emph{down-border} if $\colour_d = \whiteBox$, (iii) \emph{right-border} if $\colour_r = \whiteBox$, and (iv) \emph{up-border} if $\colour_u = \whiteBox$.
We also say that tiles $\tile \deff (\colour_l, \colour_d, \colour_r, \colour_u)$ and $\tile' \deff (\colour_l', \colour_d', \colour_r', \colour_u')$ from $\tiles$ are (i) \emph{$\rmH$-compatible} if $\colour_r = \colour_l'$, and (ii) \emph{$\rmV$-compatible} if $\colour_u = \colour_d'$.
We say that~$\tilingsys$ \emph{covers} $\ZZ_n \times \ZZ_m$ (where $n$ and $m$ are positive integers) if there exists a mapping $\tilesmap \colon \ZZ_n \times \ZZ_m \to \tiles$ such that for all pairs $(x, y) \in \ZZ_n \times \ZZ_m$ with $\tilesmap(x,y) \deff (\colour_l, \colour_d, \colour_r, \colour_u)$ the following conditions are satisfied:
  \begin{description}\itemsep0em
    \item[\desclabel{(TBorders)}{tiles:Borders}] $x = 0$ iff $\colour_l = \tileswhite$; $x = n{-}1$ iff $\colour_r = \tileswhite$; $y = 0$ iff $\colour_d = \tileswhite$; $y = m{-}1$ iff $\colour_u = \tileswhite$;
    \item[\desclabel{(THori)}{tiles:Hori}] If $(x{+}1,y) \in \ZZ_n \times \ZZ_m$ then $\tilesmap(x,y)$ and $\tilesmap(x{+}1, y)$ are $\rmH$-compatible.
    \item[\desclabel{(TVerti)}{tiles:Verti}] If $(x,y{+}1) \in \ZZ_n \times \ZZ_m$ then $\tilesmap(x,y)$ and $\tilesmap(x, y{+}1)$ are $\rmV$-compatible.
  \end{description} 
Intuitively, $\tilesmap : \ZZ_n \times \ZZ_m$ can be seen as a white-bordered rectangle of size $n \times m$ coloured by unit $4$-sided tiles (with coordinates corresponding to the left, down, right, and upper colour) from $\tiles$, where sides of tiles of consecutive squares have matching colours.

\begin{exa}\label{ex:visualization-of-tilesmap}
Suppose that $\tilesCol = \{ \cyanBox, \rougeBox, \whiteBox, \limeBox \}$ and $\tiles = \tilesCol^4$.
Then the map $\tilesmap \deff \{ 
  (0,0) \mapsto \intextwang{white}{white}{cyan}{rouge},
  (1,0) \mapsto \intextwang{cyan}{white}{lime}{rouge},
  (2,0) \mapsto \intextwang{lime}{white}{rouge}{rouge}, 
  (3,0) \mapsto \intextwang{rouge}{white}{white}{cyan},
  (0,1) \mapsto \intextwang{white}{rouge}{cyan}{lime},
  (1,1) \mapsto \intextwang{cyan}{rouge}{cyan}{cyan},
  (2,1) \mapsto \intextwang{cyan}{rouge}{rouge}{cyan},
  (3,1) \mapsto \intextwang{rouge}{cyan}{white}{lime},
  (0,2) \mapsto \intextwang{white}{lime}{cyan}{white},
  (1,2) \mapsto \intextwang{cyan}{cyan}{cyan}{white},
  (2,2) \mapsto \intextwang{cyan}{cyan}{lime}{white},
  (3,2) \mapsto \intextwang{lime}{lime}{white}{white} 
  \}$ covers $\ZZ_4 \times \ZZ_3$, and can be visualised as follows.
  \begin{figure}[h]
  \centering
    \begin{tikzpicture}[transform shape]

   \wang{0}{0}{white}{white}{cyan}{rouge} 
   \wang{1}{0}{cyan}{white}{lime}{rouge} 
   \wang{2}{0}{lime}{white}{rouge}{rouge} 
   \wang{3}{0}{rouge}{white}{white}{cyan}

   \wang{0}{1}{white}{rouge}{cyan}{lime} 
   \wang{1}{1}{cyan}{rouge}{cyan}{cyan} 
   \wang{2}{1}{cyan}{rouge}{rouge}{cyan} 
   \wang{3}{1}{rouge}{cyan}{white}{lime} 

   \wang{0}{2}{white}{lime}{cyan}{white} 
   \wang{1}{2}{cyan}{cyan}{cyan}{white} 
   \wang{2}{2}{cyan}{cyan}{lime}{white} 
   \wang{3}{2}{lime}{lime}{white}{white} 
    \draw (0.5, -0.2) node[label=center:\( 0 \)] (X) {\phantom{0}};
    \draw (1.5, -0.2) node[label=center:\( 1 \)] (X) {\phantom{0}};
    \draw (2.5, -0.2) node[label=center:\( 2 \)] (X) {\phantom{0}};
    \draw (3.5, -0.2) node[label=center:\( 3 \)] (X) {\phantom{0}};
    \draw (-0.2, 0.5) node[label=center:\( 0 \)] (X) {\phantom{0}};
    \draw (-0.2, 1.5) node[label=center:\( 1 \)] (X) {\phantom{0}};
    \draw (-0.2, 2.5) node[label=center:\( 2 \)] (X) {\phantom{0}};


          \draw (6, 0+0.5) node[medrond] (A1) {};
        \draw (7.5, 0+0.5) node[medrond] (A2) {};
        \draw (9, 0+0.5) node[medrond] (A3) {};
        \draw (10.5, 0+0.5) node[medrond] (A4) {};
        \draw (6, 1+0.5) node[medrond] (B1) {};
        \draw (7.5, 1+0.5) node[medrond] (B2) {};
        \draw (9, 1+0.5) node[medrond] (B3) {};
        \draw (10.5, 1+0.5) node[medrond] (B4) {};
        \draw (6, 2+0.5) node[medrond] (C1) {};
        \draw (7.5, 2+0.5) node[medrond] (C2) {};
        \draw (9, 2+0.5) node[medrond] (C3) {};
        \draw (10.5, 2+0.5) node[medrond] (C4) {};
        \path[->] (A1) edge [red] node[yshift=3] {$\roler$} (A2);
        \path[->] (A2) edge [red] node[yshift=3] {$\roler$} (A3);
        \path[->] (A3) edge [red] node[yshift=3, xshift=-8] {$\roler$} (A4);
        \path[->] (A4) edge [red] node[yshift=3] {$\roler$} (B1);
        \path[->] (B1) edge [red] node[yshift=3] {$\roler$} (B2);
        \path[->] (B2) edge [red] node[yshift=3] {$\roler$} (B3);
        \path[->] (B3) edge [red] node[yshift=3, xshift=-8] {$\roler$} (B4);
        \path[->] (B4) edge [red] node[yshift=3] {$\roler$} (C1);
        \path[->] (C1) edge [red] node[yshift=3] {$\roler$} (C2);
        \path[->] (C2) edge [red] node[yshift=3] {$\roler$} (C3);
        \path[->] (C3) edge [red] node[yshift=3] {$\roler$} (C4);

         \smallwang{6-0.125}{-0.125+0.5}{white}{white}{cyan}{rouge} 
         \smallwang{7.5-0.125}{-0.125+0.5}{cyan}{white}{lime}{rouge} 
         \smallwang{9-0.125}{-0.125+0.5}{lime}{white}{rouge}{rouge} 
         \smallwang{10.5-0.125}{-0.125+0.5}{rouge}{white}{white}{cyan}

         \smallwang{6-0.125}{1-0.125+0.5}{white}{rouge}{cyan}{lime} 
         \smallwang{7.5-0.125}{1-0.125+0.5}{cyan}{rouge}{cyan}{cyan} 
         \smallwang{9-0.125}{1-0.125+0.5}{cyan}{rouge}{rouge}{cyan} 
         \smallwang{10.5-0.125}{1-0.125+0.5}{rouge}{cyan}{white}{lime} 

         \smallwang{6-0.125}{2-0.125+0.5}{white}{lime}{cyan}{white} 
         \smallwang{7.5-0.125}{2-0.125+0.5}{cyan}{cyan}{cyan}{white} 
         \smallwang{9-0.125}{2-0.125+0.5}{cyan}{cyan}{lime}{white} 
         \smallwang{10.5-0.125}{2-0.125+0.5}{lime}{lime}{white}{white}

        \node[below=0.1em of A1] {$\indvsld^{\interI}$};
        \node[below=0.1em of A4] {$\indvsrd^{\interI}$};
        \node[above=0.1em of C4] {$\indvsru^{\interI}$}; 
        \node[above=0.1em of C1] {$\indvslu^{\interI}$}; 

    \end{tikzpicture}
  \end{figure}
\end{exa}

W.l.o.g. we assume that $\tiles$ does not contain tiles having more than $2$ white sides. 
A~system $\tilingsys$ is \emph{solvable} if there exist positive integers $n, m \in \N$ for which $\tilingsys$ covers $\ZZ_n \times \ZZ_m$.
The problem of deciding whether an input domino tiling system is solvable is undecidable, which can be shown by a minor modification of classical undecidability proofs~\cite[Lemma~3.9]{PrattHartmann23}\cite{van1997convenience}.
For a domino tiling system $\tilingsys \deff (\tilesCol, \tiles, \tileswhite)$ we employ fresh concept names from $\concepttilepath^{\tiles} \deff \{ \conceptC_\tile \mid \tile \in \tiles \}$ to encode mappings $\tilesmap$ from some $\ZZ_n \times \ZZ_m$ to $\tiles$ in interpretations~$\interI$ as certain $\roler^+$-paths $\pathrho$ from $\indvsld^{\interI}$ to $\indvsru^{\interI}$ passing through $\indvsrd^{\interI}$ and $\indvslu^{\interI}$  (where the individual names from $\namessnakeT \deff \{ \indvsld, \indvsrd, \indvslu, \indvsru \}$ are fresh). 
Consult the figure in Example~\ref{ex:visualization-of-tilesmap}. 

\begin{defi}\label{def:snake}
Consider a domino tiling system $\tilingsys \deff (\tilesCol, \tiles, \tileswhite)$.
An interpretation $\interI$ is a $\tilingsys$-\emph{snake} whenever all seven criteria listed below are fulfilled:
\begin{description}\itemsep0em
    \item[\desclabel{(SPath)}{snake:Path}] There is an $\roler^+$-path $\pathrho$ that starts in $\indvsld^{\interI}$, then passes through $\indvsrd^{\interI}$, then passes through $\indvslu^{\interI}$ and finishes in $\indvsru^{\interI}$.
    More formally, there are indices $1 < i < j < |\pathrho|$ such that $\pathrho_1 = \indvsld^{\interI}$,  $\pathrho_i = \indvsrd^{\interI}$, $\pathrho_j = \indvslu^{\interI}$ and~$\pathrho_{|\pathrho|} = \indvsru^{\interI}$.
    \item[\desclabel{(SNoLoop)}{snake:NoLoop}] No $\namessnakeT$-named element can $\roler^+$-reach itself.
    \item[\desclabel{(SUniqTil)}{snake:UniqTil}] For every element $\domelemd$ that is $\roler^*$-reachable from $\indvsld^{\interI}$ there exists precisely one tile $\tile \in \tiles$ such that $\domelemd \in \conceptC_\tile^{\interI}$ (we say that $\domelemd$ is \emph{labelled} by a tile $\tile$ or that $\domelemd$ \emph{carries} $\tile$).
    \item[\desclabel{(SSpecTil)}{snake:SpecTil}] The $\namessnakeT$-named elements are unique elements $\roler^*$-reachable from $\indvsld^{\interI}$ that are labelled by tiles with two white sides.
    Moreover, we have that 
    (a) $\indvsld^{\interI}$ carries a tile that is left-border and down-border,
    (b) $\indvsrd^{\interI}$ carries a tile that is right-border and down-border,
    (c) $\indvslu^{\interI}$ carries a tile that is left-border and up-border,
    (d) $\indvsru^{\interI}$ carries a tile that is right-border and up-border.
    \item[\desclabel{(SHori)}{snake:Hori}] For all elements $\domelemd$ different from $\indvsru^{\interI}$ that are $\roler^*$-reachable from $\indvsld^{\interI}$ and labelled by some tile $\tile \deff (\colour_l, \colour_d, \colour_r, \colour_u)$, there exists a tile $\tile' \deff (\colour_l', \colour_d', \colour_r', \colour_u')$ for which all $\roler$-successors $\domeleme$ of $\domelemd$ carry the tile $\tile'$ and: (i) $\tile, \tile'$ are $\rmH$-compatible,
    (ii) if $\colour_d = \whiteBox$ then ($\colour_r \neq \whiteBox$ iff $\colour_d' = \whiteBox$), 
    and (iii) if $\colour_u = \whiteBox$ then~$\colour_u' = \whiteBox$.
    \item[\desclabel{(SLen)}{snake:Len}] There exists a unique positive integer $\rmN$ such that all $\roler^+$-paths between $\indvsld^{\interI}$ and $\indvsrd^{\interI}$ are of length $\rmN{-}1$. Moreover, $\indvsrd^{\interI}$ is the only element $\roler^{\rmN{-}1}$-reachable from $\indvsld^{\interI}$.
    \item[\desclabel{(SVerti)}{snake:Verti}] For all elements $\domelemd$ that are $\roler^*$-reachable from $\indvsld^{\interI}$ and labelled by some $\tile \in \tiles$ that is not up-border, we have that (a) there exists a tile $\tile' \in \tiles$ such that all elements $\domeleme$ $\roler^{\rmN}$-reachable (for $\rmN$ guaranteed by~$\mathrm{\ref{snake:Len}}$) from $\domelemd$ carry $\tile'$, (b) $\tile$ and $\tile'$ are $\rmV$-compatible, (c) $\tile$ is left-border (resp. right-border) if and only if $\tile'$ is.
\end{description}%
\noindent Note that tiles are not ``deterministic'' in the following sense: it could happen that two elements carry the same tile but tiles of their (horizontal or vertical) successors do not~coincide.
\end{defi}

If $\interI$ satisfies all but the last two conditions, we call it a $\tilingsys$-\emph{pseudosnake}.
The key properties of our encoding are extracted and established in Lemmas \ref{lemma:from-tiling-to-snake}--\ref{lemma:from-snakes-to-tilling}.
\begin{lem}\label{lemma:from-tiling-to-snake}
  If a domino tiling system $\tilingsys$ is solvable then there exists a $\tilingsys$-snake. 
\end{lem}
\begin{proof}
  Suppose that $\tilingsys$ covers~$\ZZ_n \times \ZZ_m$ and let $\tilesmap$ be a mapping witnessing it. 
  Define an interpretation $\interI$ as~follows: 
  \begin{enumerate}[(i)]\itemsep0em
  \item $\DeltaI \deff \ZZ_{n \cdot m}$, 
  \item $\indvsld^{\interI} \deff 0$, 
  $\indvsrd^{\interI} \deff n{-}1$, 
  $\indvslu^{\interI} \deff (m{-}1) \cdot n$, 
  $\indvsru^{\interI} \deff m \cdot n - 1$, 
  and $\indva^{\interI} \deff 0$ for all~$\indva \in \Ilang \setminus \namessnakeT$.
  \item $\conceptC_{\tile}^{\interI} \deff \{ (x + y \cdot n) \in \DeltaI \mid \tilesmap(x,y) = \tile \}$ for all $\tile \in \tiles$, and $\conceptC^{\interI} \deff \emptyset$ for $\conceptC \in \Clang \setminus \concepttilepath^{\tiles}$,
  \item $\roler^{\interI} \deff \{ (i, i{+}1) \mid i \in \ZZ_{n \cdot m - 1} \}$, and $(\roler')^{\interI} \deff \emptyset$ for all $\roler' \in \Rlang \setminus \{ \roler \}$.
  \end{enumerate}
  Thus $\interI$ is an $n \cdot m$ element $\roler^+$-path, labelled accordingly to $\tilesmap$.
  As $\tilesmap$ respects $\mathrm{\ref{tiles:Borders}}$, $\mathrm{\ref{tiles:Hori}}$ and $\mathrm{\ref{tiles:Verti}}$, we can readily verify that $\interI$ is indeed a $\tilingsys$-snake.
  The only case that requires treatment, is to verify the satisfaction of $\mathrm{\ref{snake:Hori}}$ for elements of the form $\domelemd \deff (n{-}1) + y \cdot n$ for some $y \in \N$.
  Then, by $\mathrm{\ref{tiles:Borders}}$, $\domelemd$ carries a right-border tile and its $\roler$-successor $\domelemd' \deff 0 + (y{+}1) \cdot n$ carries a left-border tile. 
  Hence, their tiles are $\rmH$-compatible.
\end{proof}

\begin{lem}\label{lemma:from-snakes-to-tilling}
  If there exists a $\tilingsys$-snake for a domino tiling system $\tilingsys$, then $\tilingsys$ is solvable. 
  \end{lem}
\begin{proof}
  Suppose that $\interI$ is a $\tilingsys$-snake. 
  Let $\pathrho \deff \pathrho_1\ldots\pathrho_{|\pathrho|}$ be the path guaranteed by $\mathrm{\ref{snake:Path}}$ and let $\rmN$ be the integer guaranteed by~$\mathrm{\ref{snake:Len}}$.
  We show by induction that $|\pathrho|$ is divisible by $\rmN$. 
  The inductive assumption states that for all integers $k \in \N$ with $k \cdot \rmN \leq \len{\pathrho}$ we have:
  \begin{enumerate}[(i)]\itemsep0em
    \item $\pathrho_{k \cdot \rmN + 1}$ carries a left-border tile,
    \item There is no $2 \leq i < \rmN$ such that $\pathrho_{k \cdot \rmN + i}$ carries a left-border tile or a right-border tile,
    \item $\pathrho_{k \cdot \rmN + \rmN}$ carries a right-border tile.
  \end{enumerate}
  Then by Property (iii) and the fact that $\pathrho_{|\pathrho|}$ (equal to $\indvsru^{\interI}$ by~$\mathrm{\ref{snake:Path}}$) carries a right-border tile (by Property~(d) of~$\mathrm{\ref{snake:SpecTil}}$), we can conclude that $|\pathrho|$ is indeed divisible by $\rmN$.

  We heavily rely on the fact that every element of $\pathrho$ is labelled by precisely one tile, which is due to~$\mathrm{\ref{snake:UniqTil}}$.
  We start with the case of $k = 0$. 
  Then $\pathrho_{0 \cdot \rmN + 1} = \pathrho_1$ is equal to $\indvsld^{\interI}$, by $\mathrm{\ref{snake:Path}}$. 
  Moreover, $\pathrho_1$ is labelled with a left-border tile, by Property (a) of~$\mathrm{\ref{snake:SpecTil}}$.
  What is more, $\pathrho_{0 \cdot \rmN + \rmN} = \pathrho_\rmN$ is equal to $\indvsrd^{\interI}$ (by~$\mathrm{\ref{snake:Len}}$), which carries a right-border tile by Property (b) of~$\mathrm{\ref{snake:SpecTil}}$.
  This resolves Properties (i) and (iii).
  To establish Property~(ii), assume towards a contradiction that there is $i$ between $2$ and $\rmN$ for which $\pathrho_{i}$ carries a left-border tile (the proof for a right-border tile is analogous). 
  Take the smallest such $i$. 
  By $\mathrm{\ref{snake:Hori}}$ we infer that~$\pathrho_{(i{-}1)}$ carries a right-border tile. 
  In particular, this means that~$i > 2$ because the tile carried by
   $\pathrho_1$ is not right-border.
  By exhaustive application of $\mathrm{\ref{snake:Hori}}$ and the fact that the tile of $\pathrho_1$ is down-border, we deduce that the tile of~$\pathrho_{(i{-}1)}$ is also down-border.
  Hence, by Property (b) of~$\mathrm{\ref{snake:SpecTil}}$ we have that $\pathrho_{(i{-}1)}$ is equal to~$\indvsrd^{\interI}$.
  But then the path $\pathrho_{(i{-}1)}\pathrho_{i}\ldots\pathrho_{\rmN}$ witnesses $\roler^+$-reachability of~$\indvsrd^{\interI}$ from itself, which is forbidden by~$\mathrm{\ref{snake:NoLoop}}$. A~contradiction.
  For the inductive step, assume that Properties (i)--(iii) hold true for some~$k$, and consider the case of $k{+}1$.
  Note that Property (i) follows from Property~(iii) of the inductive assumption by $\mathrm{\ref{snake:Hori}}$.
  We next show that Property (ii) holds. Assume ad absurdum that there is $i$ for which $\pathrho_{(k{+}1)\cdot\rmN + i}$ carries a left-border (resp. right-border) tile. But then, invoking Item (c) of $\mathrm{\ref{snake:Verti}}$, we infer that $\pathrho_{k\cdot\rmN + i}$ is also left-border (resp. right-border). This contradicts Property (ii) of the inductive assumption. Hence Property (ii) holds true.
  By inductive assumption, we know that $\pathrho_{k\cdot\rmN + \rmN}$ carries a right-border tile. 
  Then, we apply Property~(c) of~$\mathrm{\ref{snake:Verti}}$ to infer that $\pathrho_{k\cdot\rmN + \rmN + \rmN} = \pathrho_{(k{+}1)\cdot\rmN + \rmN}$ is right-border, as desired. This establishes Property (iii), and concludes the induction.

  Let $\rmM \deff |\pathrho|/\rmN$.
  By the previous claim, we know that $\rmM \in \N$.
  Consider a function $\tilesmap : \ZZ_\rmN \times \ZZ_\rmM \to \tiles$ that maps all $(x,y)$ to the unique tile carried by $\pathrho_{x + \rmN \cdot y + 1}$. (Note that we number paths from $1$!)
  This function is well-defined by $\mathrm{\ref{snake:UniqTil}}$ and it satisfies $\mathrm{\ref{tiles:Hori}}$ and $\mathrm{\ref{tiles:Verti}}$ due to the satisfaction of $\mathrm{\ref{snake:Hori}}$ and~$\mathrm{\ref{snake:Verti}}$.
  The satisfaction of the first two statements of $\mathrm{\ref{tiles:Borders}}$ by $\tilesmap$ is guaranteed by Properties (i)--(iii) from the induction above.
  Finally, the last two statements of $\mathrm{\ref{tiles:Borders}}$ are due to straightforward induction that employs~$\mathrm{\ref{snake:SpecTil}}$ and the last statement of $\mathrm{\ref{snake:Hori}}$.
  As we proved that $\tilesmap$ covers $\ZZ_\rmN \times \ZZ_\rmM$, we conclude that~$\tilingsys$ is indeed solvable.
\end{proof}

While $\tilingsys$-snakes do not seem to be directly axiomatizable even in $\ALCvpl$, we at least see how to express $\tilingsys$-pseudosnakes in $\ALCOregrhashshash$. The next lemma is routine.
\begin{lem}\label{lemma:expressing-pseudosnakes-in-ALCVPLO}
For every domino tiling system $\tilingsys \deff (\tilesCol, \tiles, \tileswhite)$,
there exists an $\ALCOregrhashshash$-concept $\conceptsnake^\tilingsys$, that employs only the role $\roler$, individual names from~$\namessnakeT$ and concept names from $\concepttilepath^{\tiles}$, such that for all interpretations~$\interI$ we have that $\interI$ is a $\tilingsys$-pseudosnake iff $\interI \models \conceptsnake^\tilingsys$.
\end{lem}
\begin{proof}
  We present a rather straightforward axiomatization of the aforementioned properties, written from the point of view of the interpretation of a nominal $\indvsld$.
  \begin{align*} 
    \conceptC_{\mathrm{\ref{snake:Path}}} & \deff \{ \indvsld\} \dland \exists{\roler^+}.\left(  \{ \indvsrd\} \dland \exists{\roler^+}.\left( \{ \indvslu \} \dland \exists{\roler^+}.\{\indvsru\} \right) \right).\\
    \conceptC_{\mathrm{\ref{snake:NoLoop}}} & \deff \bigdland_{\indva \in \namessnakeT} \forall{\roler^*}.\Big[ \{ \indva \} \to \forall{\roler^+}.\neg\{\indva\} \Big].\\ 
    \conceptC_{\mathrm{\ref{snake:UniqTil}}} & \deff \forall{\roler^*}\big[ \bigdlor_{\conceptC \in \concepttilepath^{\tiles}}\left( \conceptC \dland \bigdland_{\conceptC' \in \concepttilepath^{\tiles}, \conceptC' \neq \conceptC} \neg \conceptC' \right) \Big].\\
    %
    \conceptC_{\mathrm{\ref{snake:SpecTil}}} & \deff \forall{\roler^*} \big[ \left( \bigdlor_{\tile \in \tiles \ \text{with two white sides}} \conceptC_\tile \right) \leftrightarrow \bigdlor_{\indva \in \namessnakeT} \{ \indva \}  \big] \dland \conceptC_{\mathrm{\ref{snake:SpecTil}}}^{\text{down}} \dland \conceptC_{\mathrm{\ref{snake:SpecTil}}}^{\text{up}}, \ \text{where}\\
  \end{align*}
  \begin{align*}
     \conceptC_{\mathrm{\ref{snake:SpecTil}}}^{\text{down}} & \deff \left( \{ \indvsld \} \dland \bigdlor_{\tile \deff (\whiteBox, \whiteBox, \colour_r, \colour_u) \in \tiles } \conceptC_\tile  \right) \dland \exists{\roler^*}.\left( \{ \indvsrd \} \dland \bigdlor_{\tile \deff (\colour_l, \whiteBox, \whiteBox, \colour_u) \in \tiles } \conceptC_\tile  \right),\\
     \conceptC_{\mathrm{\ref{snake:SpecTil}}}^{\text{up}} & \deff \exists{\roler^*}.\left( \{ \indvslu \} \dland \bigdlor_{\tile \deff (\whiteBox, \colour_d, \colour_r, \whiteBox) \in \tiles } \conceptC_\tile  \right) \dland \exists{\roler^*}.\left( \{ \indvsru \} \dland \bigdlor_{\tile \deff (\colour_l, \colour_d, \whiteBox, \whiteBox) \in \tiles } \conceptC_\tile  \right).\\
    \conceptC_{\mathrm{\ref{snake:Hori}}} & \deff \bigdland_{\tile\in\tiles} \forall{\roler^*}.\big[(\neg\{ \indvsru \} \dland \conceptC_\tile) \to \left( (\exists{\roler}.\top) \dland \bigdlor_{\tile'\in\tiles \ \text{satisfying cond. (i)--(iii) of } \mathrm{\ref{snake:Hori}}} \forall{\roler}.\conceptC_{\tile'}\right)\!\big].
  \end{align*}
    We can now define $\conceptsnake^\tilingsys$ as the conjunction of all the concept definitions presented above.
    It~follows immediately from the semantics of $\ALCOregrhashshash$ that the presented concept definition is consistent if and only there exists an element starting a pseudosnake. 
\end{proof}

Note that the property that pseudosnakes are missing in order to be proper snakes, is the ability to measure. 
We tackle this issue by introducing a gadget called a ``yardstick''.

\begin{defi}\label{def:yardstick}
  Let $\tiles$ be a finite non-empty set, let $\namesrulerT \deff \{ \indvyst, \indvmd, \indvmd_\tile, \indvyend_\tile \mid \tile \in \tiles \}$ be composed of (pairwise different) individual names.
  A $\tiles$-\emph{yardstick} is any interpretation $\interI$ satisfying all the conditions below.
  \begin{description}\itemsep0em
    \item[\desclabel{(YNom)}{yardstick:DifNom}] $\namesrulerT$-named elements in $\interI$ are pairwise different and $(\roler+\roles)^*$-reachable from $\indvyst^{\interI}$.
    \item[\desclabel{(YNoLoop)}{yardstick:NoLoop}] No $\namesrulerT$-named element can $(\roler+\roles)^+$-reach itself.
    \item[\desclabel{(YMid)}{yardstick:Mid}] $\indvmd^{\interI}$ is the \emph{unique} element with an $\roles$-successor that is $\roler^*$-reachable from $\indvyst^{\interI}$.
    \item[\desclabel{(YSuccOfMid)}{yardstick:SuccOfMid}] The $\roles$-successors of $\indvmd^{\interI}$ are precisely $\{ \indvmd_\tile \mid \tile \in \tiles \}$-named elements.
    \item[\desclabel{(YReachMidT)}{yardstick:ReachMidT}] For every $\tile \in \tiles$  we have that $\indvmd_\tile^{\interI}$ can $\roles^*$-reach $\indvyend_\tile^{\interI}$ but it cannot $\roles^*$-reach $\indvyend_{\tile'}^{\interI}$ for all~$\tile' \neq \tile$.
    \item[\desclabel{(YEqDst)}{yardstick:EqDist}] The elem. $\langvpaeq{\roler}{\roles}$-reachable from $\indvyst^{\interI}$ are precisely the $\{ \indvyend_\tile \mid \tile \in \tiles \}$-named~ones.
    \item[\desclabel{(YNoEqDst)}{yardstick:NoEqDist}] No $\{ \indvyend_\tile \mid \tile \in \tiles \}$-named element is $\langvpaeq{\roler}{\roles}$-reachable from an element $(\roles+\roler)^+$-reachable~from~$\indvyst^{\interI}$. 
  \end{description}
\end{defi}
\noindent An~example $\{ \heartsuit, \spadesuit \}$-yardstick is depicted below.
A~``minimal'' yardstick contains the grey nodes~only.
\begin{figure}[h]
  \centering
  \begin{tikzpicture}[transform shape]
      \draw (-1.5, 0) node[medrond] (BeforeA) {};
      \draw (0, 0) node[medrond, fill=black!15] (A) {};
      \draw (0, 0.6) node[label=center:\( \indvyst^{\interI} \)] (AA) {\phantom{0}};
      \draw (1.5, 0) node[medrond, fill=black!15] (B) {};
      \draw (1.5, 1) node[medrond] (BB) {};
      \draw (1.5, -1) node[medrond] (BBB) {};
      \draw (3.0, 0) node[medrond, fill=black!15] (C) {};
      \draw (3.0, -1) node[medrond] (CCC) {};
      \draw (4.5, 0) node[medrond, fill=black!15] (D) {};
      \draw (4.5, 1) node[medrond] (DD) {};
      \draw (4.5, -1) node[medrond] (DDD) {};
      \draw (6, 0) node[medrond, fill=black!15] (E) {};
      \draw (6, 0.6) node[label=center:\( \indvmd^{\interI} \)] (EE) {\phantom{0}};

      \draw (7.5, 1) node[medrond, fill=black!15] (XA) {};
      \draw (7.5, 1.6) node[label=center:\( \indvmd_{\heartsuit}^{\interI} \)] (XAA) {\phantom{0}};

      \draw (9.0, 1) node[medrond, fill=black!15] (XB) {};
      \draw (10.5, 1) node[medrond, fill=black!15] (XC) {};
      \draw (12.0, 1) node[medrond, fill=black!15] (XD) {};
      \draw (12.0, 1.6) node[label=center:\( \indvyend_{\heartsuit}^{\interI} \)] (XDD) {\phantom{0}};

      \draw (7.5, -1) node[medrond, fill=black!15] (YA) {};
      \draw (7.5, -0.4) node[label=center:\( \indvmd_{\spadesuit}^{\interI} \)] (YAA) {\phantom{0}};
      \draw (9.0, -1) node[medrond, fill=black!15] (YB) {};
      \draw (9.0, 0) node[medrond] (YBB) {};
      \draw (10.5, -1) node[medrond, fill=black!15] (YC) {};
      \draw (10.5, 0) node[medrond] (YCC) {};
      \draw (12.0, -1) node[medrond, fill=black!15] (YD) {};
      \draw (12.0, -0.4) node[label=center:\( \indvyend_{\spadesuit}^{\interI} \)] (YDD) {\phantom{0}};

      \path[->] (BeforeA) edge [red] node[yshift=3, xshift=-2] {$\roler$} (A);
      \path[->] (A) edge [red] node[yshift=3, xshift=-2] {$\roler$} (B);
      \path[->] (A) edge [red] node[yshift=3, xshift=-2] {$\roler$} (BB);
      \path[->] (A) edge [red] node[yshift=3, xshift=2] {$\roler$} (BBB);
      \path[->] (B) edge [red] node[yshift=3, xshift=-2] {$\roler$} (C);
      \path[->] (BB) edge [red] node[yshift=5, xshift=-2] {$\roler$} (C);
      \path[->] (BBB) edge [red] node[yshift=5, xshift=-2] {$\roler$} (CCC);
      \path[->] (C) edge [red] node[yshift=3, xshift=-2] {$\roler$} (D);
      \path[->] (CCC) edge [red] node[yshift=3, xshift=-2] {$\roler$} (DDD);
      \path[->] (D) edge [red] node[yshift=3, xshift=-2] {$\roler$} (E);
      \path[->] (DDD) edge [red] node[yshift=3, xshift=-2] {$\roler$} (E);
      \path[->] (C) edge [red] node[yshift=3, xshift=-2] {$\roler$} (DD);
      \path[->] (C) edge [red] node[yshift=3, xshift=1] {$\roler$} (DDD);
      \path[->] (DD) edge [red] node[yshift=3, xshift=3] {$\roler$} (E);

      \path[->] (E) edge [blue, dotted] node[yshift=3, xshift=-2] {$\roles$} (XA);
      \path[->] (XA) edge [blue, dotted] node[yshift=3, xshift=-2] {$\roles$} (XB);
      \path[->] (XB) edge [blue, dotted] node[yshift=3, xshift=-2] {$\roles$} (XC);
      \path[->] (XC) edge [blue, dotted] node[yshift=3, xshift=-2] {$\roles$} (XD);

      \path[->] (E) edge [blue, dotted] node[yshift=3, xshift=-2] {$\roles$} (YA);
      \path[->] (YA) edge [blue, dotted] node[yshift=3, xshift=-2] {$\roles$} (YB);
      \path[->] (YA) edge [blue, dotted] node[yshift=3, xshift=-2] {$\roles$} (YBB);
      \path[->] (YB) edge [blue, dotted] node[yshift=3, xshift=-2] {$\roles$} (YC);
      \path[->] (YBB) edge [blue, dotted] node[yshift=3, xshift=-2] {$\roles$} (YCC);
      \path[->] (YC) edge [blue, dotted] node[yshift=3, xshift=-2] {$\roles$} (YD);
      \path[->] (YCC) edge [blue, dotted] node[yshift=3, xshift=1] {$\roles$} (YD);
  \end{tikzpicture}
\end{figure}

The forthcoming lemma explains the name ``yardstick''. Intuitively it says that in any $\tiles$-yardstick $\interI$, all $\roles^*$-paths from $\indvmd^{\interI}$ to all $\indvyend_\tile^{\interI}$ have equal length.
\begin{lem}\label{lemma:yardstickequidist}
Let $\interI$ be a $\tiles$-yardstick.
Then there exists a unique positive integer $\rmN$ such that:
(i) for all~$\tile \in \tiles$ we have that $\indvyend_\tile^{\interI}$ is $\roles^{\rmN}$-reachable from $\indvmd^{\interI}$, and
(ii) for all $\tile \in \tiles$ we have that $\indvyend_\tile^{\interI}$ is $\roles^{\rmN{-}1}$-reachable from~$\indvmd_\tile^{\interI}$.
We will call $\rmN$ the \emph{length} of $\interI$.
\end{lem}
\begin{proof}
  Fix $\tile_\star \in \tiles$. 
  By $\mathrm{\ref{yardstick:EqDist}}$ we know that $\indvyst^{\interI}$ $\langvpaeq{\roler}{\roles}$-reaches $\indvyend_{\tile_\star}^{\interI}$, and let $\pathrho \deff \pathrho_1 \ldots \pathrho_{2\rmN{+}1}$ be a path witnessing it. 
  We claim that this is the desired length of $\interI$.
  First, note that $\rmN > 0$ by $\mathrm{\ref{yardstick:DifNom}}$. 
  Second, by the semantics of $\langvpaeq{\roler}{\roles}$, for all $i \leq \rmN$ we have $(\pathrho_i, \pathrho_{i{+}1}) \in \roler^{\interI}$ and $(\pathrho_{\rmN{+}i}, \pathrho_{\rmN{+}i{+}1}) \in \roles^{\interI}$.
  Thus $\pathrho_{\rmN{+}1}$ is $\roler^*$-reachable from~$\indvyst^{\interI}$ and has an $\roles$-successor.
  These two facts imply (by $\mathrm{\ref{yardstick:Mid}}$) that $\pathrho_{\rmN{+}1}$ is equal to $\indvmd^{\interI}$.
  It remains to show that all the paths leading from $\indvmd^{\interI}$ to some $\indvyend_\tile$ are of length $\rmN$.
  Towards a contradiction, assume that there is $\tile' \in \tiles$ and an integer $\rmM \neq \rmN$ such that $\indvmd^{\interI}$ $\roles^\rmM$-reaches $\indvyend_{\tile'}^{\interI}$ via a path~$\pathrho' \deff \pathrho_1'\ldots\pathrho_\rmM'$.
  We stress that $\pathrho_1' = \indvmd^{\interI}$ and  $\pathrho_\rmM' = \indvyend_{\tile'}^{\interI}$ (by design of $\pathrho'$), and $\pathrho_2' = \indvmd_{\tile'}^{\interI}$ (by a conjunction of $\mathrm{\ref{yardstick:SuccOfMid}}$ and $\mathrm{\ref{yardstick:ReachMidT}}$).
  To conclude the proof, it suffices to resolve the following two cases.
  \begin{itemize}\itemsep0em
  \item Suppose that $\rmM < \rmN$. 
  Then $\pathrho_{\rmN{+}1{-}\rmM}$ $(\roler^{\rmM}\roles^{\rmM})$-reaches (thus also $\langvpaeq{\roler}{\roles}$-reaches) $\indvyend_{\tile'}^{\interI}$, as witnessed by the path $\pathrho_{\rmN{+}1{-}\rmM}\ldots\pathrho_{\rmN}\pathrho'$.
  Moreover $\pathrho_{\rmN{+}1{-}\rmM}$ is $\roler^+$-reachable from $\indvyst^{\interI}$, witnessed by the path $\pathrho_1\ldots\pathrho_{\rmN{+}1{-}\rmM}$ (note that its length is positive by the inequality $\rmM < \rmN$).
  This yields a contradiction with~$\mathrm{\ref{yardstick:NoEqDist}}$.
  \item Suppose that $\rmM > \rmN$. 
  Consider a path $\pathrho_1\ldots\pathrho_\rmN\pathrho_1'\ldots\pathrho_\rmN'$.
  By construction, such a path witnesses the fact that $\indvyst^{\interI}$ $(\roler^{\rmN}\roles^{\rmN})$-reaches (and thus also $\langvpaeq{\roler}{\roles}$-reaches) $\pathrho_\rmN'$.
  By $\mathrm{\ref{yardstick:EqDist}}$ we infer that $\pathrho_\rmN'$ is then $\{ \indvyend_\tile \mid \tile \in \tiles \}$-named.
  As $\pathrho_2' = \indvmd_{\tile'}^{\interI}$ $\roles^+$-reaches $\pathrho_\rmN'$, we infer that $\pathrho_\rmN' = \indvyend_{\tile'}^{\interI}$ (otherwise we would have a contradiction with $\mathrm{\ref{yardstick:ReachMidT}}$).
  But then $\indvyend_{\tile'}^{\interI}$ $\roles^+$-reaches itself via a path $\pathrho_\rmN'\ldots\pathrho_\rmM$, which is of positive length by the fact that $\rmM > \rmN$. 
  This yields a contradiction with $\mathrm{\ref{yardstick:NoLoop}}$. 
  \end{itemize}

  This establishes Property (i). Property (ii) is now immediate by~$\mathrm{\ref{yardstick:SuccOfMid}}$. \qedhere
\end{proof}

The next lemma proves the existence of arbitrary large yardsticks.
\begin{lem}\label{lemma:yardsticks-of-len-N-exist}
For every finite non-empty set $\tiles$ and a positive integer $\rmN$, there exists a $\tiles$-yardstick of length $\rmN$.
\end{lem}
\begin{proof}
Consider the following interpretation $\interI$ with $\DeltaI \deff \ZZ_\rmN \cup \{ \rmN \} \cup (\ZZ_{\rmN{-}1} \times \tiles)$:
\begin{itemize}\itemsep0em
\item $\indvyst^{\interI} \deff 0$, $\indvmd^{\interI} \deff \rmN$, $(\indvmd_\tile)^{\interI} \deff (0, \tile)$, $(\indvyend_\tile)^{\interI} \deff (\rmN{-}1, \tile)$ for all $\tile \in \tiles$, and $\indva^{\interI} \deff 0$ for all names $\indva \in \Ilang \setminus \namesrulerT$.
\item $\roler^{\interI} \deff \{ (i,i{+}1) \mid i \in \ZZ_{\rmN} \}$,
$\roles^{\interI} \deff \{ (\rmN, (0, \tile)), ((i,\tile),(i{+}1, \tile)) \mid \tile \in \tiles, i \in \ZZ_{\rmN{-}1} \}$, and $(\roler')^{\interI} = \emptyset$ for all role names $\roler' \in \Rlang \setminus \{ \roler, \roles\}$.
\end{itemize}
An example such $\interI$ for $\rmN = 3$ and $\tiles = \{ \heartsuit, \spadesuit \}$ is depicted above (in restriction to grey nodes only).
It is routine to verify that~$\interI$ satisfies all the properties in Definition~\ref{def:yardstick}, as well as conditions (i) and~(ii) from the statement of Lemma~\ref{lemma:yardstickequidist}.
This concludes the proof.
\end{proof}

To make use of yardsticks in our proofs, we need to axiomatise them inside $\ALCOregrhashshash$.
\begin{lem}\label{lemma:expressing-yardstick-in-ALCVPLO}
There exists an $\ALCOregrhashshash$-concept $\conceptyardstick^\tiles$, that employs only role names $\roler, \roles$ and individual names from~$\namesrulerT$, such that for all interpretations $\interI$ we have: $\interI$ is a $\tiles$-yardstick if and only if $\interI$ is a model~of~$\conceptyardstick^\tiles$.
\end{lem}
\begin{proof}
  The following concepts are written from the point of view of the interpretation~of~$\indvyst$.
 \begin{align*} 
  \conceptC_{\mathrm{\ref{yardstick:DifNom}}} & \deff \bigdland_{\indva \in \namesrulerT} \exists{(\roler+\roles)^*} \Big[ \{ \indva \}  \dland \bigdland_{\indvb \in \namesrulerT \setminus \{ \indva \} } \neg \{ \indvb \} \Big].\\
\end{align*}
\begin{align*}
  \conceptC_{\mathrm{\ref{yardstick:NoLoop}}} & \deff \bigdland_{\indva \in \namesrulerT} \forall{(\roler+\roles)^*}.\Big[ \{ \indva\} \to \forall{(\roler+\roles)^+}.\neg\{\indva\} \Big].\\ 
  \conceptC_{\mathrm{\ref{yardstick:Mid}}} & \deff 
  \Big[ \exists{\roler^*}.\left( (\exists{\roles}.\top) \dland \{ \indvmd \} \right)  \Big] 
  \dland
  \Big[ \forall{\roler^*}.\left( (\exists{\roles}.\top) \to \{ \indvmd \} \right) \Big].\\ 
  \conceptC_{\mathrm{\ref{yardstick:SuccOfMid}}} & \deff \forall{(\roler + \roles)^*}.
  \left(
  \{ \indvmd \} \to \Big[
  \bigdland_{\indva \in \{ \indvmd_\tile \mid \tile \in \tiles \}} \exists{\roles}.\{ \indva \}
  \Big] \dland 
  \Big[
  \forall{\roles}.\bigdlor_{\indva \in \{ \indvmd_\tile \mid \tile \in \tiles \}}\{ \indva \}
  \Big]
  \right).\\
  \conceptC_{\mathrm{\ref{yardstick:ReachMidT}}} & \deff \bigdland_{\tile \in \tiles} \forall{(\roler + \roles)^*}. \{ \indvmd_\tile \} \to
  \Big[ \left( \exists{\roles^*}.\{ \indvyend_\tile \} \right) \dland \bigdland_{\tile' \in \tiles, \tile \neq \tile'} \forall{\roles^*}.\neg\{ \indvyend_{\tile'} \}  \Big].\\
    \conceptC_{\mathrm{\ref{yardstick:EqDist}}} & \deff 
    \left( \bigdland_{\tile \in \tiles} \exists{\langvpaeq{\roler}{\roles}}.\{\indvyend_\tile\} \right)
    \dland \forall{\langvpaeq{\roler}{\roles}}.\left( \bigdlor_{\tile \in \tiles} \{\indvyend_\tile\} \right).\\
    \conceptC_{\mathrm{\ref{yardstick:NoEqDist}}} & \deff \forall{(\roler+\roles)^+}.\forall{\langvpaeq{\roler}{\roles}}.\left( \bigdland_{\tile \in \tiles} \neg \{\indvyend_\tile\} \right).\\
  \end{align*} 
    We define $\conceptyardstick^{\tiles}$ as the conjunction of $\{ \indvyst \}$ and all the concept definitions presented above.
    By semantics of $\ALCOregrhashshash$ we have that $\conceptyardstick^{\tiles}$ is consistent if and only if~$(\conceptyardstick^\tiles)^\interI = \{ \indvyst^{\interI} \}$.
\end{proof}

We next put pseudosnakes and yardsticks together, obtaining metricobras.
The intuition behind their construction is fairly simple: (i) we take a disjoint union of a pseudosnake and a yardstick, (ii) we then connect (via the role $\roles$) every element carrying a tile~$\tile$ with the interpretation of the corresponding nominal $\indvmd_\tile$, and finally (iii) we synchronise the length of the underlying yardstick, say~$\rmN$, with the length of the path between the interpretations of $\indvsld$ and $\indvsrd$. 
After such ``merging'', retrieving $\mathrm{\ref{snake:Verti}}$ is easy: rather than testing if every $\rmN$-reachable element from some $\domelemd$ carries a suitable tile $\tile$ (for an a priori unknown $\rmN$) we can check instead whether $\domelemd$ can $\langvpaeq{\roler}{\roles}$-reach the interpretation of $\indvyend_\tile$.

A formal definition and a picture come next.

\begin{figure}[h]
  \centering
  \begin{tikzpicture}[transform shape]
      \draw (0, 0) node[medrond] (A1) {};
      \draw (1.5, 0) node[medrond] (A2) {};
      \draw (3, 0) node[medrond] (A3) {};
      \draw (4.5, 0) node[medrond] (A4) {};

      \draw (0, 1.5) node[medrond] (B1) {};
      \draw (1.5, 1.5) node[medrond] (B2) {};
      \draw (3, 1.5) node[medrond] (B3) {};
      \draw (4.5, 1.5) node[medrond] (B4) {};

      \draw (6, 1.5) node[medrond] (C1) {};
      \draw (7.5, 1.5) node[medrond] (C2) {};
      \draw (9, 1.5) node[medrond] (C3) {};
      \draw (10.5, 1.5) node[medrond] (C4) {};

      \path[->] (A1) edge [red] node[yshift=3] {$\roler$} (A2);
      \path[->] (A2) edge [red] node[yshift=3] {$\roler$} (A3);
      \path[->] (A3) edge [red] node[yshift=3, xshift=-8] {$\roler$} (A4);
      \path[->] (A4) edge [red] node[yshift=3] {$\roler$} (B1);
      \path[->] (B1) edge [red] node[yshift=3] {$\roler$} (B2);
      \path[->] (B2) edge [red] node[yshift=3] {$\roler$} (B3);
      \path[->] (B3) edge [red] node[yshift=3, xshift=-8] {$\roler$} (B4);
      \path[->] (B4) edge [red] node[yshift=3] {$\roler$} (C1);
      \path[->] (C1) edge [red] node[yshift=3] {$\roler$} (C2);
      \path[->] (C2) edge [red] node[yshift=3] {$\roler$} (C3);
      \path[->] (C3) edge [red] node[yshift=3] {$\roler$} (C4);

       \smallwang{-0.125}{-0.125}{white}{white}{cyan}{rouge} 
       \smallwang{1.5-0.125}{-0.125}{cyan}{white}{lime}{rouge} 
       \smallwang{3-0.125}{-0.125}{lime}{white}{rouge}{rouge} 
       \smallwang{4.5-0.125}{-0.125}{rouge}{white}{white}{cyan}

       \smallwang{-0.125}{1.5-0.125}{white}{rouge}{cyan}{lime} 
       \smallwang{1.5-0.125}{1.5-0.125}{cyan}{rouge}{cyan}{cyan} 
       \smallwang{3-0.125}{1.5-0.125}{cyan}{rouge}{rouge}{cyan} 
       \smallwang{4.5-0.125}{1.5-0.125}{rouge}{cyan}{white}{lime} 

       \smallwang{6-0.125}{1.5-0.125}{white}{lime}{cyan}{white} 
       \smallwang{7.5-0.125}{1.5-0.125}{cyan}{cyan}{cyan}{white} 
       \smallwang{9-0.125}{1.5-0.125}{cyan}{cyan}{lime}{white} 
       \smallwang{10.5-0.125}{1.5-0.125}{lime}{lime}{white}{white}

      \draw (0, -0.5) node[] (XXX) {};
      \node[above=0.1em of A1] {$\indvsld^{\interI}$};
      \node[above=0.1em of A4] {$\indvsrd^{\interI}$};
      \node[above=0.1em of C4] {$\indvsru^{\interI}$}; 
      \node[above=0.1em of C1] {$\indvslu^{\interI}$}; 

      \draw (1.5, -1.25) node[medrond] (Cobra) {};
      \node[right=0.1em of Cobra] {$\indvcbra^{\interI}$}; 
      \path[->] (Cobra) edge [red] node[yshift=3] {$\roler$} (A1);

      \draw (0, -2.5) node[medrond] (Y1) {};
      \node[above=0.1em of Y1] {$\indvyst^{\interI}$};
      \path[->] (Cobra) edge [blue, dotted] node[yshift=3, xshift=1] {$\roles$} (Y1);

      \draw (1.5, -2.5) node[medrond] (Y2) {};
      \path[->] (Y1) edge [red] node[yshift=3] {$\roler$} (Y2);

      \draw (3, -2.5) node[medrond] (Y3) {};
      \path[->] (Y2) edge [red] node[yshift=3] {$\roler$} (Y3);

      \draw (4.5, -2.5) node[medrond] (Y4) {};
      \path[->] (Y3) edge [red] node[yshift=3] {$\roler$} (Y4);

      \draw (4.5, -2.5) node[medrond] (Y4) {};
      \path[->] (Y3) edge [red] node[yshift=3] {$\roler$} (Y4);

      \draw (6, -2.5) node[medrond] (Y5) {};
      \node[above=0.1em of Y5] {$\indvmd^{\interI}$};
      \path[->] (Y4) edge [red] node[yshift=3] {$\roler$} (Y5);

      \draw (7.4, -1.5) node[medrond] (YA1) {};
      \node[above=0.1em of YA1] {$\indvmd_{\intextwang{rouge}{white}{white}{cyan}}^{\interI}$};
      \path[->] (Y5) edge [blue, dotted] node[yshift=3, xshift=1] {$\roles$} (YA1);
      \draw (8.6, -1.5) node[medrond] (YA2) {};
      \path[->] (YA1) edge [blue, dotted] node[yshift=3, xshift=1] {$\roles$} (YA2);
      \draw (9.8, -1.5) node[medrond] (YA3) {};
      \path[->] (YA2) edge [blue, dotted] node[yshift=3, xshift=1] {$\roles$} (YA3);
      \draw (11, -1.5) node[medrond] (YA4) {};
      \node[above=0.1em of YA4] {$\indvyend_{\intextwang{rouge}{white}{white}{cyan}}^{\interI}$};
      \path[->] (YA3) edge [blue, dotted] node[yshift=3, xshift=1] {$\roles$} (YA4);

      \node[below=0.1em of A1] {$\ldots$};
      \node[below=0.1em of A2] {$\ldots$};
      \node[below=0.1em of B1] {$\ldots$};
      \node[below=0.1em of B2] {$\ldots$};
      \node[below=0.1em of B3] {$\ldots$};
      \node[below=0.1em of B4] {$\ldots$};
      \node[below=0.1em of C1] {$\ldots$};
      \node[below=0.1em of C2] {$\ldots$};
      \node[below=0.1em of C3] {$\ldots$};
      \node[below=0.1em of C4] {$\ldots$};

      \path[->] (A4) edge [blue, dotted] node[yshift=3, xshift=1] {$\roles$} (YA1);

      \draw (7.4, -3) node[medrond] (YB1) {};
      \node[above=0.1em of YB1] {$\indvmd_{\intextwang{lime}{white}{rouge}{rouge} }^{\interI}$};
      \path[->] (Y5) edge [blue, dotted] node[yshift=3, xshift=1] {$\roles$} (YB1);
      \draw (8.6, -3) node[medrond] (YB2) {};
      \path[->] (YB1) edge [blue, dotted] node[yshift=3, xshift=1] {$\roles$} (YB2);
      \draw (9.8, -3) node[medrond] (YB3) {};
      \path[->] (YB2) edge [blue, dotted] node[yshift=3, xshift=1] {$\roles$} (YB3);
      \draw (11, -3) node[medrond] (YB4) {};
      \node[above=0.1em of YB4] {$\indvyend_{\intextwang{lime}{white}{rouge}{rouge} }^{\interI}$};
      \path[->] (YB3) edge [blue, dotted] node[yshift=3, xshift=1] {$\roles$} (YB4);

      \path[->] (Y5) edge [blue, dotted] node[yshift=3, xshift=1] {$\roles$} (6.75, -3.25);      
      \path[->] (Y5) edge [blue, dotted] node[yshift=3, xshift=1] {$\roles$} (6.25, -3.25);      
      \path[->] (Y5) edge [blue, dotted] node[yshift=3, xshift=1] {$\roles$} (5.75, -3.25);
      \path[->] (Y5) edge [blue, dotted] node[yshift=3, xshift=1] {$\roles$} (5.25, -3.25);
      \path[->] (A3) edge [blue, dotted] node[] {$\roles$} (YB1);

  \end{tikzpicture}
  \caption{A fragment of an example $\tilingsys$-metricobra representing $\tilesmap$ from Example~\ref{ex:visualization-of-tilesmap}.
  The upper part corresponds to a $\tilingsys$-snake, and the lower part corresponds to a $\tiles$-yardstick. 
  The distances between named~elements are important.}
  \label{fig:an-example-cobra}
\end{figure}

\begin{defi}
Let $\tilingsys \deff (\tilesCol, \tiles, \tileswhite)$ be a domino tiling system  and $\indvcbra$ be an individual name.
An interpretation~$\interI$ is a $\tilingsys$-\emph{metricobra} if all the conditions below are satisfied: 
  \begin{description}\itemsep0em
    \item[\desclabel{(MInit)}{metricobra:Init}] $\interI$ is a $\tilingsys$-pseudosnake and a $\tiles$-yardstick, and $\indvcbra^{\interI}$ has precisely two successors: one $\roler$-successor, namely $\indvsld^{\interI}$, and one $\roles$-successor, namely~$\indvyst^{\interI}$.
    \item[\desclabel{(MTile)}{metricobra:Tiles}] For every tile $\tile \in \tiles$ and every element $\domelemd \in \DeltaI$ that is $\roler^*$-reachable from $\indvsld^{\interI}$ we have that $\domelemd$ carries a tile $\tile \in \tiles$ if and only if $\domelemd$ has a unique $\roles$-successor and such a successor is equal to $\indvmd_\tile^{\interI}$.
    \item[\desclabel{(MSync)}{metricobra:Sync}] Let $\tile \in \tiles$ be the tile labelling $\indvsrd^{\interI}$.
    Then (a) $\indvcbra^{\interI}$ $\langvpaeq{\roler}{\roles}$-reaches $\indvyend_\tile^{\interI}$ and cannot $\langvpaeq{\roler}{\roles}$-reach any of~$\indvyend_{\tile'}^{\interI}$ for $\tile' \neq \tile$,
    (b) $\indvcbra^{\interI}$ cannot $\langvpaeq{\roler}{\roles}$-reach an element that can $\roles^+$-reach $\indvyend_\tile^{\interI}$,
    (c) no element $\roler^*$-reachable from $\indvsld^{\interI}$ can $\langvpaeq{\roler}{\roles}$-reach~$\indvyend_\tile^{\interI}$.  
    \item[\desclabel{(MVerti)}{metricobra:Verti}] For all elements $\domelemd$ that are $\roler^*$-reachable from $\indvsld^{\interI}$ and are labelled by some $\tile \in \tiles$ that is not up-border, we have that there exists a tile $\tile' \in \tiles$ such that (a) $\tile$ and $\tile'$ are $\rmV$-compatible, (b) $\tile$ is left-border (resp. right-border) iff $\tile'$ is, and (c) $\domelemd$ can $\langvpaeq{\roler}{\roles}$-reach $\indvyend_{\tile'}$ but cannot reach $\langvpaeq{\roler}{\roles}$-reach $\indvyend_{\tile''}$ for all~$\tile'' \neq \tile'$.
  \end{description}
\end{defi}

We first show that $\tilingsys$-metricobras are axiomatizable in $\ALCOregrhashshash$.
\begin{lem}\label{lemma:expressing-metricobras-in-ALCVPLO}
There exists an $\ALCOregrhashshash$-concept $\conceptcobra^\tilingsys$ such that for all interpretations $\interI$ we have that $\interI$ is a $\tilingsys$-metricobra if and only if $(\conceptcobra^\tilingsys)^\interI = \{ \indvcbra^{\interI} \}$.
\end{lem}
\begin{proof}
We present a rather straightforward axiomatization of the aforementioned properties, written from the point of view of the interpretation of $\indvcbra$.
Note that as the interpretation of $\indvcbra$ has a unique $\roler$-successor, we can simplify concepts of the form $\forall{\roler}.[\{ \indvsld \} \to (\forall{\roler^*}.\conceptC)]$ to $\forall{\roler^+}.\conceptC$, which we frequently do below.
  \begin{align*} 
    \conceptC_{\mathrm{\ref{metricobra:Init}}} & \deff 
    \exists{\roler}.\conceptsnake^\tilingsys \dland \exists{\roles}.\conceptyardstick^\tiles \dland
    \forall{\roler}.(\{\indvsld \} \dland \neg \{\indvyst\}) 
    \dland \forall{\roles}.(\neg \{ \indvsld \} \dland \{ \indvyst \}). \\
    \conceptC_{\mathrm{\ref{metricobra:Tiles}}} & \deff \bigdland_{\tile \in \tiles} \forall{\roler^+}.\left( \conceptC_\tile \leftrightarrow  [\exists{\roles}.\{\indvmd_{\tile}\} \dland \forall{\roles}.\{ \indvmd_{\tile}\}] \right). \\
    \conceptC_{\mathrm{\ref{metricobra:Verti}}} & \deff \forall{\roler^+}. \bigdland_{\tile \in \tiles}^{\text{not up-border}} \left( \conceptC_\tile \to \hspace{-1.5em}\bigdlor_{\tile' \in \tiles}^{\text{sat.\ (a),(b) of } \mathrm{\ref{metricobra:Verti}}} \hspace{-3em}\big[ (\exists{\langvpaeq{\roler}{\roles}}.\{ \indvyend_{\tile'} \}) \dland \hspace{-1em}\bigdland_{\tile'' \neq \tile', \tile'' \in \tiles} \hspace{-0.5em}\forall{\langvpaeq{\roler}{\roles}}.\neg\{ \indvyend_{\tile''} \} \big] \right). \\
  \end{align*}  
  \[
    \!\conceptC_{\mathrm{\ref{metricobra:Sync}}}^\tile\!\deff\!(\exists{\langvpaeq{\roler}{\roles}}.\{\indvyend_\tile\}\hspace*{-1pt}) \dland \bigdland_{\tile' \neq \tile} (\forall{\langvpaeq{\roler}{\roles}}.\neg\{\indvyend_{\tile'}\}\hspace*{-1pt}) \dland (\forall{\langvpaeq{\roler}{\roles}}.\forall{\roles^+}.\neg\{\indvyend_\tile\}\hspace*{-1pt}) \dland \forall{\roler^+}.\forall{\langvpaeq{\roler}{\roles}}.\neg\{\indvyend_\tile\}.
  \]
      We define $\conceptcobra^\tilingsys \deff
        \{ \indvcbra \} \dland \conceptC_{\mathrm{\ref{metricobra:Init}}} \dland \conceptC_{\mathrm{\ref{metricobra:Tiles}}} \dland \conceptC_{\mathrm{\ref{metricobra:Verti}}} \dland \bigdland_{\tile \in \tiles} [\left( \exists{\roler^+}.( \{ \indvsrd \} \dland \conceptC_\tile ) \right) \to\conceptC_{\mathrm{\ref{metricobra:Sync}}}^\tile].
      $
      Its correctness follows from the semantics of $\ALCOregrhashshash$ and Lemmas~\ref{lemma:expressing-pseudosnakes-in-ALCVPLO} and~\ref{lemma:expressing-yardstick-in-ALCVPLO}.
\end{proof}

As the second step, we show that for every $\tilingsys$-snake we can construct a $\tilingsys$-metricobra.
\begin{lem}\label{lemma:from-snake-to-cobra}
  If $\interI$ is a $\tilingsys$-snake then there exists a $\tilingsys$-metricobra $\interJ$.
  \end{lem}
  \begin{proof}
  Let $\interI$ be a $\tilingsys$-snake, $\rmN$ be the integer guaranteed by $\mathrm{\ref{snake:Len}}$, and $\interI'$ be any $\tiles$-yardstick of length $\rmN$ (existence guaranteed by Lemma~\ref{lemma:yardsticks-of-len-N-exist}).
  We construct an interpretation $\interJ$ as a disjoint union of $\interI$, $\interI'$ and an additional domain element that we interpret as $\indvcbra^{\interJ}$.
  We let $\interJ$ interpret all names from $\namesrulerT$ as in $\interI'$, all names from $\namessnakeT$ as in $\interI$, and all other (unused in our concept definitions) names as $\indvcbra^{\interJ}$.
  Interpretation of concept names is inherited from $\interI$ and $\interJ$.
  Finally we interpret role names as in $\interI$ and $\interI'$ with minor corrections.
  More precisely:
  (i) we alter the interpretation of $\roler^\interJ$ to include an extra pair $(\indvcbra^{\interJ}, \indvsld^{\interJ})$, and 
  (ii) we alter the interpretation of $\roles^\interJ$ to include the pair $(\indvcbra^{\interJ}, \indvyst^{\interJ})$ and $\bigcup_{\tile \in \tiles} \conceptC_\tile^{\interJ} \times \{ \indvmd_\tile^{\interJ} \}$. 
  Consult Figure~\ref{fig:an-example-cobra} to see an example construction of $\interJ$ from $\interI$ and~$\interI'$.
  The satisfaction of properties~$\mathrm{\ref{metricobra:Init}}$ and $\mathrm{\ref{metricobra:Tiles}}$ follow  from the construction of $\interJ$, while the other two properties are due to, respectively, $\mathrm{\ref{snake:Len}}$ and $\mathrm{\ref{snake:Verti}}$.
  Thus $\interJ$ is the desired metricobra.
\end{proof}

Finally, we show that every $\tilingsys$-metricobra is actually a $\tilingsys$-snake.
\begin{lem}\label{lemma:from-cobra-to-snake}
If $\interI$ is a $\tilingsys$-metricobra then it is also a $\tilingsys$-snake.
\end{lem}
\begin{proof}
  As $\interI$ is a pseudosnake by definition, it suffices to show that it satisfies the missing conditions of Definition~\ref{def:snake}.
  Our first goal is to establish $\interI \models$ $\mathrm{\ref{snake:Len}}$.
  Note that by $\mathrm{\ref{snake:SpecTil}}$ we have that $\indvsrd^{\interI}$ is the only element $\roler^*$-reachable from $\indvsld^{\interI}$ that carries a down- and right-border tile, say $\tile$.
  Furthermore by $\mathrm{\ref{metricobra:Tiles}}$ and $\mathrm{\ref{yardstick:ReachMidT}}$ we infer that $\indvsrd^{\interI}$ is the only element $\roler^*$-reachable from $\indvsld^{\interI}$ that can $\roles^*$-reach $\indvyend_\tile^{\interI}$.
  Take $\rmN$ to be the length of the yardstick.
  By Lemma~\ref{lemma:yardstickequidist} we know that $\indvmd_\tile^{\interI}$ $\roles^{\rmN{-}1}$-reaches~$\indvyend_\tile^{\interI}$, hence $\indvsrd^{\interI}$ $\roles^{\rmN}$-reaches~$\indvyend_\tile^{\interI}$ (and there is no other integer $\rmM \neq \rmN$ for which such reachability conditions hold).
  From $\mathrm{\ref{metricobra:Sync}}$ we know that~$\indvcbra^{\interI}$ $\langvpaeq{\roler}{\roles}$-reaches~$\indvyend_\tile^{\interI}$, thus by previous observations we deduce that~$\indvcbra^{\interI}$ $\roler^{\rmN}\roles^{\rmN}$-reaches~$\indvyend_\tile^{\interI}$ (whence~$\indvsld^{\interI}$ $\roler^{\rmN{-}1}$-reaches $\indvsrd^{\interI}$).
  This establishes the existence of a path of length $\rmN$ mentioned in $\mathrm{\ref{snake:Len}}$, and we next need to show that all such paths have equal length.
  Consider~cases:
  \begin{itemize}
    \item There exists $\rmM < \rmN$ such that $\indvsld^{\interI}$ $\roler^{\rmM{-}1}$-reaches $\indvsrd^{\interI}$.
    Then $\indvcbra^{\interI}$ $\roler^{\rmM}\roles^{\rmM}$-reaches some element that can $\roles^+$-reach $\indvyend_\tile^{\interI}$. This yields a contradiction with condition~(b) of $\mathrm{\ref{metricobra:Sync}}$.
    \item There exists $\rmM > \rmN$ such that $\indvsld^{\interI}$ $\roler^{\rmM{-}1}$-reaches $\indvsrd^{\interI}$.
    Then there is an element $\roler^*$-reachable from $\indvsld^{\interI}$ that can $\roler^{\rmN}\roles^{\rmN}$-reach $\indvyend_\tile^{\interI}$. This contradictions condition~(c) of $\mathrm{\ref{metricobra:Sync}}$.
  \end{itemize}
  Hence $\rmN$ is indeed unique. 
  The fact that $\indvsrd^{\interI}$ is the only element $\roler^{\rmN{-}1}$-reachable from $\indvsld^{\interI}$ follows from the uniqueness of the tile assigned to~$\indvsrd^{\interI}$, see~$\mathrm{\ref{snake:SpecTil}}$ and Property (a) of~$\mathrm{\ref{metricobra:Sync}}$. 
  It remains now to show that~$\interI \models$ $\mathrm{\ref{snake:Verti}}$. To do so, it suffices to observe that the following property holds. 
  As all the $\roles^*$-paths from an element carrying a tile $\tile$ to $\indvyend_\tile$ are of length $\rmN$ (by the previous discussion and Lemma~\ref{lemma:yardstickequidist}), we can see that $\langvpaeq{\roler}{\roles}$-reachability of $\indvyend_\tile$ is equivalent to $\roler^{\rmN}$-reachability of some element carrying $\tile$.
  Then the satisfaction of~$\mathrm{\ref{snake:Verti}}$ follows immediately by $\mathrm{\ref{metricobra:Verti}}$.
  Hence, $\interI$ is indeed a $\tilingsys$-snake.
\end{proof}

By collecting all previous lemmas we establish the correspondence between solvability of tiling systems and satisfiability of $\conceptcobra^\tilingsys$.
\begin{lem}
A domino tiling system $\tilingsys$ is solvable if and only if $\conceptcobra^\tilingsys$ has a model. 
\end{lem}
\begin{proof}
  If $\tilingsys$ is solvable, then by Lemma~\ref{lemma:from-tiling-to-snake} there exists a $\tilingsys$-snake, which by Lemma~\ref{lemma:from-snake-to-cobra} implies the existence of a $\tilingsys$-metricobra, which is a model of $\conceptcobra^\tilingsys$ (see Lemma~\ref{lemma:expressing-metricobras-in-ALCVPLO}).
  For the other direction, if $\conceptcobra^\tilingsys$ has a model, then such a model is a $\tilingsys$-metricobra (by Lemma~\ref{lemma:expressing-metricobras-in-ALCVPLO}), as well as a~$\tilingsys$-snake (by Lemma~\ref{lemma:from-cobra-to-snake}). 
  As the existence of a~$\tilingsys$-snake guarantees that $\tilingsys$ is solvable by Lemma~\ref{lemma:from-snakes-to-tilling}, this finishes the proof.
\end{proof}

Hence, we can conclude the main theorem of the paper.
\begin{thm}\label{thm:ALCOvpl-is-undecidable}
The concept satisfiability problem for $\ALCOvpl$ is undecidable, even if the languages allowable in concepts are restricted to $\{ \roler, \roles, \roler^+, \roles^+, \roler^*, \roles^*, (\roler + \roles)^*, (\roler + \roles)^+, \langvpaeq{\roler}{\roles} \}$.
\end{thm}

The logic $\ALCOreg$ is a notational variant of Propositional Dynamic Logic with nominals~\cite{KaminskiS14}. 
Thus the above theorem provides results also in the realm of formal verification. 

%% file: sections/querying.tex

\section{Negative results III: Entailment of queries with non-regular atoms}\label{sec:querying-negative}

We conclude the negative part of the paper by showing that positive results regarding entailment of conjunctive queries with visibly-pushdown atoms in the database setting~\cite[Thm.~2]{LangeL15} do not generalise even to $\ALC$ ontologies. 
We provide a reduction from the \emph{White-bordered Octant Tiling Problem}, which we are going to define next.
Roughly speaking, the ontology used in our reduction will define a ``grid'' covered with tiles, while the query counterpart will serve as a tool to detect mismatches in its lower triangle (a.k.a. octant).

\begin{figure}[h]
  \centering
  \begin{tikzpicture}[transform shape]
      \draw (0, 0) node[medrond] (A0) {};
      \draw (1.5, 0) node[medrond] (A1) {};
      \draw (3, 0) node[medrond] (A2) {};
      \draw (4.5, 0) node[medrond] (A3) {};
      \draw (6, 0) node[medrond] (A4) {};
      \draw (1.5, 1) node[medrond] (B1) {};
      \draw (3, 1) node[medrond] (B2) {};
      \draw (4.5, 1) node[medrond] (B3) {};
      \draw (6, 1) node[medrond] (B4) {};
      \draw (3, 2) node[medrond] (C2) {};
      \draw (4.5, 2) node[medrond] (C3) {};
      \draw (6, 2) node[medrond] (C4) {};
      \draw (4.5, 3) node[medrond] (D3) {};
      \draw (6, 3) node[medrond] (D4) {};
      \draw (6, 4) node[medrond] (E4) {};

      \path[->] (A0) edge [red] node[yshift=3] {$\roler$} (A1);
      \path[->] (A1) edge [red] node[yshift=3] {$\roler$} (A2);
      \path[->] (A2) edge [red] node[yshift=3] {$\roler$} (A3);
      \path[->] (A3) edge [red] node[yshift=3] {$\roler$} (A4);

      \path[->] (A1) edge [blue, dotted] node {$\roles$} (B1);
      \path[->] (A2) edge [blue, dotted] node {$\roles$} (B2);
      \path[->] (B2) edge [blue, dotted] node {$\roles$} (C2);

      \path[->] (A3) edge [blue, dotted] node {$\roles$} (B3);
      \path[->] (B3) edge [blue, dotted] node {$\roles$} (C3);
      \path[->] (C3) edge [blue, dotted] node {$\roles$} (D3);

      \path[->] (A4) edge [blue, dotted] node {$\roles$} (B4);
      \path[->] (B4) edge [blue, dotted] node {$\roles$} (C4);
      \path[->] (C4) edge [blue, dotted] node {$\roles$} (D4);
      \path[->] (D4) edge [blue, dotted] node {$\roles$} (E4);

     \smallwang{-0.125}{-0.125}{white}{white}{white}{white} 
     \smallwang{1.5-0.125}{-0.125}{white}{rouge}{cyan}{white}  
     \smallwang{3-0.125}{-0.125}{cyan}{rouge}{rouge}{cyan} 
     \smallwang{4.5-0.125}{-0.125}{rouge}{cyan}{cyan}{lime}  
     \smallwang{6-0.125}{-0.125}{cyan}{cyan}{cyan}{lime}  
     \smallwang{1.5-0.125}{1-0.125}{white}{white}{white}{white} 
     \smallwang{3-0.125}{1-0.125}{white}{cyan}{rouge}{white} 
     \smallwang{4.5-0.125}{1-0.125}{rouge}{lime}{cyan}{lime} 
     \smallwang{6-0.125}{1-0.125}{cyan}{lime}{cyan}{lime} 
     \smallwang{3-0.125}{2-0.125}{white}{white}{white}{white} 
     \smallwang{4.5-0.125}{2-0.125}{white}{lime}{rouge}{white} 
     \smallwang{6-0.125}{2-0.125}{rouge}{lime}{cyan}{lime} 
     \smallwang{4.5-0.125}{3-0.125}{white}{white}{white}{white} 
     \smallwang{6-0.125}{3-0.125}{white}{lime}{rouge}{white} 
     \smallwang{6-0.125}{4-0.125}{white}{white}{white}{white} 

     \draw (6.75, 0) node[] (XXX) {$\ldots$};
     \draw (6.75, 1) node[] (XXX) {$\ldots$};
     \draw (6.75, 2) node[] (XXX) {$\ldots$};
     \draw (6.75, 3) node[] (XXX) {$\ldots$};
     \draw (6.75, 4) node[] (XXX) {$\ldots$};
     \draw (6.75, 5) node[] (XXX) {$\ldots$};

  \end{tikzpicture}
    \caption{Visualization of a fragment of a $\tiles$-octant interpretation.}
    \label{fig:octant-structure}
\end{figure}

We refer to the set $\octant \deff \{ (n,m) \mid n,m \in \N, 0 \leq m \leq n \}$ as \emph{the octant}.
Let~$\tilingsys \deff (\tilesCol, \tiles, \tileswhite)$ be a domino tiling system (defined as in Section~\ref{sec:nominals}).
For our reduction it is convenient to impose several restrictions on $\tiles$, namely that the all-white tile $\intextwang{white}{white}{white}{white}$ belongs to~$\tiles$, and that all other tiles from $\tiles$ containing white colour are both left- and up-border but neither down- nor right-border.
We say that~$\tilingsys$ \emph{covers} $\octant$ if there exists a mapping $\tilesmap \colon \octant \to \tiles$ such that for all pairs $(n,m) \in \octant$ the following conditions are satisfied:
\begin{description}\itemsep0em
  \item[\desclabel{(OInit)}{octtiles:Init}] $\tilesmap(0,0) = \intextwang{white}{white}{white}{white}$ and $\tilesmap(1,0) \neq \intextwang{white}{white}{white}{white}$.
  \item[\desclabel{(OVerti)}{octtiles:Verti}] If $(n,m{+}1) \in \octant$ then $\tilesmap(n,m)$ and $\tilesmap(n, m{+}1)$ are $\rmV$-compatible.
  \item[\desclabel{(OHori)}{octtiles:Hori}] The tiles $\tilesmap(n,m)$ and $\tilesmap(n{+}1, m)$ are $\rmH$-compatible. 
\end{description}

In the \emph{White-bordered Octant Tiling Problem} we ask if an input domino tiling system~$\tilingsys$ (with additional conditions on tiles mentioned above) covers the octant $\octant$. 
\begin{obs}\label{obs:octant-tiling-proper-white}
  Let $\tilingsys$ be a tiling system and let $\tilesmap \colon \octant \to \tiles$ be covering the octant.
  Then for all $i \in \N$ we have that $\tilesmap(i,i) = \intextwang{white}{white}{white}{white}$, and $\tilesmap(i{+}1,i)$ is left- and up-bordered.
  Moreover, no position $(i,j)$ satisfying $0 < j < i{-}1$ carries a white-border tile.
\end{obs}
\begin{proof}
  The first statement follows by routine induction. 
  From $\mathrm{\ref{octtiles:Init}}$ we have $\tilesmap(0,0) = \intextwang{white}{white}{white}{white}$.
  Suppose now that $\tilesmap(i,i) = \intextwang{white}{white}{white}{white}$. 
  Then, by $\mathrm{\ref{octtiles:Hori}}$ and our choice of tiles, we have that $\tilesmap(i{+}1, i)$ is left- and up-border.
  Thus, as $\intextwang{white}{white}{white}{white}$ is the only tile in $\tiles$ that is down-border, we conclude that $(i{+}1, i{+}1)$ carries $\intextwang{white}{white}{white}{white}$.
  For the second statement, suppose that there is a pair~$(i,j)$ for which  $j < i{-}1$ and $\tilesmap(i,j)$ contains a white-border tile, and take lexicographically smallest such~$(i,j)$.
  By design of $\tilingsys$, the tile $\tilesmap(i,j)$ is left- and up-border.
  From~$\mathrm{\ref{octtiles:Init}}$ we infer~$i > 1$.
  By $\mathrm{\ref{octtiles:Hori}}$, the tile $\tilesmap(i{-}1,j)$ is also right-border, contradicting the minimality of $(i,j)$.
\end{proof}

We first show undecidability of the White-bordered Octant Tiling Problem, which follows by a straightforward reduction from the Octant Tiling Problem~\cite[Sec.~3.1]{BresolinEtAl2010}.
\begin{lem}\label{lemma:white-bordered-octant-tiling-is-undec}
  The white-bordered Octant Tiling Problem is undecidable. 
\end{lem}
\begin{proof}
  We say that a domino tiling system $\tilingsys$ \emph{almost covers} the octant if there is a mapping $\tilesmap \colon \octant \to \tiles$ such that  all pairs $(n,m) \in \octant$ fulfil $\mathrm{\ref{octtiles:Verti}}$ and $\mathrm{\ref{octtiles:Hori}}$.
  It was stated by Bresolin et al.~\cite[Sec. 3.1]{BresolinEtAl2010} (and follows from the classical work of van Emde Boas~\cite{van1997convenience}) that deciding whether $\tilingsys$, that does not contain white-bordered tiles, almost covers the octant is undecidable. 
  To~prove undecidability of the White-bordered Octant Tiling Problem we provide a reduction from the tiling problem of Bresolin et al.
  Thus, let $\tilingsys \deff (\tilesCol, \tiles, \tileswhite)$ be a domino tiling system that does not have white-bordered tiles, and consider $\tilingsys' \deff (\tilesCol, \tiles', \tileswhite)$ to be the tiling system obtained by putting $\tiles' \deff \{ \intextwang{white}{white}{white}{white}, (\colour_l, \colour_d, \colour_r, \colour_u), (\whiteBox, \colour_d, \colour_r, \whiteBox) \mid (\colour_l, \colour_d, \colour_r, \colour_u) \in \tiles \}$.
  We~claim that~$\tilingsys'$ covers the octant if and only if $\tilingsys$ almost covers the octant. Indeed:

  \begin{itemize}\itemsep0em
  \item The ``if'' direction is relatively straightforward.
  Let $\tilesmap$ be a map witnessing that $\tilingsys$ almost covers the octant.
  We alter the tiles assigned by $\tilesmap$ to the diagonal to make them left- and up-border, then we shift $\tilesmap$ right, and fill the remaining places with~$\intextwang{white}{white}{white}{white}$.
  Formally, let $\tilesmap'$ be defined as: (i) $\tilesmap'(i,i) \deff \intextwang{white}{white}{white}{white}$ for all $i \in \N$, (ii) $\tilesmap'(i,i{-}1) \deff (\whiteBox, \colour_d, \colour_r, \whiteBox)$, where $\tilesmap(i,i) = (\colour_l, \colour_d, \colour_r, \colour_u)$, for all positive $i \in \N$, and (iii) $\tilesmap'(i,j) = \tilesmap'(i,j{-}1)$ for all remaining $i, j \in \N$.
  It can be readily verified that $\tilesmap'$ satisfies the required conditions.
  \item For the ``only if'' direction, let $\tilesmap'$ be a map witnessing that $\tilingsys'$ covers the octant. It suffices to take $\tilesmap\colon (i,j) \mapsto \tilesmap'(i{+}2, j)$. By Observation~\ref{obs:octant-tiling-proper-white}) none of the tiles assigned by $\tilesmap$ is white-bordered. As $\tilesmap'$ covers the octant, we infer that $\tilesmap$ (almost) covers the~octant. \qedhere 
  \end{itemize}
\end{proof}

Our undecidability proof relies on concepts from $\concepttilepath^{\tiles}$ and the non-regular language~$\langvpaeq{\roler}{\roles}$. 
We fix and enumerate a set of tiles $\tiles$ as $\tile_1, \ldots, \tile_\rmN$ for $\rmN \deff |\tiles|$.
As the first building block, we introduce octant interpretations.
An interpretation $\interI$ is called $\tiles$-\emph{octant} if there exists a function $\tilesmap\colon \octant \to \tiles$ fulling $\mathrm{\ref{octtiles:Init}}$ and $\mathrm{\ref{octtiles:Verti}}$ (but not necessarily $\mathrm{\ref{octtiles:Hori}}$) that satisfies the following conditions:
(i) $\DeltaI = \octant$,
(ii) $\roler^{\interI} = \{ ((n, 0), (n{+}1, 0)) \mid n \in \N \}$, 
(iii)~$\roles^{\interI} =  \{ ((n, m), (n, m{+}1)) \mid n,m \in \N, m < n \}$, and
(iv) $\conceptC_{\tile}^{\interI} =  \{ (n,m) \mid \tilesmap(n,m) = \tile \}$ for all tiles $\tile \in \tiles$.
In this case $\interI$ \emph{represents} $\tilesmap$. 
Due to the fact that $\tilesmap$ is a function, every domain element carries precisely one tile.
Moreover, for every such $\tilesmap$ we can easily find a $\tiles$-octant representing $\tilesmap$ (just employ the above definition). 
For more intuitions consult~Figure~\ref{fig:octant-structure}.

To avoid disjunction in the forthcoming query, we need to extend octant interpretations with yet another way of representing tiles, which will be based on distances. 
Suppose that an element $(n,m)$ from an octant interpretation $\interI$ carries a tile $\tile_i$, and that the tile assigned to its horizontal predecessor $(n{-}1, m)$, if exists, is equal to $\tile_j$.
We equip $(n,m)$ with an outgoing $\roles^+$-path~$\pathrho$ of length~$\rmN$, composed of fresh elements (we employ fresh concept names $\concept{In}$ and~$\overline{\concept{In}}$ to make a distinction between elements in octant and the ones present in ``extra paths'').
The $i$-th element of~$\pathrho$ will be the unique element of~$\pathrho$ that belongs to the interpretation of a concept~$\concept{Cur}$.
Similarly, the $j$-th element of $\pathrho$ will be the unique element of $\pathrho$ that belongs to the interpretation of a concept~$\concept{Prev}$.
Thus, the distance from $(n,m)$ to an element labelled by the concept $\concept{Cur}$ (resp.~$\concept{Prev}$) uniquely determines the tile of a current node (resp. its horizontal predecessor).
As a mere technicality, needed for query design, we enrich the element $(0,0)$ with an incoming $\roler^+$-path of length~$\rmN$. 
A formalization comes next.

\begin{defi}\label{def:hyperoctant}
  Let $\interI$ be an interpretation with a domain $\octant \times \ZZ_{\rmN{+}1} \cup \{ ({-}i{-}1, 0, 0) \mid i \in \ZZ_\rmN \}$,
  and let $\interI_{\octant}$ be the restriction of $\interI$ to the set $\octant \times \{ 0 \}$. 
  We call $\interI$ a $\tiles$-\emph{hyperoctant} if:
  \begin{itemize}\itemsep0em
  \item $\interI_{\octant}$ is isomorphic (via a projection $(n,m,0) \mapsto (n,m)$) to a $\tiles$-octant interpretation,
  \item $\concept{In}^{\interI} = \octant \times \{ 0 \}$, $\overline{\concept{In}}^{\interI} = \DeltaI \setminus \concept{In}^{\interI}$, $\overline{\concept{Prev}}^{\interI} = \DeltaI \setminus \concept{Prev}^{\interI}$,
  \item $\roler^{\interI} = \roler^{\interI_{\octant}} \cup \{ (({-}i{-}1, 0, 0), ({-}i, 0, 0)) \mid i \in \ZZ_\rmN \}$, 
  \item $\roles^{\interI} = \roles^{\interI_{\octant}} \cup\{ ((n,m,k),(n,m,k{+}1)) \mid (n,m) \in \octant, k \in \ZZ_\rmN \}$,
  \item $\conceptC_{\tile_k}^{\interI} = \conceptC_{\tile_k}^{\interI_{\octant}}$ and $\concept{Cur}^{\interI} = \{ (n,m,k) \mid (n,m,0) \in \conceptC_{\tile_k}^{\interI}, \tile_k \in \tiles \}$, and
  \item for every $(n,m) \in \octant$ there is precisely one positive $k$ for which $(n,m,k) \in \concept{Prev}^{\interI}$ holds, and for such a number $k$ we have that $\tile_k$ and the tile carried by $(n,m,0)$ are $\rmH$-compatible.
  \end{itemize}
  Note that, in addition to what is present in the definition of $\tiles$-hyperoctant, we employed fresh concept names, namely $\concept{In}$, $\concept{Prev}$, $\concept{Cur}$, $\overline{\concept{In}}$, and $\overline{\concept{Prev}}$.
  The purpose of ``overlined'' concepts is to help with a design of a query, as the use of negation is not allowed there.
  We call $\interI$ \emph{proper} if for all $(n,m,k) \in \concept{Cur}^{\interI}$ we have $(n{+}1,m,k) \in \concept{Prev}^{\interI}$.
  Note that properness is not definable in $\ALC$, but we will enforce it with a query.
  The map $\tilesmap\colon \octant \to \tiles$ \emph{represented} by a hyperoctant $\interI$ is the map represented by its $\tiles$-octant substructure~$\interI_{\octant}$.
\end{defi}

\begin{figure}[h]
  \centering
  \begin{tikzpicture}[transform shape]

      \draw (-3, 0) node[minicarre] (Aminusthree) {};
      \draw (-2, 0) node[minicarre] (Aminustwo) {};
      \draw (-1, 0) node[minicarre] (Aminusone) {};

      \draw (0, 0) node[medrond] (A0) {};
      \draw (2.5, 0) node[medrond] (A1) {};
      \draw (5, 0) node[medrond] (A2) {};
      \draw (2.5, 2.5) node[medrond] (B1) {};
      \draw (5, 2.5) node[medrond] (B2) {};
      \draw (5, 5) node[medrond] (C2) {};

      \draw (0-0.5, 0+0.5) node[minicarre, fill=lime] (A01) {};
      \draw (0-1, 0+1) node[minicarre] (A02) {};
      \draw (0-1.5, 0+1.5) node[minicarre] (A03) {};

      \draw (2.5-0.5, 0+0.5) node[minicarre, fill=black!40] (A11) {};
      \draw (2.5-1, 0+1) node[minicarre, fill=black] (A12) {};
      \draw (2.5-1.5, 0+1.5) node[minicarre] (A13) {};

      \draw (5-0.5, 0+0.5) node[minicarre] (A21) {};
      \draw (5-1, 0+1) node[minicarre, fill=black!40] (A22) {};
      \draw (5-1.5, 0+1.5) node[minicarre, fill=black] (A23) {};

      \draw (2.5-0.5, 2.5+0.5) node[minicarre, fill=lime] (B11) {};
      \draw (2.5-1, 2.5+1) node[minicarre] (B12) {};
      \draw (2.5-1.5, 2.5+1.5) node[minicarre] (B13) {};

      \draw (5-0.5, 2.5+0.5) node[minicarre, fill=black!40] (B21) {};
      \draw (5-1, 2.5+1) node[minicarre, fill=black] (B22) {};
      \draw (5-1.5, 2.5+1.5) node[minicarre] (B23) {};

      \draw (5-0.5, 5+0.5) node[minicarre,fill=lime] (C21) {};
      \draw (5-1, 5+1) node[minicarre] (C22) {};
      \draw (5-1.5, 5+1.5) node[minicarre] (C23) {};

      \path[->] (Aminusthree) edge [red] node[yshift=3] {$\roler$} (Aminustwo);
      \path[->] (Aminustwo) edge [red] node[yshift=3] {$\roler$} (Aminusone);
      \path[->] (Aminusone) edge [red] node[yshift=3] {$\roler$} (A0);

      \path[->] (A0) edge [red] node[yshift=3] {$\roler$} (A1);
      \path[->] (A1) edge [red] node[yshift=3] {$\roler$} (A2);

      \path[->] (A1) edge [blue, dotted] node {$\roles$} (B1);
      \path[->] (A2) edge [blue, dotted] node {$\roles$} (B2);
      \path[->] (B2) edge [blue, dotted] node {$\roles$} (C2);

      \path[->] (A0) edge [blue, dotted] node {$\roles$} (A01);
      \path[->] (A01) edge [blue, dotted] node {$\roles$} (A02);
      \path[->] (A02) edge [blue, dotted] node {$\roles$} (A03);

      \path[->] (A1) edge [blue, dotted] node {$\roles$} (A11);
      \path[->] (A11) edge [blue, dotted] node {$\roles$} (A12);
      \path[->] (A12) edge [blue, dotted] node {$\roles$} (A13);

      \path[->] (A2) edge [blue, dotted] node {$\roles$} (A21);
      \path[->] (A21) edge [blue, dotted] node {$\roles$} (A22);
      \path[->] (A22) edge [blue, dotted] node {$\roles$} (A23);

      \path[->] (B1) edge [blue, dotted] node {$\roles$} (B11);
      \path[->] (B11) edge [blue, dotted] node {$\roles$} (B12);
      \path[->] (B12) edge [blue, dotted] node {$\roles$} (B13);

      \path[->] (B2) edge [blue, dotted] node {$\roles$} (B21);
      \path[->] (B21) edge [blue, dotted] node {$\roles$} (B22);
      \path[->] (B22) edge [blue, dotted] node {$\roles$} (B23);

      \path[->] (C2) edge [blue, dotted] node {$\roles$} (C21);
      \path[->] (C21) edge [blue, dotted] node {$\roles$} (C22);
      \path[->] (C22) edge [blue, dotted] node {$\roles$} (C23);

     \smallwang{-0.125}{-0.125}{white}{white}{white}{white} 
     \smallwang{2.5-0.125}{-0.125}{white}{rouge}{rouge}{white}  
     \smallwang{5-0.125}{-0.125}{rouge}{rouge}{rouge}{rouge} 
     \smallwang{2.5-0.125}{2.5-0.125}{white}{white}{white}{white} 
     \smallwang{5-0.125}{2.5-0.125}{white}{rouge}{rouge}{white} 
     \smallwang{5-0.125}{5-0.125}{white}{white}{white}{white} 

     \draw (6.75, 0) node[] (XXX) {$\ldots$};
     \draw (6.75, 1) node[] (XXX) {$\ldots$};
     \draw (6.75, 2) node[] (XXX) {$\ldots$};
     \draw (6.75, 3) node[] (XXX) {$\ldots$};
     \draw (6.75, 4) node[] (XXX) {$\ldots$};
     \draw (6.75, 5) node[] (XXX) {$\ldots$};

  \end{tikzpicture}
    \caption{A fragment of a proper $\{ \tile_1, \tile_2, \tile_3 \}$-hyperoctant for $\tile_1 = \intextwang{white}{white}{white}{white}$, $\tile_2 = \intextwang{white}{rouge}{rouge}{white}$, and $\tile_3 = \intextwang{rouge}{rouge}{rouge}{rouge}$.
    Elements in $\concept{In}^{\interI}$ are depicted as circles.
    Elements in $\concept{Prev}^{\interI}$ are marked grey, elements in $\concept{Cur}^{\interI}$ are marked black, and the lime elements belong to $\concept{Prev}^{\interI} \cap \concept{Cur}^{\interI}$.}
    \label{fig:hyperoctant}
\end{figure}

\noindent We next relate proper $\tiles$-hyperoctants and domino tiling systems.
\begin{lem}\label{lemma:hyperoctant-works}
  Let $\tilingsys \deff (\tilesCol, \tiles, \tileswhite)$ be a domino tilling system. 
  For every proper $\tiles$-hyperoctant~$\interI$, the map $\tilesmap\colon \octant \to \tiles$ represented by $\interI$ witnesses that $\tilingsys$ covers the octant.
  If $\tilingsys$ covers the octant, as witnessed by a map $\tilesmap\colon \octant \to \tiles$, then there exists a proper $\tiles$-hyperoctant $\interI$ representing $\tilesmap$.
\end{lem}
\begin{proof}
  Take a proper $\tiles$-hyperoctant $\interI$, and the map $\tilesmap$ represented by $\interI$.
  By definition of a $\tiles$-octant, we have that $\tilesmap$ satisfies $\mathrm{\ref{octtiles:Init}}$ and $\mathrm{\ref{octtiles:Verti}}$. 
  To establish $\mathrm{\ref{octtiles:Hori}}$ take any $(n,m) \in \octant$, and suppose that $\tilesmap(n,m) = \tile_i$ and $\tilesmap(n{+}1,m) = \tile_j$ hold for some integers $i, j \in \N$.
  By definition of a hyperoctant and properness of $\interI$, we have that $(n{+}1,m,0)$ carries the tile~$\tile_j$, and $(n{+}1,m,i) \in \concept{Prev}^{\interI}$. 
  Thus, by the 6th item of Definition~\ref{def:hyperoctant} we infer that tiles $\tile_i$ and~$\tile_j$ are $\rmH$-compatible.
  The proof of the other statement is routine: take $\interI_{\octant}$ to be a $\tiles$-octant representing~$\tilesmap$. 
  W.l.o.g. assume that $\interI_{\octant}$ interprets all role names and concept names that are not mentioned in its definitions as empty sets.
  We next rename the domain of $\interI_{\octant}$ to $\octant \times \{ 0 \}$ and append fresh elements to make the domain equal to $\octant \times \ZZ_{\rmN{+}1} \cup \{ ({-}i{-}1, 0, 0) \mid i \in \ZZ_\rmN \}$. 
  Then, we enlarge the interpretations of $\roler, \roles, \concept{In}$, and $\concept{Cur}$ in a minimal way according to the first five items of Definition~\ref{def:hyperoctant}.
  Finally, we interpret $\concept{Prev}$ as the set composed of all triples $(n{+}1,m,k)$ for all $(n,m) \in \octant$ carrying a tile $\tile_k$, and all triples $(n,n,\ell)$ for $x \in \N$ and $\ell$ denoting the index of $\intextwang{white}{white}{white}{white}$ in $\tiles$.
  Call the resulting interpretation $\interI$.
  Clearly, $\interI$ is $\tiles$-hyperoctant due to the fact that $\tilesmap$ is a map and respects conditions $\mathrm{\ref{octtiles:Init}}$, $\mathrm{\ref{octtiles:Verti}}$, and $\mathrm{\ref{octtiles:Hori}}$.
\end{proof}

We employ a $\VPL$-\CQ~$\queryq_{\blacktriangle}^{\tilingsys}(\varu_1, \varu_2, \varv_1, \varv_2, \varw_1, \varw_2, \varx_1, \varx_2, \vary_1, \vary_2, \varz_1, \varz_2)$ as a tool for detecting whether a given $\tiles$-hyperoctant~$\interI$ is proper. 
Observe that $\interI$ \emph{is not} proper if and only if there is a position $(n,m) \in \octant$ and a number $1 \leq k \leq \rmN$, for which we have $(n,m,k) \in \concept{Cur}^{\interI}$ and $(n{+}1,m,k) \not\in \concept{Prev}^{\interI}$.
This is precisely the condition that is going to be expressed with $\queryq_{\blacktriangle}^{\tilingsys}$, informally presented at Figure~\ref{fig:Visualisation-of-queryq-blacktriangle}.
The intuition behind~$\queryq_{\blacktriangle}^{\tilingsys}$ is as~follows.

\begin{figure}[h]
  \centering
  \begin{tikzpicture}[transform shape]

      \draw (-4.5, 0) node[medrond] (A0prim) {};
      \draw (-3, 0) node[medrond] (A1prim) {};
      \draw (-1.5, 0) node[] (A2prim) {$\ldots$};

      \draw (0, 0) node[medrond] (A0) {};
      \draw (1.5, 0) node[medrond] (A1) {};
      \draw (3, 0) node[] (A2) {$\ldots$};
      \draw (4.5, 0) node[medrond] (A3) {};
      \draw (8, 0) node[medrond] (A4) {};
      \draw (4.5, 1) node[] (B3) {$\ldots$};
      \draw (8, 1) node[] (B4) {$\ldots$};
      \draw (4.5, 2) node (C3) {$\dots$};
      \draw (8, 2) node[] (C4) {$\dots$};
      \draw (4.5, 3) node[medrond] (D3) {};
      \draw (8, 3) node[medrond] (D4) {};

      \path[->] (A0prim) edge [red] node[yshift=3] {$\roler$} (A1prim);
      \path[->] (A1prim) edge [red] node[yshift=3] {$\roler$} (A2prim);
      \path[->] (A2prim) edge [red] node[yshift=3] {$\roler$} (A0);

      \path[->] (A0) edge [red] node[yshift=3] {$\roler$} (A1);
      \path[->] (A1) edge [red] node[yshift=3] {$\roler$} (A2);
      \path[->] (A2) edge [red] node[yshift=3] {$\roler$} (A3);
      \path[->] (A3) edge [red] node[yshift=3] {$\roler$} (A4);

      \path[->] (A3) edge [blue, dotted] node {$\roles$} (B3);
      \path[->] (B3) edge [blue, dotted] node {$\roles$} (C3);
      \path[->] (C3) edge [blue, dotted] node {$\roles$} (D3);

      \path[->] (A4) edge [blue, dotted] node {$\roles$} (B4);
      \path[->] (B4) edge [blue, dotted] node {$\roles$} (C4);
      \path[->] (C4) edge [blue, dotted] node {$\roles$} (D4);

      \node at (A0prim.center) {$\varu_1$};
      \node at (A1prim.center) {$\varu_2$};

      \node at (A0.center) {$\varx_1$};
      \node at (A1.center) {$\varx_2$};

      \node at (A3.center) {$\vary_1$};
      \node at (A4.center) {$\vary_2$};

      \node at (D3.center) {$\varz_1$};
      \node at (D4.center) {$\varz_2$};

      \node[left=0.5em of D3] {$\concept{In}$}; 
      \node[right=0.5em of D4] {$\concept{In}$}; 

      \draw [decorate, decoration={brace, amplitude=8pt}] (A0.north) -- (A3.north) node[midway, above=10pt] {$m$};
      \draw [decorate, decoration={brace, mirror, amplitude=8pt}] (A1.south) -- (A4.south) node[midway, below=10pt] {$m$};
      \draw [decorate, decoration={brace, amplitude=8pt, mirror, raise=3pt}] (A4.east) -- (D4.east) node[midway, right=5pt] {\ $m$};
      \draw [decorate, decoration={brace, amplitude=8pt, mirror, raise=3pt}] (A3.east) -- (D3.east) node[midway, right=5pt] {\ $m$};


      \draw (4.5, 4) node[ptcarre] (X1) {$\varw_1$};
      \draw (8, 4) node[ptcarre] (Y1) {$\varw_2$}; 

      \draw (4.5, 5) node[] (X2) {$\ldots$};
      \draw (8, 5) node[] (Y2) {$\ldots$}; 

      \draw (4.5, 6) node[ptcarre] (X3) {$\varv_1$};
      \draw (8, 6) node[ptcarre] (Y3) {$\varv_2$}; 

      \node[above=0.7em of X3] {$\concept{Cur}$}; 
      \node[above=0.7em of Y3] {$\overline{\concept{Prev}}$};

      \node[left=0.3em of X1] {$\overline{\concept{In}}$}; 
      \node[right=0.3em of Y1] {$\overline{\concept{In}}$};


      \draw (5.2, 0) node[] (A3Shadow) {};
      \draw (9, 0) node[] (A4Shadow) {};
      \draw (5.2, 5.8) node[] (X3Shadow) {};
      \draw (9, 5.8) node[] (Y3Shadow) {}; 

      \draw [decorate, decoration={brace, amplitude=8pt, mirror, raise=3pt}] (A3Shadow.east) -- (X3Shadow.east) node[midway, right=5pt] {\ $m{+}k$};
      \draw [decorate, decoration={brace, amplitude=8pt, mirror, raise=3pt}] (A4Shadow.east) -- (Y3Shadow.east) node[midway, right=5pt] {\ $m{+}k$};

      \draw [decorate, decoration={brace, amplitude=8pt}] (A0prim.north) -- (A0.north) node[midway, above=10pt] {$k$};
      \draw [decorate, decoration={brace, mirror, amplitude=8pt}] (A1prim.south) -- (A1.south) node[midway, below=10pt] {$k$};

      \path[->] (D3) edge [blue, dotted] node {$\roles$} (X1);
      \path[->] (D4) edge [blue, dotted] node {$\roles$} (Y1);
      \path[->] (X1) edge [blue, dotted] node {$\roles$} (X2);
      \path[->] (X2) edge [blue, dotted] node {$\roles$} (X3);
      \path[->] (Y1) edge [blue, dotted] node {$\roles$} (Y2);
      \path[->] (Y2) edge [blue, dotted] node {$\roles$} (Y3);

      \scoped[on background layer] \filldraw [red!50, draw opacity=0.2, fill opacity=0.2, line width=2.5em, line join=round] (X3.center) -- (Y3.center) -- cycle;

      \node at (6.25, 6.1) {\textcolor{red!80}{\text{mismatch!}}};

  \end{tikzpicture}
    \caption{Visualisation of the query $\queryq_{\blacktriangle}^{\tilingsys}(\varu_1, \varu_2, \varv_1, \varv_2, \varw_1, \varw_2, \varx_1, \varx_2, \vary_1, \vary_2, \varz_1, \varz_2)$.}
    \label{fig:Visualisation-of-queryq-blacktriangle}
\end{figure}

We first ensure that $\varz_1, \varz_2$ are mapped to some elements representing the coordinates $(n, m)$ and $(n', m')$ of the octant for some integers satisfying $n' = n {+} 1$ and~$m = m'$.
The~fact that they belong to the octant is handled by means of the $\concept{In}$ concept.
The equality $n' = n {+} 1$ is achieved by introducing variables $\vary_1, \vary_2$, stating their $\roler$-connectedness, and the $\roles^*$-reachability of $\varz_1$ and $\varz_2$ from, respectively, $\vary_1$ and $\vary_2$.
Thus $\vary_1, \vary_2$ are placed ``at the bottom'' of the variables $\varz_1$ and $\varz_2$.
Next, the equi-hight of $\varz_1, \varz_2$ is ensured with extra $\roler$-connected variables $\varx_1, \varx_2$ located to the left of $\vary_1$, and non-regular atoms $\langvpaeq{\roler}{\roles}(\varx_1, \varz_1)$ and $\langvpaeq{\roler}{\roles}(\varx_2, \varz_2)$, enforcing equality of the distance between $\varx_i$ and $\vary_i$, and the distance between $\vary_i$ and $\varz_i$, for all~$i \in \{ 1, 2 \}$.
Once we know that the variables~$\varz_1, \varz_2$ are mapped by a query as desired, we need to express that they violate properness of $\interI$.
Recall that we want to establish $(n,m,k) \in \concept{Cur}^{\interI}$ and $(n{+}1,m,k) \not\in \concept{Prev}^{\interI}$ for some $k$.
Such elements will be represented, respectively, by variables $\varv_1$ and $\varv_2$.  
To express the mentioned constraint, we introduce fresh $\roler$-connected variables $\varu_1, \varu_2$ that are located to the left of $\varx_1, \varx_2$, and whose distance to $\varu_1, \varu_2$ will be precisely the~$k$ that we are looking for.
We stress that we do not ``hardcode'' the value of $k$ in the query. 
As~the variable $\varu_1$ is free to map whenever it wants, this mimics a disjunction over possible values of $k$.
We ensure the variables $\varv_1, \varv_2$ are mapped to elements outside the octant by expressing that they are $\roler^*$-reachable from the $\roles$-successors $\varw_1, \varw_2$ of $\varz_1, \varz_2$, that are labelled with $\overline{\concept{In}}$.
Note that just expressing that $\varv_1, \varv_2$ belong to $\overline{\concept{In}}$ does not suffice, as the path leading from some of $\varz_i$ to $\varv_i$ could contain elements in $\concept{In}$ (which we implicitly forbid).
With non-regular atoms $\langvpaeq{\roler}{\roles}(\varu_1, \varv_1)$ and $\langvpaeq{\roler}{\roles}(\varu_2, \varv_2)$ we make sure that the distance between $\varz_1$ and $\varv_1$ (respectively~$\varz_2$ and~$\varv_2$) is indeed~$k$.

Despite the high technicality of our construction, we hope that after reading the above intuition and glancing at Figure~\ref{fig:Visualisation-of-queryq-blacktriangle}, the following definition of the query $\queryq_{\blacktriangle}^{\tilingsys}$ should be now~understandable:

\begin{align*}
  \queryq_{\blacktriangle}^{\tilingsys} \deff \roler(\varu_1, \varu_2) \land \roler^*(\varu_2, \varx_1) \land \roler(\varx_1, \varx_2) \land \roler^*(\varx_2, \vary_1) \land \roler(\vary_1, \vary_2) \land \concept{Cur}(\varv_1) \land \overline{\concept{Prev}}(\varv_2) \land \\
  \bigwedge_{i=1}^{2} \big[ \roles^*(\vary_i, \varz_i) \land \concept{In}(\varz_i) \land \roles(\varz_i, \varw_i) \land \overline{\concept{In}}(\varw_i) \land \roles^*(\varw_i, \varv_i)  \land \langvpaeq{\roler}{\roles}(\varx_i, \varz_i) \land \langvpaeq{\roler}{\roles}(\varu_i, \varv_i) \big].
\end{align*}

By routine case analysis with a bit of calculations, we can show that:
\begin{lem}\label{lemma:encoding-octant-uno}
Let $\tilingsys \deff (\tilesCol, \tiles, \tileswhite)$ be a domino tilling system and let $\interI$ be a $\tiles$-hyperoctant.
Then we have that $\interI$ is proper if and only if $\interI \not\models \queryq_{\blacktriangle}^{\tilingsys}$.
\end{lem}
\begin{proof}
  We start from the ``only if'' direction.
  Suppose that $\interI$ is proper, but there is a match~$\matcheta$ witnessing $\interI \models \queryq_{\blacktriangle}^{\tilingsys}$.
  We provide a few tedious calculations.
  As the atoms $\roler(\varx_1, \varx_2)$, $\roler(\vary_1, \vary_2)$, and $\roler(\varu_1, \varu_2)$ belong to $\queryq_{\blacktriangle}^{\tilingsys}$, there are $a, b, n \in \ZZ$ such that $\matcheta(\varu_1) = (a, 0, 0)$, $\matcheta(\varu_2) = (a{+}1, 0, 0)$, $\matcheta(\varx_1) = (b, 0, 0)$, $\matcheta(\varx_2) = (b{+}1, 0, 0)$, $\matcheta(\vary_1) = (n, 0, 0)$, and $\matcheta(\vary_2) = (n{+}1, 0, 0)$.
  By the presence of atoms $\roler^*(\varu_2, \varx_1)$ and $\roler^*(\varx_2, \vary_1)$ in the query, we infer that $a < b < n$.
  From atoms $\roles^*(\vary_i, \varz_i), \concept{In}(\varz_i)$ present in $\queryq_{\blacktriangle}^{\tilingsys}$, we can deduce that $\matcheta(\varz_1) = (n, m, 0)$, $\matcheta(\varz_2) = (n{+}1, m', 0)$ hold for some $m, m' \in \N$.
  Next, by the satisfaction of the atoms $\langvpaeq{\roler}{\roles}(\varx_1, \varz_1)$ and $\langvpaeq{\roler}{\roles}(\varx_2, \varz_2)$ the equations $m = n - b$ and $m' = (n{+}1)-(b{+}1)$ follow. 
  Thus $m' = m$.
  As the next step, we deal with the variables $\varw_1$ and $\varw_2$.
  From the atoms $\roles(\varz_i, \varw_i)$ and $\overline{\concept{In}}(\varw_i)$ we can deduce that $\matcheta(\varw_1) = (n, m, 1)$, $\matcheta(\varw_2) = (n{+}1, m, 1)$.
  Together with atoms $\roles^*(\varw_1, \varv_1)$ and $\roles^*(\varw_1, \varv_1)$, this implies that $\matcheta(\varv_1) = (n, m, k)$, and $\matcheta(\varz_2) = (n{+}1, m, k')$ hold for some positive~$k, k' \in \N$.
  By the satisfaction of the atoms $\langvpaeq{\roler}{\roles}(\varu_1, \varv_1)$ and $\langvpaeq{\roler}{\roles}(\varu_2, \varv_2)$ the equations $m{+}k = n{-}a$ and $m{+}k' = (n{+}1)-(a{+}1)$ follow. 
  Hence, $k = k'$ holds and by collecting all the previous equations we conclude that $\matcheta(\varv_1) = (n, m, k)$ and  $\matcheta(\varv_1) = (n{+}1, m, k)$.
  Finally, the atoms $\concept{Cur}(\varv_1)$ and $\overline{\concept{Prev}}(\varv_2)$ of $\queryq_{\blacktriangle}^{\tilingsys}$ imply that $(n, m, k) \in \concept{Cur}^{\interI}$ but $(n{+}1, m, k) \not\in \concept{Prev}^{\interI}$.
  This contradicts the properness of $\interI$. 

  We show the ``if'' direction by contraposition. 
  Suppose that $\interI$ is not proper, and that the properness of $\interI$ is violated by $(n, m, k) \in \concept{Cur}^{\interI}$ and $(n{+}1, m, k) \in \overline{\concept{Prev}}^{\interI}$.
  Then the map $\varv_1 \mapsto (n,m,k)$, $\varv_2 \mapsto (n{+}1,m,k)$, $\varw_1 \mapsto (n,m,1)$, $\varw_2 \mapsto (n{+}1,m,1)$, $\varz_1 \mapsto (n,m,0)$, $\varz_2 \mapsto (n{+}1,m,0)$, $\vary_1 \mapsto (n,0,0)$, $\vary_2 \mapsto (n{+}1,0,0)$, $\varx_1 \mapsto (n{-}m, 0,0)$, $\varx_2 \mapsto (n{-}m{+}1, 0,0)$, $\varu_1 \mapsto (n{-}m{-}k, 0,0)$, $\varu_2 \mapsto (n{-}m{-}k{+}1, 0,0)$ is a match for $\queryq_{\blacktriangle}^{\tilingsys}$. This finishes the proof.

  As a side remark, note that if the value of $n{-}m$ is sufficiently small, the value of $n{-}m{-}k$ can be negative (but not smaller than $-\rmN$). This explains at last why we appended an incoming $\roler^*$-path of length $\rmN$ to the element $(0,0,0)$ in the construction of~hyperoctants.
\end{proof}

As the final step of our construction, we define an $\ALC$-TBox $\tboxT_{\blacktriangle}^{\tilingsys}$ whose intended tree-like models will contain $\tiles$-hyperoctants.
Most of the axioms written below are formalizations of straightforward properties satisfied by $\tiles$-hyperoctants and grids.
As we are aiming to design an $\ALC$-TBox, all the properties expressed in $\tboxT_{\blacktriangle}^{\tilingsys}$ are interpreted ``globally''. Hence, to express existence of a starting point of a $\tiles$-hyperoctant we employ a fresh role name $\role{aux}$ and say that every element has an $\role{aux}$-successor (with the intended meaning that such a successors ``starts'' a $\tiles$-hyperoctant).
For brevity, we say that an element is inner (resp. outer) if it belongs (resp. does not belong) to the interpretation of $\concept{In}$.
For a language $\languageL$ we also say that an element $\domeleme$ in $\interI$ is $\languageL$-outer-reachable from $\domelemd$ if there exists a $\languageL$-path from $\domelemd$ to $\domeleme$ composed solely of outer elements (with a possible exception~of~$\domelemd$).
 
\begin{description}\itemsep0em
  \item[\desclabel{(GStart)}{grid:Start}] Every element has an $\role{aux}$-successor.
  Every such successor has an outgoing $\roler^{\rmN}$-path composed of outer elements that leads to an inner element carrying a tile $\intextwang{white}{white}{white}{white}$ that has an inner $\roler$-successor carrying a tile different from $\intextwang{white}{white}{white}{white}$.
  \item[\desclabel{(GCompl)}{grid:Compl}] Concept name $\overline{\concept{In}}$ (resp. $\overline{\concept{Prev}}$) is interpreted as the complement of the interpretation of concept name $\concept{In}$ (resp. $\concept{Prev}$).
  \item[\desclabel{(GTil)}{grid:Til}] Every inner element is labelled with precisely one concept name from $\concepttilepath^{\tiles}$, and there are no outer elements labelled with such concept names.
  \item[\desclabel{(GSuc)}{grid:Suc}] If an element has an inner $\roler$-successor then such a successor also has an inner $\roler$-successor.
  Every inner element has an inner $\roles$-successor.\footnote{This ensures for each $i \in \N$ the existence of $i$ $\roles$-successors for the $i$-th element in the bottom of the octant.}
  \item[\desclabel{(GPath)}{grid:Path}] Every inner element has an outgoing $\roles^\rmN$-path composed solely of outer elements.
  \item[\desclabel{(GCur)}{grid:Cur}] For every inner element carrying the tile $\tile_k$ we have that: 
  (i) all $\roles^k$-outer-reachable elements are labelled with $\concept{Cur}$,
  (ii) there is no $\roles^\ell$-outer-reachable element, for $\ell \leq \rmN$ and $\ell \neq k$, that is labelled with $\concept{Cur}$.

  \item[\desclabel{(GPrev)}{grid:Prev}] For every inner element there exists a number $k \leq \rmN$ such that:
  (i) all $\roles^k$-outer-reachable elements are labelled with $\concept{Prev}$,
  (ii) there is no $\roles^\ell$-outer-reachable element, for $\ell \leq \rmN$ and $\ell \neq k$, that is labelled with $\concept{Prev}$.

  \item[\desclabel{(GVerti)}{grid:Verti}] Every pair of tiles carried by $\roles$-successors is $\rmV$-compatible.

  \item[\desclabel{(GHori)}{grid:Hori}] For all tiles $\tile \in \tiles$ and all elements carrying $\tile$ that can $\roles^\ell$-outer-reach an element labelled with $\concept{Prev}$ for some $\ell$, we have that $\tile_\ell$ and $\tile$ are $\rmH$-compatible.
\end{description}

\noindent The definition of $\tboxT_{\blacktriangle}^{\tilingsys}$ is provided in the proof of the following lemma.

\begin{lem}\label{lemma:Tblacktriangle-exists}
  There exists an $\ALC$-TBox $\tboxT_{\blacktriangle}^{\tilingsys}$ expressing the above properties.
\end{lem}
\begin{proof}
  Given $\ALC$-concepts $\conceptC, \conceptD$ and a role name $\roler \in \Rlang$ we employ macros $(\forall{\roler}.\conceptC)^{n}.\conceptD$ and $(\exists{\roler}.\conceptC)^{n}.\conceptD$ defined inductively for $n \geq 1$ as follows: 
  \[
  (\forall{\roler}.\conceptC)^{1}.\conceptD \deff \forall{\roler}.(\conceptC \to \conceptD), \qquad (\forall{\roler}.\conceptC)^{n{+}1}.\conceptD \deff \forall{\roler}.(\conceptC \to [(\forall{\roler}.\conceptC)^{n}.\conceptD]), 
  \]
  \[
    (\exists{\roler}.\conceptC)^{1}.\conceptD \deff \exists{\roler}.(\conceptC \dland \conceptD), \qquad (\exists{\roler}.\conceptC)^{n{+}1}.\conceptD \deff \exists{\roler}.(\conceptC \dland [(\exists{\roler}.\conceptC)^{n}.\conceptD]).
  \]
  Our $\tboxT_{\blacktriangle}^{\tilingsys}$ is composed of all the GCIs listed below, 
  describing properties $\mathrm{\ref{grid:Start}}$, $\mathrm{\ref{grid:Compl}}$, $\mathrm{\ref{grid:Til}}$, $\mathrm{\ref{grid:Suc}}$, $\mathrm{\ref{grid:Path}}$, $\mathrm{\ref{grid:Cur}}$, $\mathrm{\ref{grid:Prev}}$, $\mathrm{\ref{grid:Verti}}$, and $\mathrm{\ref{grid:Hori}}$ in precisely~this~order.
    \begin{align*} 
        \top & \dlsubseteq (\exists{\role{aux}}.\top) \dland \forall{\role{aux}}.\left( (\exists{\roler}.\overline{\concept{In}})^{\rmN}. [ \exists{\roler}.\left( \concept{In} \dland \conceptC_{\intextwang{white}{white}{white}{white}} \dland \exists{\roler}.(\concept{In} \dland \neg \conceptC_{\intextwang{white}{white}{white}{white}})  \right) ]  \right), \\
        \top & \dlsubseteq (\overline{\concept{In}} \leftrightarrow \neg \concept{In}) \dland (\overline{\concept{Prev}} \leftrightarrow \neg \concept{Prev}),\\
        \top & \dlsubseteq [(\bigdlor_{\tile \in \tiles} \conceptC_{\tile}) \to \concept{In}] \dland [\concept{In} \to (\bigdlor_{\tile \in \tiles} (\conceptC_{\tile} \dland \bigdland_{\tile' \in \tiles \setminus \{ \tile \}} \neg \conceptC_{\tile'}) ], \\
      \end{align*}
      \begin{align*}
        \top & \dlsubseteq (\forall{\roler}.[\concept{In} \to (\exists{\roler}.\concept{In})]) \dland (\concept{In} \to \exists{\roles}.\concept{In} ), \\ 
                \top & \dlsubseteq \concept{In} \to (\exists{\roles}.\overline{\concept{In}})^{\rmN}.\top, \\
        \top & \dlsubseteq \bigdland_{1 \leq k \leq \rmN} \left( 
          \conceptC_{\tile_k} \to \Big[
            [ (\forall{\roles}.\overline{\concept{In}})^k.\concept{Cur} ] \dland \bigdland_{1 \leq \ell \leq \rmN,\ \ell \neq k} [ (\forall{\roles}.\overline{\concept{In}})^\ell.\neg\concept{Cur} ]
          \Big]
        \right),\\
        \top & \dlsubseteq \bigdlor_{1 \leq k \leq \rmN} \left( 
          [ (\forall{\roles}.\overline{\concept{In}})^k.\concept{Prev} ] \dland \bigdland_{1 \leq \ell \leq \rmN,\ \ell \neq k} [ (\forall{\roles}.\overline{\concept{In}})^\ell.\neg\concept{Prev} ]
        \right),\\
        \top & \dlsubseteq \bigdlor_{\tile \in \tiles} \left( \conceptC_{\tile} \to [ \bigdland_{\tile' \in \tiles, (\tile, \tile') \ \text{are not}\ \rmV\text{-compatible}} \neg \exists{\roles}.\conceptC_{\tile'} ] \right),\\
        \top & \dlsubseteq \bigdlor_{\tile \in \tiles} \left( \conceptC_{\tile} \to [ \bigdland_{\tile_{\ell} \in \tiles, (\tile_{\ell}, \tile) \ \text{are not}\ \rmH\text{-compatible}} \neg (\exists{\roles}.\overline{\concept{In}})^{\ell}.\concept{Prev}  ] \right).
    \end{align*}
  It should be clear that the above GCIs formalise the required properties.
\end{proof}

We first see that $\tiles$-hyperoctant can be extended to (counter)models of $\tboxT_{\blacktriangle}^{\tilingsys}$ and $\queryq_{\blacktriangle}^{\tilingsys}$.

\begin{lem}\label{lemma:from-hyperoctants-to-models}
Every proper $\tiles$-hyperoctant can be extended to a model of $\tboxT_{\blacktriangle}^{\tilingsys}$ that violates~$\queryq_{\blacktriangle}^{\tilingsys}$.
\end{lem}
\begin{proof}
  Let $\interI$ be a proper $\tiles$-hyperoctant. 
  Consider an interpretation $\interJ$ with the domain $\DeltaJ \deff \DeltaI \cup ((\N \times \N) \setminus \octant) \times \ZZ_{\rmN{+}1}$ defined as follows:
  \begin{itemize}\itemsep0em
  \item $\interJ$ restricted to $\DeltaI$ is isomorphic to $\interI$,
  \item $\roles^{\interJ} \deff \roles^{\interI} \cup \{ ((n,m,k),(n,m,k{+}1)), ((n,m,0),(n,m{+}1,0)) \mid (n,m) \in (\N \times \N) \setminus \octant, k \in \ZZ_\rmN \}$,
  \item $\roler^{\interJ} \deff \roler^{\interI}$,
  \item $\role{aux}^{\interJ} \deff \DeltaJ \times \{ (-\rmN,0,0)\}$, and
  \item For each concept name $\conceptA$, each tuple $(n,m,k)$ with $(n,m) \not\in \octant$ belongs $\conceptA^{\interJ}$ iff~$(0,0,k) \in \conceptA^{\interI}$.
  \end{itemize}
  Intuitively, we enlarged $\interI$ to a grid and filled fresh places with copies of the diagonal of~$\interI$ (this is needed to fulfil the second conjunct of the axiomatisation of $\mathrm{\ref{grid:Suc}}$).
  By construction of $\interJ$ (and the fact that $\interI$ is a $\tiles$-hyperoctant) it follows $\interJ$ is a model of~$\tboxT_{\blacktriangle}^{\tilingsys}$.
  Call $\interJ$ \emph{proper} if for all $(n,m,k) \in \concept{Cur}^{\interJ}$ we have $(n{+}1,m,k) \in \concept{Prev}^{\interJ}$.
  Note that as~$\interI$ is proper, the above condition holds for all triples $(n,m,k)$ with $(n,m) \in \octant$.
  For other tuples, we simply use the fact that the interpretation of concepts for $(n,m,k)$ with $(n,m) \not\in \octant$ is inherited from the elements of the form $(0,0,k)$.
  Thus $\interJ$ is indeed proper. Now, without any further changes, the proof of the ``only if'' direction of Lemma~\ref{lemma:encoding-octant-uno} establishes~$\interJ \not\models \queryq_{\blacktriangle}^{\tilingsys}$.
\end{proof}

As a handy lemma used in the forthcoming proof, we need to establish that:
\begin{lem}\label{lemma:tree-models-of-Tblacktriangle-contain-hyperoctants}
Every single-role tree-like model of $\tboxT_{\blacktriangle}^{\tilingsys}$ contains a substructure that is isomorphic to a $\tiles$-hyperoctant.
\end{lem}
\begin{proof}
  The proof idea is simple but tedious. 
  We take single-role tree-like model $\interI$ of~$\tboxT_{\blacktriangle}^{\tilingsys}$, and select elements that will later constitute a hyperoctant branch-by-branch. 
  Let $\domelemd_{(-\rmN, 0, 0)}$ be any element of~$\interI$ which is an $\role{aux}$-successor of some element of $\interI$ (it exists as $\DeltaI \neq \emptyset$ and $\interI$ satisfies~$\mathrm{\ref{grid:Start}}$). 
  Thus there exists an $\roler^*$-path of elements $\domelemd_{(-\rmN, 0, 0)}, \domelemd_{(-\rmN{+}1, 0, 0)}, \ldots, \domelemd_{(0,0,0)},$ $\domelemd_{(1,0,0)}$ witnessing the satisfaction of $\mathrm{\ref{grid:Start}}$, where only $\domelemd_{(0,0,0)}, \domelemd_{(1,0,0)}$ belong to $\concept{In}^{\interJ}$ (by $\mathrm{\ref{grid:Start}}$ and $\mathrm{\ref{grid:Compl}}$), $\domelemd_{(0,0,0)}$ carries a tile $\intextwang{white}{white}{white}{white}$ and $\domelemd_{(1,0,0)}$ carries a tile that differs from $\intextwang{white}{white}{white}{white}$.  
  After employing a routine induction, we see that the first conjunct of~$\mathrm{\ref{grid:Suc}}$ yields the existence of an infinite $\roler^*$-path $\domelemd_{(1,0,0)}, \domelemd_{(2,0,0)}, \ldots$ of inner elements.
  By the second conjunct of $\mathrm{\ref{grid:Suc}}$ we know that every $\domelemd_{(i,0,0)}$ has an outgoing $\roles^{i}$-path $\domelemd_{(i,1,0)}, \ldots, \domelemd_{(i,i,0)}$ composed of inner elements. 
  Finally, by $\mathrm{\ref{grid:Path}}$, every element~$\domelemd_{(i,j,0)}$ has an outgoing $\roles^{\rmN}$-path $\domelemd_{(i,j,1)}$,$\ldots$, $\domelemd_{(i,j,\rmN)}$ composed of outer elements.
  Note that by tree-likeness of $\interI$ we have that all of selected elements $\domelemd_{(i,j,k)}$ are pairwise-different, and any pair of elements is connected by at most one role.
  Let $\interJ$ be the restriction of $\interI$ to the set of selected element, with the domain renamed with a map $\domelemd_{(i,j,k)} \mapsto (i,j,k)$.
  It follows from $\mathrm{\ref{grid:Til}}$ and~$\mathrm{\ref{grid:Compl}}$ that~$\interJ$ restricted to $\concept{In}^{\interJ}$ is a $\tiles$-octant. 
  Other properties of $\tiles$-hyperoctants can be readily verified with $\mathrm{\ref{grid:Cur}}$, $\mathrm{\ref{grid:Prev}}$, $\mathrm{\ref{grid:Verti}}$, $\mathrm{\ref{grid:Hori}}$.
  Thus $\interJ$ is a $\tiles$-hyperoctant, as desired.
\end{proof}

The next lemma summarises our reduction.

\begin{lem}\label{lemma:D-covers-iff-Tblacktriangle-does-not-entail}
$\tilingsys$ covers the octant if and only if $\tboxT_{\blacktriangle}^{\tilingsys} \not\models \queryq_{\blacktriangle}^{\tilingsys}$. 
\end{lem}
\begin{proof}
  Suppose that $\tilingsys$ covers the octant. 
  Then, by Lemma~\ref{lemma:hyperoctant-works}, there exists a proper $\tiles$-hyperoctant $\interI$.
  Hence, from Lemma~\ref{lemma:from-hyperoctants-to-models} we conclude existence of a model $\interJ$ of $\tboxT_{\blacktriangle}^{\tilingsys}$ such that $\interJ \not\models \queryq_{\blacktriangle}^{\tilingsys}$.
  Thus $\tboxT_{\blacktriangle}^{\tilingsys} \not\models \queryq_{\blacktriangle}^{\tilingsys}$.
  For the other direction, suppose that there exists an interpretation~$\interI$ such that $\interI \models \tboxT_{\blacktriangle}^{\tilingsys}$ but $\interI \not\models \queryq_{\blacktriangle}^{\tilingsys}$. 
  By Corollary~\ref{cor:tree-model-property-for-ALCVPL} we can assume that $\interI$ is single-role tree-like.
  From Lemma~\ref{lemma:tree-models-of-Tblacktriangle-contain-hyperoctants} we know that $\interI$ contains a substructure~$\interJ$ that is a $\tiles$-hyperoctant.
  As~$\interI \not\models \queryq_{\blacktriangle}^{\tilingsys}$, we know that $\interJ \not\models \queryq_{\blacktriangle}^{\tilingsys}$.
  Thus $\interJ$ is a proper $\tiles$-hyperoctant.
  Hence, by Lemma~\ref{lemma:hyperoctant-works} we conclude that $\tilingsys$ covers the octant.
\end{proof}

\noindent As the octant tiling problem is undecidable, we conclude the main theorem of this section.
\begin{thm}\label{thm:undecidability-querying-with-VPLs}
The entailment problem of $\{ \roler, \roles, \roler^*, \roles^*, \langvpaeq{\roler}{\roles} \}$-conjunctive-queries over $\ALC$-TBoxes is undecidable.
\end{thm} 

Interestingly enough\footnote{I would like to thank Anni-Yasmin Turhan for asking this question.}, our proof technique can be adjusted (with little effort) to derive related results for query languages involving inverses. 
Before moving to the undecidability result, let us define the class of \emph{Visibly-Pushdown Path Queries}~\cite[p.~330]{LangeL15} ({\VPQ}s) as the class of single-atom $\VPL$-{\CQ}s (\ie $\VPL$-{\CQ}s without conjunction). 
We stress that the entailment problem of such queries is decidable.
The main idea is that the non-satisfaction of a {\VPQ} of the form $\languageL(\varx,\vary)$ can be expressed with an $\ALCvpl$-GCI $\top \dlsubseteq \neg \exists{\languageL}.\top$. Indeed:

\begin{cor}\label{cor:VPQ-entailment-is-decidable}
  The entailment problem of {\VPQ}s over $\ALCvpl$-TBoxes is $\TwoExpTime$-complete (assuming that the {$\VPL$}s from the query and the TBox are over the same alphabet).
\end{cor}
\begin{proof}
  The lower bound is inherited from the concept satisfiability problem for $\ALCvpl$~\cite[Thm.~19]{LodingLS07}.
  Let $\tboxT$ be an $\ALCvpl$-TBox and $\queryq$ be a {\VPQ} of the form $\languageL(\varx,\vary)$.
  The crucial observation is that by the semantics of queries we have that $\tboxT \not\models \queryq$ if and only if $\tboxT' \deff \tboxT \cup \{ \top \dlsubseteq \neg \exists{\languageL}.\top \}$ has a model. 
  Hence, by a well-known internalisation of TBoxes as concepts in the presence of regular expressions~\cite[p.~186]{2003handbook}, we can compute (in time polynomial w.r.t. $|\tboxT|$) an $\ALCvpl$-concept $\conceptC$ that is satisfiable if and only if $\tboxT'$ is.
  Now it suffices to check the satisfiability of $\conceptC$. 
  This in turn can be done in doubly-exponential time by the result of L\"oding et al.~\cite[Thm.~18]{LodingLS07}, finishing the proof. 
\end{proof}

We next extend the class of {\VPQ}s with an inverse operator, obtaining the class of \emph{Two-Way Visibly Pushdown Path Queries} ({\TwoVPQ}s). More precisely, such queries are {\VPQ}s that allow for letters of the form $\roler^-$ for all role names $\roler \in \Rlang$ in the underlying visibly-pushdown alphabet. When evaluating {\TwoVPQ}s over interpretations $\interI$, the role $(\roler^-)^{\interI}$ is interpreted as the set-theoretic inverse of the role $\roler^{\interI}$.
In what follows we sketch the proof of the fact that the entailment of {\TwoVPQ}s over $\ALC$-TBoxes is undecidable.
The key ingredients of our reduction are the previously-defined $\ALC$-TBox $\tboxT_{\blacktriangle}^{\tilingsys}$ and a fresh {\TwoVPQ} $\queryq_{{\TwoVPQ}}^{\tilingsys} \deff \languageL_{\downarrow\rightarrow\uparrow}(\varx,\vary)$, where
\[
  \languageL_{\downarrow\rightarrow\uparrow}\!\deff\!\Big\{ \concept{Cur}?\ \overline{\concept{In}}? \left( \roles^-\; \overline{\concept{In}}?\right)^k \ \roles \; (\concept{In}?\ \roles)^m \;
  (\concept{In}? \; \roler \; \concept{In}?) \; (\roles \ \concept{In}?)^m \roles 
  \left(\overline{\concept{In}}? \; \roles \right)^k  \overline{\concept{In}}?\ \overline{\concept{Prev}}? \!\mid\! k, m \in \N \Big\}.
\]

It is an easy  exercise to construct a pushdown automaton for the language $\languageL_{\downarrow\rightarrow\uparrow}$. Such an automaton becomes visibly-pushdown under the requirement that $\roles$ and~$\roles^-$ are, respectively, return and call symbols, and $\roler$ is an internal symbol.
The forthcoming Lemma~\ref{lemma:encoding-octant-revisited} relates the queries $\queryq_{\blacktriangle}^{\tilingsys}$ and $\queryq_{{\TwoVPQ}}^{\tilingsys}$.
As it can be established analogously to the other lemmas of this section, we only sketch the proof and leave some minor details to the reader.

\begin{lem}\label{lemma:encoding-octant-revisited}
  Let $\tilingsys \deff (\tilesCol, \tiles, \tileswhite)$ be a domino tilling system and let $\interI$ be a $\tiles$-hyperoctant.
  We have that $\interI \models \queryq_{\blacktriangle}^{\tilingsys}$ if and only if $\interI \models \queryq_{{\TwoVPQ}}^{\tilingsys}$.
  Moreover, by applying the construction from Lemma~\ref{lemma:from-hyperoctants-to-models}, every proper $\tiles$-hyperoctant can be extended to a model of $\tboxT_{\blacktriangle}^{\tilingsys}$ that violates~$\queryq_{{\TwoVPQ}}^{\tilingsys}$.
\end{lem}
\begin{proof}[Proof sketch.]
  We first sketch the proof of the equivalence $\interI \models \queryq_{\blacktriangle}^{\tilingsys}$ if and only if $\interI \models \queryq_{{\TwoVPQ}}^{\tilingsys}$.
  First, take any match $\matcheta$ witnessing $\interI \models \queryq_{\blacktriangle}^{\tilingsys}$.
  Then it can be readily verified that the mapping $\varx \mapsto \matcheta(\varv_1), \vary \mapsto \matcheta(\varv_2)$ is a match for $\interI$ and $\queryq_{{\TwoVPQ}}^{\tilingsys}$.
  For the opposite direction, let $\domelemd, \domeleme \in \DeltaI$ be domain elements for which the mapping $\varx \mapsto \domelemd, \vary \mapsto \domeleme$ is a match for $\interI$ and $\queryq_{{\TwoVPQ}}^{\tilingsys}$. 
  Hence,~$\domeleme$ is $\textstyle[\overline{\concept{In}}? \left( \roles^-\; \overline{\concept{In}}?\right)^k \ \roles \; (\concept{In}?\ \roles)^m \;
  (\concept{In}? \; \roler \; \concept{In}?) \; (\roles \ \concept{In}?)^m \roles 
  \left(\overline{\concept{In}}? \; \roles \right)^k  \overline{\concept{In}}?]$-reachable from $\domelemd$ for some integers $m, k \in \N$.
  By analysing the shape of $\tiles$-hyperoctants and the above expression, we infer the existence of an integer $n \in \N$ for which $\domelemd = (n,m,k)$, and $\domeleme = (n{+}1,m,k)$. Indeed, the equality of the second and third coordinates is ensured by the presence of $m$ and $k$ in the above path expression; the fact that the first coordinate of $\domeleme$ is the successor value of the first coordinate of $\domelemd$ is guaranteed by the subexpression $(\concept{In}? \; \roler \; \concept{In}?)$ in $\queryq_{{\TwoVPQ}}^{\tilingsys}$. 
  Now the match for $\queryq_{\blacktriangle}^{\tilingsys}$ and $\interI$ can be defined as:
  $\varv_1 \mapsto (n,m,k)$, $\varv_2 \mapsto (n{+}1,m,k)$, $\varw_1 \mapsto (n,m,1)$, $\varw_2 \mapsto (n{+}1,m,1)$, $\varz_1 \mapsto (n,m,0)$, $\varz_2 \mapsto (n{+}1,m,0)$, $\vary_1 \mapsto (n,0,0)$, $\vary_2 \mapsto (n{+}1,0,0)$, $\varx_1 \mapsto (n{-}m, 0,0)$, $\varx_2 \mapsto (n{-}m{+}1, 0,0)$, $\varu_1 \mapsto (n{-}m{-}k, 0,0)$, $\varu_2 \mapsto (n{-}m{-}k{+}1, 0,0)$.

  For the remaining part of the proof, take any proper $\tiles$-hyperoctant $\interI$. 
  We apply the construction from the proof of Lemma~\ref{lemma:from-hyperoctants-to-models}, and obtain a model $\interJ$ of $\tboxT_{\blacktriangle}^{\tilingsys}$ that violates $\queryq_{\blacktriangle}^{\tilingsys}$. 
  It~suffices to show that $\interJ \not\models \queryq_{{\TwoVPQ}}^{\tilingsys}$. Towards a contradiction, suppose the opposite, and take any match $\matcheta$ for $\interJ$ and $\queryq_{{\TwoVPQ}}^{\tilingsys}$. 
  Then, by the construction of $\interJ$ and the shape of $\queryq_{{\TwoVPQ}}^{\tilingsys}$, we conclude that the image of $\matcheta$ belongs to the $\tiles$-hyperoctant part $\interI$ of $\interJ$.
  By equivalence between the queries $\queryq_{{\TwoVPQ}}^{\tilingsys}$ and $\queryq_{\blacktriangle}^{\tilingsys}$ sketched above, we conclude $\interI \models \queryq_{\blacktriangle}^{\tilingsys}$. A contradiction.
\end{proof}

The second auxiliary lemma has a proof analogous to the proof Lemma~\ref{lemma:D-covers-iff-Tblacktriangle-does-not-entail}.

\begin{lem}\label{lemma:encoding-octant-revisited-II}
  Let $\tilingsys \deff (\tilesCol, \tiles, \tileswhite)$ be a domino tilling system and let $\interI$ be a $\tiles$-hyperoctant.
  We have that $\tboxT_{\blacktriangle}^{\tilingsys} \not\models \queryq_{{\TwoVPQ}}^{\tilingsys}$ if and only if $\tilingsys$ covers the octant.
\end{lem}
\begin{proof}
  If $\tilingsys$ covers the octant, by Lemma~\ref{lemma:hyperoctant-works} there exists a proper $\tiles$-hyperoctant~$\interI$. 
  Applying the second part of Lemma~\ref{lemma:encoding-octant-revisited}, we extend $\interI$ to a model of $\tboxT_{\blacktriangle}^{\tilingsys}$ violating $\interJ \not\models \queryq_{{\TwoVPQ}}^{\tilingsys}$.
  Thus $\tboxT_{\blacktriangle}^{\tilingsys} \not\models \queryq_{{\TwoVPQ}}^{\tilingsys}$.
  For the reverse direction, take any model~$\interI$ of $\tboxT_{\blacktriangle}^{\tilingsys}$ such that $\interI \not\models \queryq_{{\TwoVPQ}}^{\tilingsys}$. 
  By Corollary~\ref{cor:tree-model-property-for-ALCVPL} we can assume that $\interI$ is single-role tree-like.
  From Lemma~\ref{lemma:tree-models-of-Tblacktriangle-contain-hyperoctants} we know that $\interI$ contains a substructure~$\interJ$ that is a $\tiles$-hyperoctant.
  As~$\interI \not\models \queryq_{{\TwoVPQ}}^{\tilingsys}$, we know that $\interJ \not\models \queryq_{{\TwoVPQ}}^{\tilingsys}$, and thus, by Lemma~\ref{lemma:encoding-octant-revisited}, we know that $\interJ \not\models \queryq_{\blacktriangle}^{\tilingsys}$.
  Hence, by Lemma~\ref{lemma:encoding-octant-uno}, $\interJ$ is a proper $\tiles$-hyperoctant.
  Invoking Lemma~\ref{lemma:hyperoctant-works} we conclude that $\tilingsys$ covers the octant.
\end{proof}

Linking Lemma~\ref{lemma:encoding-octant-revisited-II} with the undecidability of the tiling problem (Lemma~\ref{lemma:white-bordered-octant-tiling-is-undec}), we~get:

\begin{thm}
  The entailment problem of {\TwoVPQ}s over $\ALC$-TBoxes is undecidable.
\end{thm}

Once again, let us point out that by the semantics of the query and the logics observe that $\tboxT_{\blacktriangle}^{\tilingsys} \not\models \queryq_{{\TwoVPQ}}^{\tilingsys}$ holds if and only if $\tboxT_{\blacktriangle}^{\tilingsys} \cup \{ \top \dlsubseteq \neg \exists{\languageL_{\downarrow\rightarrow\uparrow}}.\top \}$ (which is written in $\ALCvpl\mathcal{I}$, namely the extension of $\ALCvpl$ with the inverse operator) is satisfiable. 
Such a reduction yields undecidability of the concept satisfiability problem for $\ALCvpl\mathcal{I}$, and thus  reproves the previously-established result of G\"oller~\cite[Prop.~2.32]{Goller2008} (Corollary~\ref{cor:Goller}~in~Preliminaries).  

We conclude the section by revisiting known results concerning query entailment in (extensions of) $\ALC$, and lifting them to the case of $\ALCvpl$.
This contrasts with Theorem~\ref{thm:undecidability-querying-with-VPLs}.
\begin{cor}\label{cor:regular-querying-of-ALCvpl}
The query entailment problem for the class of Positive Conjunctive Regular Path Queries over $\ALCvpl$-TBoxes is $\TwoExpTime$-complete.
\end{cor}
\begin{proof}[Proof sketch.]
  Note that all the results given above transfer immediately to the case of TBoxes, as they can be internalised in concepts in the presence of regular expressions~\cite[p.~186]{2003handbook}. Hence, it suffices to focus on $\ALCvpl$-concepts only. 
  Let $\conceptC$ be an input $\ALCvpl$-concept and $\queryq$ be a $\REG$-\PEQ.\@
  L\"oding et al. introduced~\cite[p.~55]{LodingLS07} a model of visibly pushdown tree automata that has the following properties (unfortunately all of them are only implicit in the paper):  
  (a) generalises nondeterministic tree automata, 
  (b) is closed under intersection (and an automaton recognizing the intersection of languages can be computed in polynomial time), 
  and (c) its non-emptiness can be tested in exponential time~\cite[Thm.~4]{LodingLS07}.
  They also provided~\cite[Sec.~4.2]{LodingLS07} an automaton $\automatonA_{\conceptC}$ that accepts precisely (suitably) single-role tree-like models of $\conceptC$, and the size of $\automatonA_{\conceptC}$ is exponential w.r.t. the size of the concept $\conceptC$~\cite[Lemma~17]{LodingLS07}.
  On the other hand, Guti{\'{e}}rrez{-}Basulto et al. provide~\cite[Lemma~8]{GBasulto23} a non-deterministic tree automaton $\automatonA_{\neg\queryq}$, of size exponential w.r.t. the sizes of $\conceptC$ and~$\queryq$, that accepts all single-role tree-like structures that do not contain any matches of~$\queryq$.\footnote{In the paper of Guti{\'{e}}rrez{-}Basulto et al., the query automaton is presented in a more general setting of tree decomposition as it was designed to also work for the case of knowledge-bases (not just TBoxes). In our paper we deal with concepts only, so such an automaton works on trees as usual.}
  It follows then that the intersection of the languages of $\automatonA_{\conceptC}$ and $\automatonA_{\neg\queryq}$ is non-empty if and only if $\conceptC \not\models \queryq$.
  As visibly pushdown tree automata are closed under intersection~\cite[p.~55]{LodingLS07}, and their non-emptiness can be solved in exponential time~\cite[Thm. 4]{LodingLS07}, we infer that the non-emptiness of $\automatonA_{\conceptC} \cap \automatonA_{\neg\queryq}$ can be tested in doubly-exponential time w.r.t. the sizes of $\conceptC$ and $\queryq$.
  The matching lower bound is inherited from the concept satisfiability.
\end{proof}

%% file: sections/conclusions.tex

\section{Conclusions}\label{sec:conclusions}

We investigated the decidability status of extensions of $\ALCvpl$ (also known as Propositional Dynamic Logic with Visibly Pushdown Programs) with popular features supported by W3C ontology languages.
While undecidability of $\ALCvpl$ with inverses or role-hierarchies follows from existing work (see the end of Section~\ref{sec:preliminaries}),
we provided undecidability proofs of $\ALCvpl$ extended with the $\Self$ operator (Section~\ref{sec:self}), with nominals (Section~\ref{sec:nominals}), or non-regular queries (Section~\ref{sec:querying-negative}).
The first proof relied on the reduction from non-emptiness of the intersection of deterministic one-counter languages. The other two proofs relied on reduction from (variants~of) the tiling problem. 
We~conclude the paper with a list of open problems. 
\begin{itemize}\itemsep0em
\item Our undecidability proof for $\ALCSelfvpl$ relied on the availability of multiple visibly-pushdown languages that are encodings on deterministic one-counter languages. Can our undecidability proof be sharpened?
For instance, is the concept satisfiability of $\ALCregrhashshash$ with $\Self$ already undecidable?
\item Positive results for $\ALCvpl$~\cite[Thm.~18]{LodingLS07} concern the concept satisfiability problem, rather than the knowledge-base satisfiability problem. 
Is the later decidable for $\ALCvpl$? Classical techniques~\cite[p.~210]{GiacomoL94} for incorporating ABoxes inside concepts do not work, as the class of visibly-pushdown languages is not compositional (is of ``infinite memory'').
Note that this problem already occurs for $\ALCregrhashshash$.
\item Is an extension of $\ALCvpl$ (or even $\ALCregrhashshash$) with functionality decidable? 
De Giacomo and Lenzerini~\cite[p.~211]{GiacomoL94} proposed a satisfiability-preserving translation from $\ALCreg$ with functionality to plain $\ALCreg$.
Unfortunately, this reduction does not seem to be applicable to $\ALCvpl$. The reason is again that visibly-pushdown languages are not compositional.
What is more, functionality violates a crucial condition of ``unique diamond-path property'' from the decidability proof of $\ALCvpl$~\cite[Def.~11]{LodingLS07}.
\item Existing positive results on non-regular extensions of $\ALCreg$, especially these of L{\"{o}}ding et al.~\cite[Thm.~18]{LodingLS07}, rely on the use of (potentially infinite) tree-like models. Is the \emph{finite} satisfiability problem for $\ALCvpl$ decidable? Already the case of $\ALCregrhashshash$ is open.
\end{itemize}